\documentclass[11pt, letterpaper]{article}

\usepackage{comment}

\usepackage[utf8]{inputenc}
\usepackage[english]{babel}
\usepackage[margin=1in]{geometry}

\usepackage{amsthm, amsfonts}
\usepackage{mathtools, thmtools}
\usepackage{bbm}
\usepackage{thm-restate}
\newtheorem{theorem}{Theorem}
\newtheorem{lemma}[theorem]{Lemma}
\newtheorem{corollary}[theorem]{Corollary}
\theoremstyle{definition}
\newtheorem{definition}{Definition}
\newtheorem{invariant}{Invariant}
\usepackage[normalem]{ulem}

\usepackage{algorithm}
\usepackage{algorithmic}

\usepackage[numbers, sort]{natbib}

\usepackage[dvipsnames,usenames]{xcolor}
\usepackage[colorlinks=true,pdfpagemode=UseNone,urlcolor=RoyalBlue,linkcolor=RoyalBlue,citecolor=OliveGreen,pdfstartview=FitH]{hyperref}
\usepackage{prettyref}

\usepackage{graphicx}
\usepackage{subcaption}

\usepackage{tikz}
\usetikzlibrary{shapes, patterns, decorations, fit, intersections}

\usepackage{multicol}

\title{AdWords in a Panorama}

\author{
    Zhiyi Huang\thanks{The University of Hong Kong. Email: zhiyi@cs.hku.hk.}
    \and
    Qiankun Zhang\thanks{The University of Hong Kong. Email: qkzhang@cs.hku.hk.}
    \and
    Yuhao Zhang\thanks{The University of Hong Kong. Email: yhzhang2@cs.hku.hk.}
}

\date{August 2020}

\begin{document}

\begin{titlepage}
    \thispagestyle{empty}
    \maketitle
    \begin{abstract}
        \thispagestyle{empty}
        Three decades ago, Karp, Vazirani, and Vazirani (STOC 1990) defined the online matching problem and gave an optimal $1-\frac{1}{e} \approx 0.632$-competitive algorithm.
Fifteen years later, Mehta, Saberi, Vazirani, and Vazirani (FOCS 2005) introduced the first generalization called \emph{AdWords} driven by online advertising and obtained the optimal $1-\frac{1}{e}$ competitive ratio in the special case of \emph{small bids}.
It has been open ever since whether there is an algorithm for \emph{general bids} better than the $0.5$-competitive greedy algorithm.
This paper presents a $0.5016$-competitive algorithm for AdWords, answering this open question on the positive end.
The algorithm builds on several ingredients, including a combination of the online primal dual framework and the configuration linear program of matching problems recently explored by Huang and Zhang (STOC 2020), a novel formulation of AdWords which we call the panorama view, and a generalization of the online correlated selection by Fahrbach, Huang, Tao, and Zadimorghaddam (FOCS 2020) which we call the panoramic online correlated selection.

    \end{abstract}
\end{titlepage}

\newcommand{\E}{\mathbf{E}}
\newcommand{\Var}{\mathbf{Var}}
\renewcommand{\Pr}{\mathbf{Pr}}

\newcommand{\R}{\mathbf{R}}
\newcommand{\status}{\mathcal{S}}

\newcommand{\kmax}{k_{\max}}
\newcommand{\kmin}{k_a^{\min}}

\newcommand{\ALGIND}{\hspace{\algorithmicindent}}
\newcommand{\INDSTATE}{\STATE\ALGIND}
\newcommand{\INDINDSTATE}{\STATE\ALGIND\ALGIND}

\newcommand{\TBD}{\textcolor{red}{TBD.}}
\newcommand{\qiankun}[1]{\textcolor{red}{(Qiankun: #1)}}
\newcommand{\yuhao}[1]{\textcolor{blue}{(Yuhao: #1)}}

\newcommand{\defeq}{\stackrel{\textnormal{def}}{=}}

\section{Introduction}
\label{sec:intro}

Consider an \emph{ad platform} in online advertising, e.g., a search engine in the case of sponsored search.
Each advertiser on the platform provides its \emph{bids} for a set of keywords, for which it likes its ad to be shown.
It further has a \emph{budget} which upper bounds its payment in a day.
When a user submits a request, often referred to as an \emph{impression}, the platform sees the bids of the advertisers for it.
The platform then selects an advertiser, who pays either its bid or its remaining budget, whichever is smaller.
The goal of the platform is to allocate impressions to advertisers to maximize the total payment of the advertisers.
The revenues of online advertising in the US have surpassed those of television advertising in 2016~\cite{IAB/report/2016}, and have totaled \$57.9 billions in the first half of 2019~\cite{IAB/report/2019}.

While online advertising acquires growing importance in practice, it has also been extensively studied in theoretical computer science.
The first related research dates back to three decades ago, when \citet{KarpVV/STOC/1990} introduced the online matching problem and designed an online algorithm with the optimal $1-\frac{1}{e} \approx 0.632$ competitive ratio.
It can be viewed as the special case with unit bids and unit budgets.
Fifteen years later, \citet{MehtaSVV/JACM/2007} formally formulated it as the \emph{AdWords} problem.
They introduced an optimal $1-\frac{1}{e}$-competitive algorithm under the \emph{small-bid assumption}:
an advertiser's bid for any impression is much smaller than its budget.

Subsequently, AdWords has been studied under stochastic assumptions.
\citet{GoelM/SODA/2008} showed that \emph{assuming a random arrival order of the impressions} and small bids, a $1-\frac{1}{e}$ competitive ratio can be achieved using the greedy algorithm:
allocate each impression to the advertiser who would make the largest payment.
Later, the algorithm proposed by \citet{DevernurH/EC/2009} achieved the near-optimal competitive ratio of $1-\epsilon$ under the same random-arrival and small-bid assumptions.
\citet{MirrokniOZ/SODA/2012} analyzed the algorithm of \citet{MehtaSVV/JACM/2007} in the more restricted \emph{unknown iid model}, and obtained an improved ratio of $0.76$ for small bids.
Finally, \citet{DevanurJSW/EC/2011} proved that the greedy algorithm is $1-\frac{1}{e}$-competitive for general bids in the unknown iid model.
For small bids, they proposed a $(1 - \epsilon)$-competitive algorithm.
We refer readers to the survey by \citet{Mehta/FTTCS/2013} for further references.

Little is known, however, about the most general case of AdWords, i.e., with general bids and without stochastic assumptions.
On the positive end, we only have the greedy algorithm and the trivial $0.5$ competitive ratio.
On the negative end, there is no provable evidence that the optimal $1-\frac{1}{e}$ competitive ratio of online matching cannot be achieved in AdWords.
It has been open since \citet{MehtaSVV/JACM/2007} whether there is an online algorithm that achieves a competitive ratio strictly better than $0.5$.

\subsection{Our Contributions and Techniques}

The main result of the paper is the first online algorithm for AdWords that breaks the $0.5$ barrier.

\begin{theorem}
	\label{thm:main}
	There is a $0.5016$-competitive algorithm for AdWords.
\end{theorem}

We develop the algorithm under the online primal dual framework.
In a nutshell, by considering an appropriate linear program (LP) of the problem, the online primal dual framework designs the online algorithm according to the optimality conditions of LPs, and uses the objective of the dual LP as the benchmark in the analysis.
\citet{BuchbinderJN/ESA/2007} applied it to AdWords with small bids, using the standard matching LP, to obtain an alternative analysis of the $1-\frac{1}{e}$ competitive algorithm by \citet{MehtaSVV/JACM/2007}.
Later, \citet{DevanurJ/STOC/2012} and \citet{DevanurJK/SODA/2013} found further applications of the framework in other online matching problems.
Recently, \citet{HuangZ/STOC/2020} demonstrated an advantage of using the configuration LP instead of the standard matching LP in online matching with stochastic rewards.
The current paper also builds on the strength of the configuration LP, echoing the message of \citet{HuangZ/STOC/2020}.
See Section~\ref{sec:prelim} for details.

Our second ingredient is a novel formulation of AdWords which we call the \emph{panorama view}.
Recall that an advertiser's payment in the original formulation is either the sum of its bids for the assigned impressions or its budget, whichever is smaller.
The panorama view further associates each advertiser with an interval whose length equals the budget, and requires the algorithms to assign each impression to not only an advertiser, but further a subset of its interval with size equal to the bid.
For example, consider an impression $i$ and an advertiser $a$ whose budget is $2$ and whose bid for $i$ is $1$.
The panorama view associates advertiser $a$ with an interval $[0, 2)$.
Further, when an algorithm assigns $i$ to $a$, it must further assign $i$ to a subset of size at most $1$, e.g., $[0.5, 1.5)$.
Define an advertiser's payment in the panorama view to be the size of the union of the assigned subsets, which lower bounds the payment in the original formulation.
The panorama view allows a fine-grained characterization on how the assignment of an impression $i$ to an advertiser $a$ affects the marginal gains of the other impressions assigned to $a$.
Concretely, suppose we shortlist two advertisers for each impression, and then assign it to one of them with a fresh random bit.
In the original formulation, having advertiser $a$ in impression $i$'s shortlist decreases the marginal gain of \emph{all} other impressions that shortlist $a$ in a \emph{complicated manner}.
In the panorama view, however, it decreases the marginal gain \emph{only for those whose assigned subsets intersect with $i$'s};
more precisely, it \emph{decreases the contribution of the intersection by half}.
See Section~\ref{sec:panorama} for a formal definition of the panorama view and some examples.

Finally, instead of using a fresh random bit to select a shortlisted advertiser for each impression, our algorithm selects one with negative correlation.
If a previous impression which shortlists advertiser $a$ with an overlapping subset does \emph{not} select $a$, the current one will be more likely to select $a$.
Given the same shortlists, negatively correlated selections get larger expected gains in the panorama view than independent selections.
An algorithmic ingredient called \emph{online correlated selection} (OCS) by Huang and Tao~\cite{HuangT/arXiv/2019, Huang/arXiv/2019} provides a quantitative control of such negative correlation in the special case when bids equal budgets.
The final piece of our algorithm is a generalization of OCS which applies to the general case of AdWords in the panorama view.
We refer to it as the panoramic OCS (PanOCS).
Section~\ref{sec:prelim} includes a formal definition of OCS, Section~\ref{sec:panocs-glimpse} defines the PanOCS and sketches the main ideas behind it, and Section~\ref{sec:panocs} provides the details.

Building on these ingredients, we get a $0.50005$-competitive online primal dual algorithm for AdWords in Section~\ref{sec:basic-algorithm}, weaker than the ratio in Theorem~\ref{thm:main} yet breaking the $0.5$ barrier nonetheless.
To obtain the final ratio, we observe that the above algorithm works better for larger bids while the algorithm of \citet{MehtaSVV/JACM/2007} is better for smaller bids.
Hence, we design a $0.5016$-competitive hybrid algorithm in Section~\ref{sec:hybrid} by unifying both approaches under the online primal dual framework.
Appendix~\ref{app:small-bid} analyzes the algorithm of \citet{MehtaSVV/JACM/2007} for small bids using online primal dual and configuration LP, which may serve as a warmup for readers unfamiliar with the framework.

\subsection{Other Related Works}

AdWords is closely related to the literature of online matching started by \citet{KarpVV/STOC/1990}.
\citet{AggarwalGKM/SODA/2011} studied the vertex-weighted problem and obtained the optimal $1-\frac{1}{e}$ competitive ratio with a generalization of the algorithm by \citet{KarpVV/STOC/1990}.
\citet{FeldmanKMMP/WINE/2009} investigated edge-weighted online matching in the free-disposal model, where the algorithm may dispose a previous matched edge for free to make room for a new one.
They called it the \emph{display ads} problem, and achieved the optimal $1-\frac{1}{e}$ competitive ratio assuming large capacities, i.e., each offline vertex can be matched to a large number of online vertices.
The analysis was simplified by \citet{DevanurHKMY/TEAC/2016} under the online primal dual framework.
Further, Fehrbach et al.~\cite{FahrbachHTZ/FOCS/2020, FahrbachZ/arXiv/2017, HuangT/arXiv/2019, Huang/arXiv/2019} obtained a better than $0.5$-competitive edge-weighted algorithm without assuming large capacities.
In doing so, they introduced the OCS which directly inspired this paper.
Finally, there are generalized models which allow all vertices to be online and even consider general graphs~\cite{WangW/ICALP/2015, HuangKTWZZ/STOC/2018, HuangPTTWZ/SODA/2019, AshlagiBDJSS/EC/2019, GamlathKMSW/FOCS/2019, HuangTWZ/FOCS/2020, HuangNTWZZ/JACM/2020}.

Online matching problems are also widely investigated under different stochastic assumptions.
First, consider random arrivals of online vertices.
\citet{KarandeMT/STOC/2011} and \citet{MahdianY/STOC/2011} showed that the algorithm of \citet{KarpVV/STOC/1990} is strictly better than $1-\frac{1}{e}$-competitive in this model.
\citet{HuangTWZ/ICALP/2018} gave a better than $1-\frac{1}{e}$-competitive algorithm for the vertex-weighted problem.
\citet{KesselheimRTV/ESA/2013} showed that the greedy algorithm is $\frac{1}{e}$-competitive for the edge-weighted problem even without free-disposal.
Under the stronger assumption that online vertices are drawn iid from an unknown distribution, \citet{KapralovPV/SODA/2013} proved that greedy is $1-\frac{1}{e}$-competitive for a more general problem called online submodular welfare maximization which captures both the edge-weighted problem with free-disposal and AdWords as special cases.
Further assuming that the distribution is known leads to better competitive ratios~\cite{FeldmanMM/FOCS/2009, ManshadiOS/MOR/2012, HaeuplerMZ/WINE/2011, JailletL/MOR/2014}.
We leave for future research if the algorithm in this paper is better than $1-\frac{1}{e}$-competitive under random arrivals.

Finally, \citet{MehtaP/FOCS/2012} proposed online matching with stochastic rewards, where an edge chosen by the algorithm is successfully matched only with some probability.
They focused on the special case of equal success probabilities and gave algorithms that are $0.567$-competitive if the success probability is vanishing, and better than $0.5$-competitive in general.
Later, \citet{MehtaWZ/SODA/2014} showed a $0.534$-competitive algorithm for vanishing unequal success probabilities.
Recently, \citet{HuangZ/STOC/2020} improved the competitive ratios to $0.576$ and $0.572$ for vanishing equal and unequal success probabilities respectively.
In doing so, they showed an advantage of the configuration LP over the standard matching LP under online primal dual.
This paper echoes the above message.

\section{Preliminaries}
\label{sec:prelim}

Consider a bipartite graph $G = (A, I, E)$, where $A$ and $I$ are sets of vertices corresponding to the advertisers and impressions in AdWords respectively, and $E \subseteq A \times I$ is the set of edges between them.
Further, each edge $(a, i)$ is associated with a non-negative real number $b_{ai}$ which represents advertiser $a$'s bid for impression $i$.%
\footnote{AdWords as an online algorithm problem does not consider the strategic behaviors of the advertisers. 
We merely inherit the term \emph{bid} from the original paper of \citet{MehtaSVV/JACM/2007}.}
By allowing zero bids, we may assume without loss of generality (wlog) that $G$ is a complete bipartite graph, i.e., $E = A \times I$.
Finally, each advertiser $a$ is associated with a positive budget $B_a$ which upper bounds the payment of the advertiser.
Concretely, assigning a subset of impressions $S \subseteq I$ to an advertiser $a$ leads to a budget-additive payment:
\[
    b_a(S) \defeq \min \bigg\{ \sum_{i \in S} b_{ai}, B_a \bigg\}
    ~.
\]
By this definition, we may assume wlog that $b_{ai} \le B_a$ for any advertiser $a$ and any impression $i$.

The advertisers are given upfront, while the impressions arrive one at a time.
We write $i < i'$ if an impression $i$ arrives before another impression $i'$.
On the arrival of an impression, the algorithm must immediately and irrevocably assign it to an advertiser.
The objective is to maximize the sum of the above payments from all advertisers.
Following the standard competitive analysis of online algorithms, an algorithm is \emph{$\Gamma$-competitive} for some \emph{competitive ratio} $0 \le \Gamma \le 1$ if its expected objective is at least $\Gamma$ times the offline optimal in hindsight for any AdWords instance.

\paragraph{Configuration Linear Program.}
The algorithms in this paper and their analyses rely on the LP relaxations of the problem.
Instead of the standard matching LP, this paper considers the more expressive advertiser-side configuration LP and its dual:\\[2ex]
\begin{minipage}{.5\textwidth}
\[
    \begin{aligned}
        \textrm{max} \quad
        & 
        \sum_{a \in A} \sum_{S \subseteq I} b_a(S) x_{aS} \\
        \textrm{s.t.} \quad
        &
        \sum_{S \subseteq I} x_{aS} \le 1 && \forall a \in A \\
        &
        \sum_{a \in A} \sum_{S \ni i} x_{aS} \le 1 && \forall i \in I \\[1ex]
        &
        x_{aS} \ge 0 && \forall a \in A, \forall S \subseteq I
    \end{aligned}
\]
\end{minipage}
\begin{minipage}{.5\textwidth}
\[
    \begin{aligned}
        \textrm{min} \quad
        &
        \sum_{a \in A} \alpha_a + \sum_{i \in I} \beta_i \\
        \textrm{s.t.} \quad
        &
        \alpha_a + \sum_{i \in S} \beta_i \ge b_a(S) && \forall a \in A, \forall S \subseteq I \\[1ex]
        & 
        \alpha_a \ge 0 && \forall a \in A \\[3.5ex]
        & 
        \beta_i \ge 0 && \forall i \in I
    \end{aligned}
\]
\end{minipage}\\[3ex]

Let $P$ and $D$ denote the objectives of the primal and dual LPs respectively.
Throughout the paper we will always let $x_{aS}$ be the probability that $S$ is the subset of impressions assigned to advertiser $a$.
Then, the primal objective $P$ equals the objective of the algorithm.

\paragraph{Online Primal Dual Framework.}
We build on the online primal dual framework which uses the dual objective as an upper bound of the offline optimal in the competitive analyses of online algorithms.
In particular, this paper applies it to the configuration LP of AdWords.

\begin{lemma}
    \label{lem:online-primal-dual}
    Suppose an online algorithm is coupled with a dual algorithm which maintains a dual assignment such that for some $0 \le \Gamma \le 1$:
    \begin{enumerate}
        \item Approximate dual feasibility: $\alpha_a + \sum_{i \in S} \beta_i \ge \Gamma \cdot b_a(S)$ for any $a \in A$ and any $S \subseteq I$.
        \item Reverse weak duality: $P \ge D$;
    \end{enumerate}
    Then, it is $\Gamma$-competitive.
\end{lemma}

\begin{proof}
    By the first condition, scaling the dual assignment by a factor of $\Gamma^{-1}$ makes it feasible while changing the dual objective by the same factor.
    Therefore, by weak duality of LPs, the offline optimal is at most $\Gamma^{-1} D$.
    Putting together with the second condition proves the lemma.
\end{proof}

\paragraph{Online Correlated Selection.}
The algorithms in this paper further utilize a recent algorithmic ingredient called \emph{online correlated selection} (OCS) by Huang and Tao~\cite{Huang/arXiv/2019, HuangT/arXiv/2019}.
Consider a set of ground elements, and further a sequence of pairs of these elements arriving one at a time.
Suppose we randomly select one element from each pair with a fresh random bit.
Then, an element will be selected at least once with probability $1 - 2^{-k}$ after appearing in $k$ pairs.
The OCS correlates the randomness to achieve better efficiency.
We state below a simplified definition, removing some aspects irrelevant to AdWords.

\begin{definition}
    For any $0 \le \gamma \le 1$, a $\gamma$-OCS is an online algorithm ensuring that for any element which appears in $k$ pairs, it is selected at least once with probability at least:
    \[
        1 - 2^{-k} (1 - \gamma)^{\max \{ k-1, 0 \}}
        ~.
    \]
\end{definition}

\section{Panorama View}
\label{sec:panorama}

The algorithms in this paper are based on a novel viewpoint of the AdWords problem which we call the \emph{panorama view}.
Recall that the payment of an advertiser $a$ is budget-additive in AdWords:
assigning a subset of impressions $S$ to an advertiser $a$ gives $b_a(S) = \min \big\{ \sum_{i \in S} b_{ai}, B_a \big\}$.
Let $\mu(\cdot)$ denote the Lebesgue measure.
In the panorama view, we further associate each advertiser $a$ with an interval $[0, B_a)$;
each impression $i$ assigned to $a$ is further assigned to a subset $Y_{ai} \subseteq [0, B_a)$ whose Lebesgue measure $\mu(Y_{ai})$ is at most $b_{ai}$.
In fact, we will always choose $Y_{ai}$ to be a finite union of disjoint left-closed, right-open intervals, for which the Lebesgue measure is simply the sum of their lengths.
Further define the payment of an advertiser $a$ in the panorama view as:
\[
    \mu \big( \cup_{i \in S} Y_{ai} \big)
    ~.
\]

Correspondingly, the objective in the panorama view is the sum of the above payment from all advertisers.
Importantly, it lower bounds the original objective of AdWords.

\begin{lemma}
    \label{lem:panorama-vs-original-objective}
    For any advertiser $a$, any subset of impressions $S$ assigned to $a$, and any subsets $Y_{ai} \subseteq [0, B_a)$ with Lebesgue measure at most $b_{ai}$ for impressions $i \in S$, we have:
    \[
        \mu \big( \cup_{i \in S} Y_{ai} \big) \le \min \bigg\{ \sum_{i \in S} b_{ai}, B_a \bigg\}
        ~.
    \]
\end{lemma}

\begin{proof}
    On the one hand, by subadditivity of the Lebesgue measure function $\mu$, and by the measure upper bounds of the subsets $Y_{ai}$'s, we have $\mu \big( \cup_{i \in S} Y_{ai} \big) \le \sum_{i \in S} \mu \big( Y_{ai} \big) \le \sum_{i \in S} b_{ai}$.
    On the other hand, because $Y_{ai}$'s are subsets of $[0, B_a)$, we have $\mu \big( \cup_{i \in S} Y_{ai} \big) \le \mu\big([0, B_a)\big) = B_a$.
\end{proof}

\paragraph{Example 1 (Deterministic Algorithms).}
Consider an arbitrary deterministic algorithm.
Then, whenever it assigns an impression $i$ to an advertiser $a$, we may wlog further assign it to the leftmost unassigned interval.
For instance, suppose impressions $1, 2, 3$, and so on are assigned to advertiser $a$ in this order;
we may further assign $1$ to $[0, b_{a1})$, $2$ to $[b_{a1}, b_{a1} + b_{a2})$, $3$ to $[b_{a1} + b_{a2}, b_{a1} + b_{a2} + b_{a3})$, and so forth.
In doing so, the objectives in the panorama view is identical to the original one.

\paragraph{Oblivious Semi-randomized Algorithms.}
The panorama view of AdWords separates itself from the original one when it comes to a special family of randomized algorithms which we call the \emph{oblivious semi-randomized algorithms}.
They are semi-randomized in that for every impressions $i$, they either assign it deterministically to an advertiser-subset combination, or choose two advertiser-subset combinations and assign it to one of them with equal marginal probability.
We shall refer to the former as a \emph{deterministic round} and the latter as a \emph{randomized round}.
If an impression $i$ corresponds to a randomized round, we say that it is \emph{semi-assigned} to the advertiser-subset combinations.
For the time being, readers may think of using a fresh random bit in every randomized round for a concrete understanding of the panorama view, although our algorithms will correlate the decisions in different rounds negatively.

Further, these algorithms are oblivious:
neither the decisions of deterministic versus randomized rounds, nor the choices of advertiser-subset combinations depend on the realization of random bits in previous rounds.
Hence, the semi-assignments to the same advertiser may have overlapping subsets and thus, the objective in the panorama view no longer equals the original one in general.

\paragraph{Example 2 (Oblivious Semi-randomized Algorithms).}
Let there be two advertisers whose budgets equal $2$, and three impressions for which both advertisers bid $1$.
Further suppose that we select with a fresh random bit for each impression.
In the original budget-additive payments, with probability $\frac{1}{4}$ all impressions are assigned to the same advertiser and thus the objective equals $2$;
otherwise, the objective equals $3$.
Hence, the expected objective equals $\frac{1}{4} \cdot 2 + \frac{3}{4} \cdot 3 = \frac{11}{4}$.
In the panorama view, however, the algorithm must further assign each impression to a subset of $[0, 2)$ for both advertisers.
It is wlog to assign the first impression to $[0, 1)$, contributing $1$ to the objective.
Further, it is reasonable to further assign the second impression to $[1, 2)$ so that it is disjoint with the first one and contributes $1$ to the objective.
However, the third impression only contributes $\frac{1}{2}$ to the objective regardless of the choices of subsets because the entire interval $[0, 2)$ has been semi-assigned once. 
Therefore, the expected objective is only $1 + 1 + \frac{1}{2} = \frac{5}{2}$.

\bigskip

It may seem odd to restrict ourselves to oblivious algorithms.
Would it not be better if we first check the realized assignments in earlier rounds and then pick an advertiser-subset combination disjoint with the previous ones?
In a nutshell, we focus on oblivious algorithms to separate the algorithmic component for choosing assignments and semi-assignments, and that for correlating the decisions in different randomized rounds.
Importantly, we can achieve negative correlation:
a semi-assignment is more likely to get selected if an earlier overlapping one is not.
By contrast, screening the options based on the realization of earlier random bits could lead to positive correlations (e.g., weighted sampling without replacements \cite{Alexander/AnnStat/1989}).
That said, there may be AdWords algorithms with controlled positive correlations which are better than $0.5$-competitive.
The study of such algorithms, however, is beyond the scope of this paper and is left for future research.

\paragraph{Bookkeeping at the Point-level.}
%
It is more convenient to account for the primal objective at the point-level as follows.
We say that a point $y \in [0, B_a)$ of an advertiser $a$ is assigned if there is an impression $i$ assigned to $a$ and a subset $Y_{ai}$ containing $y$, either due to a deterministic round, or due to a semi-assignment in a randomized round which selected $a$.
For randomized algorithms, let $0 \le x_a(y) \le 1$ denote the probability that $y$ is assigned.
Then, the primal objective equals: 
\begin{equation}
    \label{eqn:panorama-primal}
    P = \sum_{a \in A} \int_0^{B_a} x_a(y) dy
    ~.
\end{equation}

Similarly, we say that $y$ is semi-assigned whenever an impression is semi-assigned to $a$ and a subset containing $y$.
Let $k_a(y)$ denote the number of times that $y$ is semi-assigned.
Further define $k_a(y) = \infty$ if $y$ has been assigned in a deterministic round, driven by the fact that semi-assignments on their own take finitely many rounds to make a point $y$ assigned with certainty.

We further introduce point-level dual variables $\alpha_a(y)$, $a \in A$, $y \in [0 ,B_a)$, and let:
\begin{equation}
    \label{eqn:alpha-point-level}
    \alpha_a = \int_0^{B_a} \alpha_a(y) dy
    ~.
\end{equation}

Then, approximate dual feasibility becomes:
\begin{equation}
    \label{eqn:dual-feasibility-panorama}
    \int_0^{B_a} \alpha_a(y) dy + \sum_{i \in S} \beta_i \ge \Gamma \cdot b_a(S)
\end{equation}

\paragraph{Panoramic Interval-level Assignments.}
We first introduce some notations which are useful throughout the paper.
For any point $y \in [0, B_a)$, any subset $Y \subseteq [0, B_a)$, and any $0 \le b \le B_a$, let $y \oplus_Y b$ denote the point in $[0, B_a)$ such that the interval $[y, y \oplus_Y b)$ excluding $Y$ has Lebesgue measure $b$.
Here, we abuse notation and allow $y \oplus_Y b$ to be smaller than $y$, in which case $[y, y \oplus_Y b)$ denotes the union of $[y, B_a)$ and $[0, y \oplus_Y b)$.
In the boundary case when the subset $[0, B_a) \setminus Y$ has a measure strictly less than $b$, i.e., $B_a - \mu(Y) < b$, define $y \oplus_Y b = y$.
Further define the reverse operation $y \ominus_Y b$ such that $[y \ominus_Y b, y) \setminus Y$ has measure $b$.

For any advertiser $a$ and the set of impressions assigned and semi-assigned to it, the algorithm will further select subsets of $[0, B_a)$ greedily as follows.
Maintain a point $y^*$ initially at $0$ which represents the start of the next subset.
For each impression $i$, further assign or semi-assign it to $[y^*, y^* \oplus_{Y_D} b_{ai}) \setminus Y_D$, where $Y_D$ is the subset of $[0, B_a)$ that has already been assigned deterministically.
Further update $y^* = y^* \oplus_{Y_D} b_{ai}$.
To this end, think of $[0, B_a)$ as a circle by gluing its endpoints;
the algorithm scans along the circle to find a subset with measure $b_{ai}$ that has not been deterministically assigned.
It is similar to taking a panorama and hence the name of the alternative view of AdWords.

The panoramic interval-level assignments equalize the numbers of times the points $y \in [0, B_a)$ are semi-assigned, among those that have not been deterministically assigned.
We omit the proof since it follows by the definition of the algorithm.
See Figure~\ref{fig:interval-level} for an illustrative example.

\begin{lemma}
    \label{lem:K-property}
    For any $a \in A$ and any $y \in [0, B_a)$, $k_a(y)$ equals either
    (1) $k_{\min} = \min_{z \in [0, B_a)} k_a(z)$, or
    (2) $k_{\min} + 1$, or
    (3) $\infty$.
    Further, the first kind satisfies $y \ge y^*$, and the second kind satisfies $y < y^*$.
\end{lemma}

\begin{figure}
    \begin{subfigure}{.33\textwidth}
        \centering
        \includegraphics[width=\textwidth]{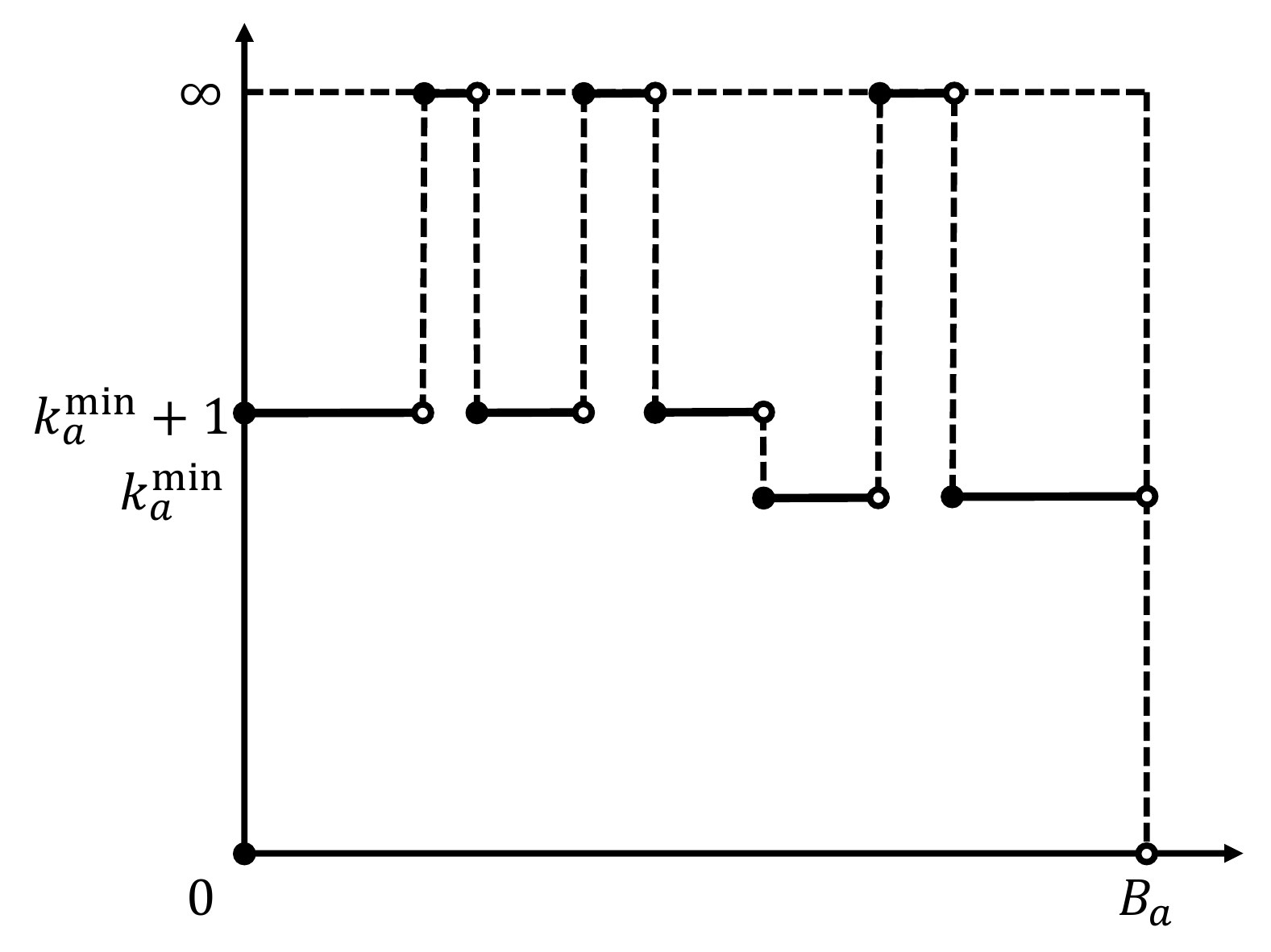}
        \caption{Interval-level assignments}
    \end{subfigure}
    \begin{subfigure}{.33\textwidth}
        \centering
        \includegraphics[width=\textwidth]{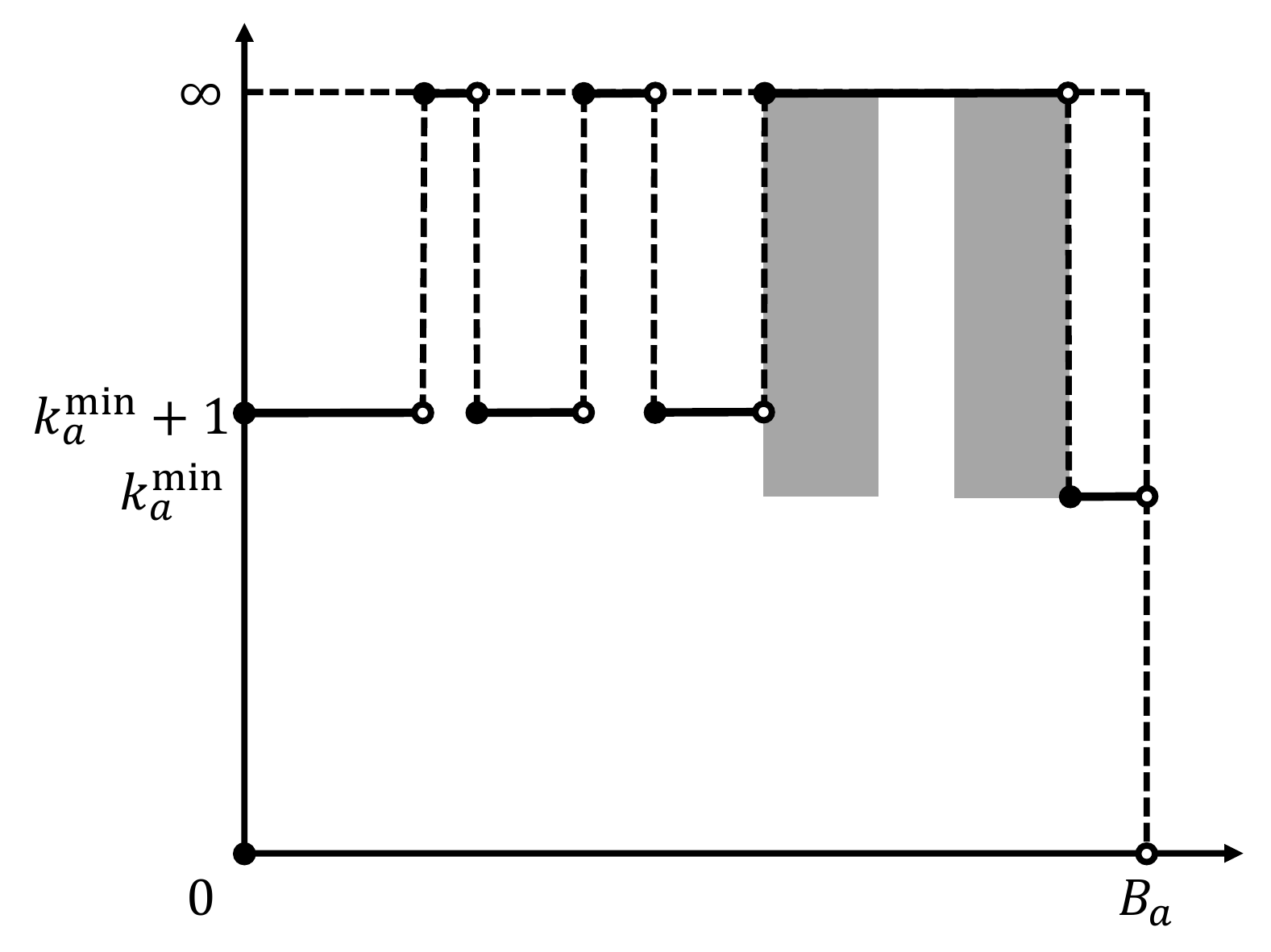}
        \caption{Update in a deterministic round}
    \end{subfigure}
    \begin{subfigure}{.33\textwidth}
        \centering
        \includegraphics[width=\textwidth]{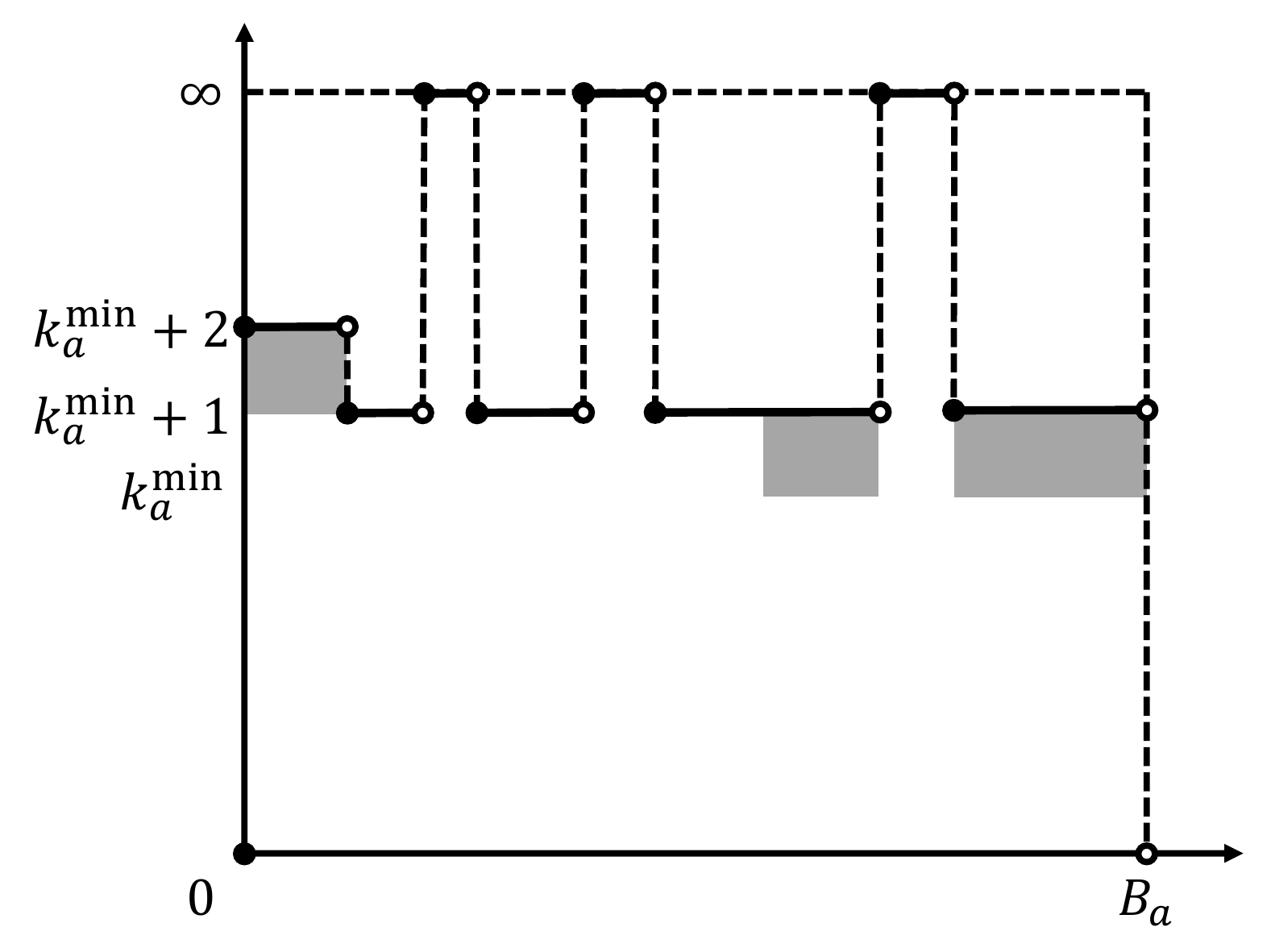}
        \caption{Update in a randomized round}
    \end{subfigure}
    \caption{Illustrative example of interval-level assignment represented by $k_a(y)$}
    \label{fig:interval-level}
\end{figure}

\section{Basic Algorithm}
\label{sec:basic-algorithm}

This section presents an oblivious semi-randomized algorithm that is better than $0.5$-competitive for AdWords.
Section~\ref{sec:panocs-glimpse} presents a brief introduction to an algorithmic ingredient called \emph{panoramic online correlated selection} (PanOCS), which correlates the randomized decisions in different rounds negatively.
Section~\ref{sec:pd-algorithm} then demonstrates an online algorithm powered by PanOCS, and Section~\ref{sec:analysis} analyzes it under the online primal dual framework.
Finally, Section~\ref{sec:basic-lp-solution} optimizes the parameters of the algorithm to achieve a $0.5041$ competitive ratio in the crux of AdWords, i.e., when all nonzero bids are large, $\frac{1}{2} B_a \le b_{ai} \le B_a$, and a smaller $0.50005$ competitive ratio in the general case.
The latter is smaller than the ratio in the main theorem but breaks the $0.5$ barrier nonetheless.

Formal descriptions of the PanOCS algorithms and their analyses may be of independent interest and are therefore deferred to a separate Section~\ref{sec:panocs}.
The $0.5016$ ratio in the main theorem requires a hybrid approach which treats large and small bids differently, which we present in Section~\ref{sec:hybrid}.

\subsection{Panoramic Online Correlated Selection at a Glimpse}
\label{sec:panocs-glimpse}

Recall the oblivious semi-randomized algorithms.
In each randomized round, such an algorithm chooses a pair of advertiser-subset combinations, oblivious to the random bits in previous rounds.
Then, the combinations are passed on to an algorithmic component which selects one of them with equal marginal probability, and \emph{correlates across different randomized rounds negatively}.
We call it the PanOCS since it is a generalization of the OCS~\cite{Huang/arXiv/2019, HuangT/arXiv/2019} in the panorama view of AdWords.

For a formal definition, recall that $x_a(y)$ is the probability a point $y \in [0, B_a)$ of an advertiser $a$ has been assigned, and $k_a(y)$ is the number of randomized rounds in which $y$ is semi-assigned.

\begin{restatable}{definition}{defpanocs}
    \label{def:panocs}
    A PanOCS is an online algorithm which takes a sequence of pairs of advertiser-subset combinations as input, and for each pair selects one combination.
    It is a $\gamma$-PanOCS for some $0 \le \gamma \le 1$ if for any advertiser $a$, and any point $y \in [0, B_a)$, we have:
    \begin{equation}
    \label{eqn:panocs-definition}
        x_a(y) \ge 1 - 2^{-k_a(y)} (1 - \gamma)^{ \text{max} \{k_a(y)-1, 0 \}}
        ~.
    \end{equation}
\end{restatable}

Observe that using an independent random bit in every randomized round is a $0$-PanOCS, since the probability of being assigned after $k$ semi-assignments with independent random bits is precisely $1 - 2^{-k}$.
The parameter $\gamma$ quantifies the advantage over independent random bits.

The intuition behind the inequality is best explained with a thought experiment.
Suppose that whenever $y$ is semi-assigned other than the first time, there is a $\gamma$ chance to be perfectly negatively correlated with the last semi-assignment of $y$:
$a$ is chosen this time if it is not chosen last time, and vice versa.
Further suppose that the above events are negatively dependent for the $k_a(y)-1$ different pairs of adjacent semi-assignments of $y$.
Then, $y$ is never assigned only if none of the events happens, whose probability is at most $(1-\gamma)^{k_a(y)-1}$, and further when none of the $k_a(y)$ independent selections picks $a$, which equals $2^{-k_a(y)}$.
Our analysis will substantiate this intuition.

To see the connection with OCS, consider a special case when the bids equal the budgets, i.e., $b_{ai} = B_a$.
Then, since $x_a(y)$ and $k_a(y)$ are independent of $y$, it suffices to consider if advertiser $a$ has been assigned.
As a result, the above definition coincides with the definition of OCS, taking the advertisers as the ground elements.
The extra challenge of PanOCS is to ensure the inequality simultaneously for all points $y \in [0, B_a)$ when the bids are arbitrary.

\begin{theorem}
    \label{thm:panocs-large-bid}
    Suppose all nonzero bids are large, i.e., $\frac{1}{2} B_a \le b_{ai} \le B_a$ or $b_{ai} = 0$ for any advertiser $a \in A$ and any impression $i \in I$.
    Then, there is a $0.05144$-PanOCS.
\end{theorem}

\begin{theorem}
    \label{thm:panocs-general-bid}
    Suppose the algorithm makes at most $k_{\max}$ semi-assignments to any point $y \in [0, B_a)$ of any advertiser $a \in A$.
    Then, there is a $0.01245 \cdot \kmax^{-1}$-PanOCS.
\end{theorem}

The proofs of the theorems are deferred to a separate Section~\ref{sec:panocs}.
We include below a proof sketch of a weaker $\frac{1}{64}$-PanOCS for large bids to foreshadow the arguments in our PanOCS analyses.

\begin{proof}[Proof Sketch of a Weaker Theorem~\ref{thm:panocs-large-bid} ($\gamma = \frac{1}{64}$)]
For any impression $i$ semi-assigned to an advertiser $a$, we write $(i+j)_a$ to denote the $j$-th impression semi-assigned to $a$ after impression $i$ (or the $(-j)$-th impression before $i$, if $j < 0$).

We next explain the algorithm.
Consider an impression $i$ in a randomized round.
Suppose it is semi-assigned to advertisers $a_1$ and $a_2$.
Then, sample $a^* \in \{a_1, a_2\}$ and $j \in \{-2, -1, 1, 2\}$ uniformly at random.
If $j > 0$, select $a_1$ or $a_2$ and the corresponding subsets with a fresh random bit;
further, pass the result to future impression $(i+j)_{a^*}$.
If $j < 0$, check if the past impression $(i+j)_{a^*}$ passes its result to $i$.
If so, makes the opposite choice:
select $a^*$ and the corresponding subset if $a^*$ was \emph{not} selected in round $(i+j)_{a^*}$, and vice versa.
Otherwise, select with a fresh random bit.

Finally, we argue this is a $\frac{1}{64}$-PanOCS.
Fix any advertiser $a$ and any point $y \in [0, B_a)$.
Suppose $i_1 < i_2 < \dots < i_k$ are the impressions semi-assigned to $a$ and subsets containing $y$.
Consider any neighboring $i_\ell$ and $i_{\ell+1}$.
We now use the assumption of large bids to get that \emph{$i_{\ell+1}$ is the first or second impression semi-assigned to advertiser $a$ after $i_\ell$}.
Hence, $i_\ell$ may pass its result to $i_{\ell+1}$, which may then make the opposite choice.
When it happens, $y$ is assigned exactly once in rounds $i_\ell$ and $i_{\ell+1}$.
More precisely, this is when $i_\ell$ samples $a^* = a$ and $j > 0$ so that $(i_\ell + j)_a = i_{\ell+1}$, and $i_{\ell+1}$ samples $a^* = a$ and $j < 0$ so that $(i_{\ell+1}+j)_a = i_\ell$, which happens with probability $\frac{1}{64}$.
Moreover, we claim that the events are negatively dependent for $k-1$ neighboring pairs of $i_\ell$ and $i_{\ell+1}$.
Hence, the probability of having no such pair is at most $(1-\frac{1}{64})^{k-1}$.
Finally, even if $i_1 < i_2 < \dots < i_k$ are independent, the probability that $a$ is never selected is only $2^{-k}$.
Putting together, $y$ has been assigned in at least one of these $k$ rounds with probability no less than $1 - 2^{-k} (1 - \frac{1}{64})^{k-1}$.
\end{proof}

\subsection{Online Primal Dual Algorithm}
\label{sec:pd-algorithm}

We demonstrate an online primal dual Algorithm~\ref{alg:pd-algorithm}, taking a $\gamma$-PanOCS as a blackbox.
Recall that an oblivious semi-randomized algorithm either deterministically assigns $i$ to an advertiser-subset combination, or semi-assigns it to two combinations in a randomized round.
In the latter case, let the $\gamma$-PanOCS select a combination.
Let $\bar{x}_a(y)$ be the lower bound of $x_a(y)$ given by a $\gamma$-PanOCS in Eqn.~\eqref{eqn:panocs-definition}, and let $\bar{P}$ be the corresponding lower bound of the primal objective in Eqn.~\eqref{eqn:panorama-primal}, i.e.:
\[
    \bar{x}_a(y) \defeq 1 - 2^{-k_a(y)} \cdot (1 - \gamma)^{\max\{k_a(y)-1,0\}}
    \quad\text{,}\qquad
    \bar{P} \defeq \sum_{a \in A} \int_0^{B_a} \bar{x}_a(y) dy
    ~.
\]

Let $\Delta x$ denote the increment of $\bar{x}_a(y)$ as $k_a(y)$ increases, i.e.:
\begin{equation}
\label{eqn:x-increment-random}
    \Delta  x(k) \defeq
    \begin{cases}
        2^{-1} & k=1 ~;\\
        2^{-k} (1-\gamma)^{k-2}(1+\gamma) & k \ge 2 ~.
    \end{cases}
\end{equation}

An online primal dual algorithm's decision for each impression is driven by maximizing $\beta_i$.
For each advertiser $a$, compute two quantities $\Delta_a^D \beta_i$ and $\Delta_a^R \beta_i$ which we shall detail shortly.
The former denotes how much $\beta_i$ would gain if $i$ is assigned to $a$.
The latter denotes how much $\beta_i$ would gain if $i$ is semi-assigned to $a$.
The corresponding subsets are decided by the panoramic interval-level assignment in Section~\ref{sec:panorama}.
Then, find advertisers $a_1$ and $a_2$ with the largest $\Delta_a^R \beta_i$, and advertiser $a^*$ with the largest $\Delta_a^D \beta_i$.
If $\Delta_{a_1}^R \beta_i + \Delta_{a_2}^R \beta_i$ is greater than $\Delta_{a^*}^D \beta_i$, semi-assign $i$ to $a_1$ and $a_2$ in a randomized round.
Otherwise, assign $i$ to $a^*$ in a deterministic round.

Next, we define $\Delta_a^D \beta_i$ and $\Delta_a^R \beta_i$ from two invariants below.
First, let the lower bound of primal equal the dual, i.e., $\bar{P} = D$.
It ensures reverse weak duality in Lemma~\ref{lem:online-primal-dual} because $P \ge \bar{P} = D$.
Second, recall that $\alpha_a(y)$'s account for dual variable $\alpha_a$ at the point-level as explained in Eqn.~\eqref{eqn:alpha-point-level}.
For a set of parameters $\Delta \alpha(\ell)$, $\ell \ge 1$, which will be optimized in the analysis, let:
\begin{equation}
    \label{eqn:alpha-invariant}
    \alpha_a(y) = \sum_ {\ell = 1}^{k_a(y)} \Delta \alpha (\ell)
    ~.
\end{equation}
%


We first derive $\Delta_a^R \beta_i$ from the invariants.
Suppose $i$ is semi-assigned to $a$ and a subset $Y_{ai}$.
For any point $y \in Y_{ai}$, the primal increment due to $y$ is $\Delta x(k_a(y)+1)$, where $k_a(y)$ denotes the value before the semi-assignment.
The dual increment in $\alpha_a(y)$ is $\Delta \alpha(k_a(y)+1)$ by the second invariant.
Finally, by the first invariant, the increment in $\beta_i$ due to point $y \in Y_{ai}$ shall equal the difference between $\Delta x(k_a(y)+1)$ and $\Delta \alpha(k_a(y)+1)$.
For convenience of notations, define:
\begin{equation}
    \label{eqn:beta-r-definition}
    \Delta \beta(k) \defeq \Delta x(k) - \Delta \alpha(k)
    ~.
\end{equation}
Our choice of $\Delta \alpha(k)$ will ensure non-negativity of $\Delta \beta(k)$.
Putting together we get that:
\begin{equation}
    \label{eqn:delta-r-beta-definition}
    \Delta_a^R \beta_i
    \defeq \int_{Y_{ai}} \Delta \beta(k_a(y)+1) dy
    ~.
\end{equation}

Similarly, suppose $i$ is assigned deterministically to $a$ and a subset $Y_{ai}$.
For any point $y \in Y_{ai}$, the primal increment due to $y$ is $\sum_{\ell > k_a(y)} \Delta x(\ell)$ since $k_a(y)$ becomes $\infty$;
the dual increment in $\alpha_a(y)$ is $\sum_{\ell > k_a(y)} \Delta \alpha(\ell)$ by the second invariant.
Thus, together with the first invariant, we let:
%
\begin{equation}
    \label{eqn:delta-d-beta-definition}
    \Delta_a^D \beta_i
    \defeq \int_{Y_{ai}} \sum_{\ell > k_a(y)} \Delta \beta(\ell) dy
    ~.
\end{equation}

\begin{algorithm}[t]
    \caption{Basic Online Primal Dual Algorithm (Parameterized by $\Delta \alpha(k)$, $k \ge 1$)}
    \label{alg:pd-algorithm}
    \begin{algorithmic}
        \smallskip
        \STATE \textbf{state variables:~}\\
            $k_a(y) \ge 0$, number of times $y$ is semi-assigned;
            $k_a(y) = \infty$ if $y$ is assigned in a deterministic round\\[1ex]
        \FORALL{impression $i$}
            \FORALL{advertiser $a \in A$}
                \STATE computer subset $Y_{ai} \subseteq [0,B_a)$ using panoramic interval-level assignment (Section~\ref{sec:panorama})
                \STATE compute $\Delta_a^R \beta_i$ and $\Delta_a^D \beta_i$ according to Equations~\eqref{eqn:beta-r-definition}, \eqref{eqn:delta-r-beta-definition}, and \eqref{eqn:delta-d-beta-definition}
            \ENDFOR
            \STATE find $a_1$, $a_2$ that maximize $\Delta_a^R \beta_i$, and $a^*$ that maximizes $\Delta_a^D \beta_i$
            \STATE \textbf{if} $\Delta^R_{a_1} \beta_i + \Delta^R_{a_2} \beta_i \ge \Delta^D_{a^*} \beta_i$ \hspace*{\fill} \textbf{\emph{\# randomized round}}
                \INDSTATE assign $i$ to what PanOCS selects between $a_1$ and $a_2$ and the corresponding subsets
                \STATE \textbf{else} (i.e., $\Delta^R_{a_1} \beta_i + \Delta^R_{a_2} \beta_i < \Delta^D_{a^*} \beta_i$) \hspace*{\fill} \textbf{\emph{\# deterministic round}}
                    \INDSTATE assign $i$ to $a^*$ and the corresponding subset
                \STATE \textbf{endif}
        \ENDFOR
    \end{algorithmic}
\end{algorithm}

\subsection{Online Primal Dual Analysis}
\label{sec:analysis}
Recall that reverse weak duality always holds because of the first invariant.
Next, we derive a set of conditions on the parameters which imply approximate dual feasibility.
These conditions will be numbered.
Then, we will optimize the competitive ratio $\Gamma$ and $\Delta \alpha(k)$'s through an LP.
For any advertiser $a$ and any subset of impressions $S \subseteq I$, recall approximate dual feasibility in Eqn.~\eqref{eqn:dual-feasibility-panorama}:
\[
    \int_0^{B_a} \alpha_a(y) dy + \sum_{i \in S} \beta_i \ge b_a(S) \cdot \Gamma
    ~.
\]



\subsubsection{Warmup}
\label{sec:dual-analysis-warm-up}

First consider a special case when $S$ has only one impression $i$ who bids $b_{ai} = B_a$.
This warmup case is simple enough to be analyzed in around one page, yet is also general enough to derive the binding conditions on the parameters that are still sufficient in the general case.

We will divide into four subcases, depending on whether impression $i$ is assigned to the advertiser $a$, and whether $i$ is a deterministic or randomized round.
In each case, we will lower bound both $\alpha_a(y)$'s and $\beta_i$ as functions of the $k_a(y)$'s.
To avoid ambiguity, let $k_a(y)$ be the final value at the end of the algorithm, and $k^i_a(y)$ be the value right before the arrival of impression $i$.

\paragraph{Case 1: Round of $i$ is randomized, and $i$ is not semi-assigned to $a$.}
By definition, both $a_1$ and $a_2$ chosen by the algorithm contribute at least $\Delta_a^R \beta_i$ to $\beta_i$.
Hence, by the definition of $\Delta_a^R \beta_i$ in Eqn.~\eqref{eqn:delta-r-beta-definition} and the invariant about $\alpha_a(y)$ in Eqn.~\eqref{eqn:alpha-invariant}, approximate dual feasibility reduces to:
\[
    \int_0^{B_a} \sum_{\ell=1}^{k_a(y)} \Delta \alpha(\ell) dy + 2 \int_0^{B_a} \Delta \beta \big( k_a^i(y) + 1 \big) dy \ge \Gamma \cdot B_a
    ~.
\]

Since $k_a(y) \ge k_a^i(y)$, and the first term is increasing in $k_a(y)$'s, it suffices to prove the inequality when $k_a(y)$ equals $k_a^i(y)$.
We will ensure the inequality pointwisely for every $y \in [0, B_a)$:
\begin{equation}
    \label{eqn:beta-bound-not-to-a}
    \forall k \ge 0 : \qquad \sum_{\ell = 1}^k \Delta \alpha(\ell) + 2 \cdot \Delta \beta(k+1) \ge \Gamma
    ~.
\end{equation}

\paragraph{Case 2: Round of $i$ is deterministic, and $i$ is not assigned to $a$.}
We reduce to the previous case by introducing a condition about the \emph{superiority of randomized rounds.}
\begin{equation}
    \label{eqn:random-vs-deter}
    \forall k \ge 1: \qquad \Delta \beta (k) \ge \sum_{\ell = k+1}^\infty \Delta \beta (\ell)
    ~.
\end{equation}

\begin{lemma}
    \label{lem:random-vs-deter}
    Assuming Eqn.~\eqref{eqn:random-vs-deter}, for any advertiser $a$ and any impression $i$, $2 \cdot \Delta_a^R \beta_i \ge \Delta_a^D \beta_i$.
\end{lemma}

We remark that the lemma holds in the general case as well.
Adding $\Delta \beta(k)$ to both sides of Eqn.~\eqref{eqn:random-vs-deter} gives $2 \cdot \Delta \beta (k) \ge \sum_{\ell = k}^\infty \Delta \beta (\ell)$.
It then follows by the definition of $\Delta_a^R \beta_i$ and $\Delta_a^D \beta_i$ in Equations~\eqref{eqn:delta-r-beta-definition} and \eqref{eqn:delta-d-beta-definition}.
Intuitively, it means that a randomized round with two equally good advertisers in terms of $\Delta_a^R \beta_i$ is better than a deterministic round with only one of them.

By the definition of the algorithm, the advertiser $a^*$ to which the algorithm deterministically assigns $i$ satisfies $\Delta_{a^*}^D \beta_i \ge \Delta_{a^*}^R \beta_i + \Delta_a^R \beta_i$.
Further by $2 \cdot \Delta_{a^*}^R \beta_{i} \ge \Delta_{a^*}^D \beta_{i}$ because of Lemma~\ref{lem:random-vs-deter}, we have $\Delta_{a^*}^R \beta_{i} \ge \Delta_{a}^R \beta_{i}$.
Thus, we get $\beta_i = \Delta_{a^*}^{D} \beta_{i} \ge \Delta_{a}^{R} \beta_{i} +\Delta_{a^*}^{R} \beta_{i} \ge 2 \cdot \Delta_{a}^{R} \beta_{i}$.
The rest is verbatim.

\paragraph{Case 3: Round of $i$ is randomized, and $i$ is semi-assigned to $a$.}
%
Since the algorithm chooses \emph{not} to deterministically assign $i$ to $a$, we have $\beta_i \ge \Delta_a^D \beta_i$.
By the definition of $\Delta_a^D \beta_i$ in Eqn.~\eqref{eqn:delta-d-beta-definition} and the invariant about $\alpha_a(y)$'s in Eqn.~\eqref{eqn:alpha-invariant}, approximate dual feasibility reduces to:
\[
    \int_0^{B_a} \sum_{\ell = 1}^{k_a(y)} \Delta \alpha(\ell) dy + \int_0^{B_a} \sum_{\ell > k_a^i(y)} \Delta \beta(\ell) dy \ge \Gamma \cdot B_a
    ~.
\]

Importantly, since $i$ is semi-assigned to $a$, we have $k_a(y) \ge k_a^i(y) + 1$.
By contrast, the previous cases only have $k_a(y) \ge k_a^i(y)$.
It suffices to ensure the inequality pointwise when $k_a(y) = k_a^i(y) + 1$:
\begin{equation}
    \label{eqn:beta-bound-half-to-a}
    \forall k \ge 1 : \qquad \sum_{\ell = 1}^k \Delta \alpha(\ell) + \sum_{\ell = k}^\infty \Delta \beta(\ell) \ge \Gamma
    ~.
\end{equation}

\paragraph{Case 4: Round of $i$ is deterministic, and $i$ is assigned to $a$.}
We have $k_a(y) = \infty$ for any $y \in [0,B_a)$ after round $i$ because $b_{ai} = B_a$ in the warmup case.
Since $k_a(y)$ may already be very large before $i$ arrives, we do not have any nontrivial lower bound of $\beta_i$.
Hence, $\alpha_a(y)$'s on their own must satisfy approximate dual feasibility.
By the invariant of $\alpha_a(y)$'s in Eqn.~\eqref{eqn:alpha-invariant}, it reduces to:
\begin{equation}
\label{eqn:bound-at-limit}
    \sum_{\ell=1}^\infty \Delta \alpha(\ell) \geq \Gamma
    ~.
\end{equation}

\paragraph{Optimizing the competitive ratio.}
We shall solve an LP, whose variables are the competitive ratio $\Gamma$ and parameters $\Delta \alpha(k)$'s and $\Delta \beta(k)$'s, and whose constraints are the first invariant about $\Delta \alpha(k)$'s and $\Delta \beta(k)$'s in Eqn.~\eqref{eqn:beta-r-definition} and the sufficient conditions for approximate dual feasibility in Equations~\eqref{eqn:beta-bound-not-to-a} to \eqref{eqn:bound-at-limit}.
\begin{align*}
    \label{eqn:matching-lp-primal}
    \text{maximize} \quad & \Gamma \\
    \text{subject to} \quad & \text{Eqn.~\eqref{eqn:beta-r-definition}, \eqref{eqn:beta-bound-not-to-a}, \eqref{eqn:random-vs-deter}, \eqref{eqn:beta-bound-half-to-a}, and \eqref{eqn:bound-at-limit}} \\
    & \Delta \alpha(k), \Delta \beta(k) \ge 0 \qquad\qquad\qquad \forall k \ge 1
\end{align*}

We will present the solution to the LP after giving the analysis for the general case.

\subsubsection{General Case}

We next present a formal proof of approximate dual feasibility in the general case under the invariant in Eqn.~\eqref{eqn:beta-r-definition}, and the conditions derived from the warmup, i.e., Equations \text{\eqref{eqn:beta-bound-not-to-a}} to \eqref{eqn:bound-at-limit}.
To simplify the argument, we further assume \emph{monotonicity of $\Delta \beta$}:
\begin{equation}
    \label{eqn:dual-monotonicity}
        \forall k \ge 1 : \qquad \Delta \beta(k) \ge \Delta \beta(k+1)
        ~.    
\end{equation}
We remark that it would be satisfied automatically by the solution of the LP even if it was not stated explicitly.
The LP becomes:
\begin{equation}
    \label{eqn:basic-factor-lp}
    \begin{aligned}
    \text{maximize} \quad & \Gamma \\
    \text{subject to} \quad & \text{Eqn.~\eqref{eqn:beta-r-definition}, \eqref{eqn:beta-bound-not-to-a}, \eqref{eqn:random-vs-deter}, \eqref{eqn:beta-bound-half-to-a}, \eqref{eqn:bound-at-limit}, and \eqref{eqn:dual-monotonicity}} \\
    & \Delta \alpha(k), \Delta \beta(k) \ge 0 \qquad\qquad\qquad \forall k \ge 1
    \end{aligned}
\end{equation}

\begin{lemma}
    \label{lem:basic-approximate-dual-feasible}
    Suppose $\Gamma$, $\Delta \alpha(k)$'s, and $\Delta \beta(k)$'s form a solution of the LP in Eqn.~\eqref{eqn:basic-factor-lp}.
    Then, Algorithm~\ref{alg:pd-algorithm} satisfies approximate dual feasibility, which we restate below:
    \[
        \int_0^{B_a} \alpha_a(y) dy + \sum_{i \in S} \beta_i \ge b_a(S) \cdot \Gamma
        ~.
    \]
\end{lemma}

Recall that reverse weak duality always holds.
Lemma~\ref{lem:online-primal-dual} leads to the next corollary.

\begin{corollary}
    Algorithm~\ref{alg:pd-algorithm} is $\Gamma$-competitive for any solution of the LP in Eqn.~\eqref{eqn:basic-factor-lp}.
\end{corollary}

\begin{proof}[Proof of Lemma~\ref{lem:basic-approximate-dual-feasible}]
    The values of $\alpha_a(y)$'s are determined by the invariant in Eqn.~\eqref{eqn:alpha-invariant}.
    The points $y$ which are deterministically assigned are special because $\alpha_a(y) \ge \Gamma$ by Eqn.~\eqref{eqn:alpha-invariant} and Eqn.~\eqref{eqn:bound-at-limit}, i.e., $\alpha_a(y)$ on its own satisfies approximate dual feasibility locally at the point-level.
    To refer to these points in the rest of the argument, define:
    %
    \[
        Y_D \defeq \big\{ y \in [0, B_a) : \text{$y$ is deterministically assigned by the end of the algorithm} \big\}
        ~.
    \]

    Similar to the warmup in the last subsection, we will lower bound $\beta_i$ differently depending on the assignment of the impression $i$.
    If $i$ is neither assigned nor semi-assigned to $a$, we will bound $\beta_i$ by $2\cdot\Delta_a^R \beta_i$ like case 1 and 2 in the warmup, and will resort to Eqn.~\eqref{eqn:beta-bound-not-to-a} and Eqn.~\eqref{eqn:random-vs-deter} and the corresponding Lemma~\ref{lem:random-vs-deter}.
    If $i$ is semi-assigned to $a$, we will bound $\beta_i$ by $\Delta_a^D \beta_i$ like case 3 in the warmup, and will resort to Eqn.~\eqref{eqn:beta-bound-half-to-a}.
    If $i$ is assigned to $a$ deterministically, we will use the trivial bound of $\beta_i \ge 0$ like case 4 in the warmup, and will resort to Eqn.~\eqref{eqn:bound-at-limit}.
    To this end, define:
    \begin{align*}
        N & \defeq \big\{ i \in S : \text{$i$ is neither assigned nor semi-assigned to $a$}\} \\
        R & \defeq \big\{ i \in S : \text{$i$ is a randomized round semi-assigned to $a$} \big\}
    \end{align*}

    The rest of the proof is a charging argument as follows.
    We will find subsets $Y_N, Y_R \subseteq [0, B_a)$ and will distribute the contribution from $\beta_i$'s for $i \in N$ and $i \in R$ to the points $y \in Y_N$ and $y \in Y_R$ respectively by defining $\beta(y)$'s such that:
    \begin{itemize}
        \item Subsets $Y_N$, $Y_R$, and $Y_D$ are disjoint.
        \item Subsets $Y_N$, $Y_R$, and $Y_D$ have total measure at least $b_a(S)$, i.e.:
            \begin{equation}
                \label{eqn:basic-total-measure}
                \mu(Y_N) + \mu(Y_R) + \mu(Y_D) \ge b_a(S)
                ~.
            \end{equation}
        \item The values of $\beta(y)$'s lower bound the $\beta_i$'s, i.e.:
            \begin{align}
                \label{eqn:basic-distribution-N}
                \sum_{i \in N} \beta_i & 
                \ge \int_{Y_N} \beta(y) dy ~; \\
                \label{eqn:basic-distribution-R}
                \sum_{i \in R} \beta_i &
                \ge \int_{Y_R} \beta(y) dy ~.
            \end{align}
        \item The values of $\beta(y)$'s satisfy approximate dual feasibility locally at the point-level, i.e.:
            \begin{align}
                \label{eqn:basic-approximate-dual-feasible-N}
                \forall y \in Y_N : \qquad \alpha_a(y) + \beta(y) \ge \Gamma ~; \\[1ex]
                \label{eqn:basic-approximate-dual-feasible-R}
                \forall y \in Y_R : \qquad \alpha_a(y) + \beta(y) \ge \Gamma ~.
            \end{align}
    \end{itemize}

    Assuming the above, approximate dual feasibility follows by a sequence of inequalities as follows:
    \begin{align*}
        \int_0^{B_a} \alpha_a(y) dy + \sum_{i \in S} \beta_i &
        \ge \int_0^{B_a} \alpha_a(y) dy + \sum_{i \in N} \beta_i + \sum_{i \in R} \beta_i
        && \text{($N, R \subseteq S$ and disjoint)} \\
        & \ge \int_0^{B_a} \alpha_a(y) dy + \int_{Y_N} \beta(y) dy + \int_{Y_R} \beta(y) dy && \text{(Eqn.~\eqref{eqn:basic-distribution-N}, \eqref{eqn:basic-distribution-R})} \\[1ex]
        & \ge \int_{Y_N \cup Y_R \cup Y_D} \alpha_a(y) dy + \int_{Y_N} \beta(y) dy + \int_{Y_R} \beta(y) dy
        && \text{($Y_N, Y_R, Y_D \subseteq [0, B_a)$)} \\[1ex]
        & = \int_{Y_N} \big( \alpha_a(y) + \beta(y) \big) dy + \int_{Y_R} \big( \alpha_a(y) + \beta(y) \big) dy\\
        & \quad + \int_{Y_D} \alpha_a(y) dy
        && \text{($Y_N$, $Y_R$, $Y_D$ disjoint)} \\[1ex]
        & \ge \Gamma \cdot \mu(Y_N) + \Gamma \cdot \mu(Y_R) + \Gamma \cdot \mu(Y_D) && \text{(Eqn.~\eqref{eqn:basic-approximate-dual-feasible-N}, \eqref{eqn:basic-approximate-dual-feasible-R}, \eqref{eqn:bound-at-limit})} \\[3ex]
        & \ge \Gamma \cdot b_a(S) ~.
        && \text{(Eqn.~\eqref{eqn:basic-total-measure})}
    \end{align*}

    The rest of the argument substantiates the above plan by constructing subsets $Y_N$, $Y_R$, and the corresponding $\beta(y)$'s and proving that they satisfy the aforementioned properties.
    See Figure~\ref{fig:general-case-charging} for an illustration of the construction.
    
    \begin{figure}
    \begin{subfigure}{.5\textwidth}
        \centering
        \includegraphics[width=\textwidth]{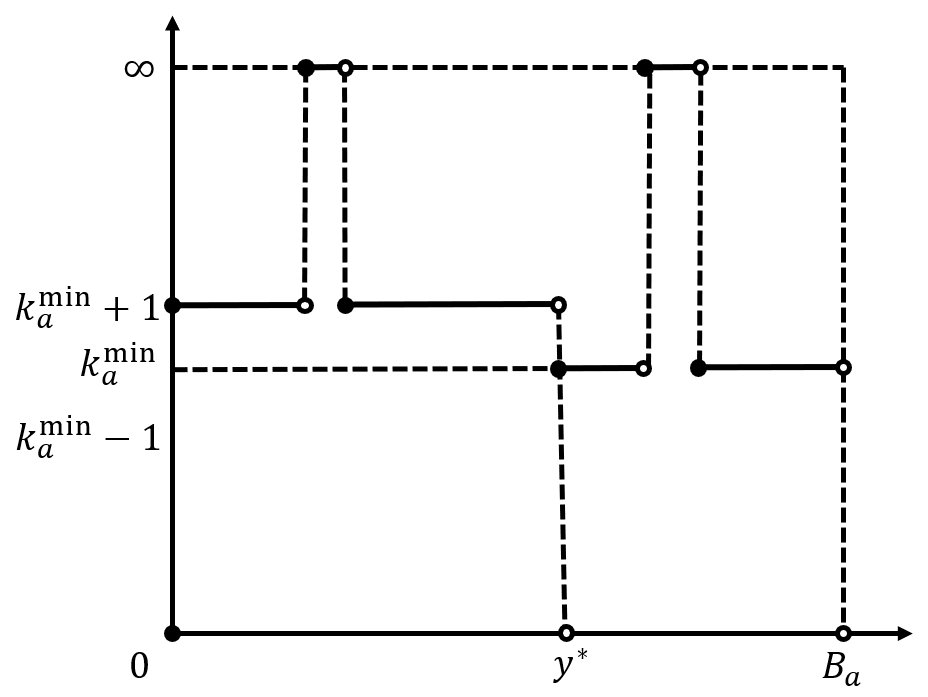}
        \caption{Final status of $a$}
    \end{subfigure}
    \begin{subfigure}{.5\textwidth}
        \centering
        \includegraphics[width=\textwidth]{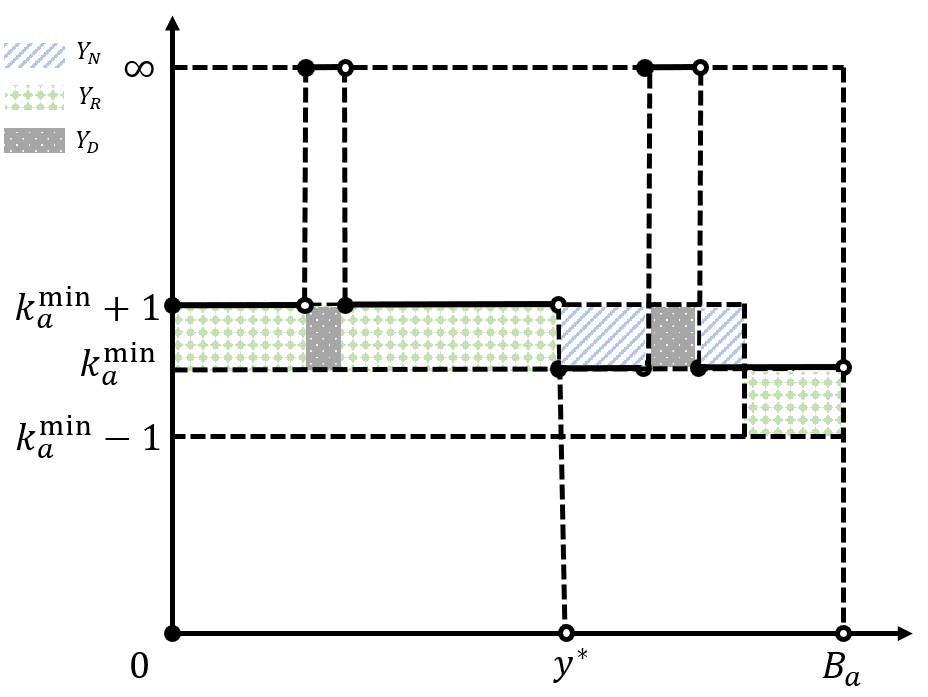}
        \caption{Constructions of $Y_D$, $Y_N$, and $Y_R$}
    \end{subfigure}
    %
    %
    \caption{Illustrative example of the subsets $Y_D$, $Y_N$, and $Y_R$}
    \label{fig:general-case-charging}
\end{figure}

    \paragraph{Construction of $Y_N$ and the Corresponding $\beta(y)$'s.}
    Similar to case 1 and 2 in the warmup, we lower bound $\beta_i$ by $2$ times $\Delta_a^R \beta_i$.
    If $i$ is a randomized round, both advertisers $a_1$ and $a_2$ to which $i$ is semi-assigned contribute at least $\Delta_a^R \beta_i$, or else advertiser $a$ should have been chosen instead.
    If $i$ is a deterministic round, it is the same argument as in the warmup, which we restate below for completeness.
    Since $i$ chooses advertiser $a^*$ deterministically instead of randomizing between advertisers $a^*$ and $a$, we have $\beta_i = \Delta_{a^*}^D \beta_i \ge \Delta_{a^*}^R \beta_i + \Delta_a^R \beta_i$.
    Further by Eqn.~\eqref{eqn:random-vs-deter} and Lemma~\ref{lem:random-vs-deter}, we have $\Delta_{a^*}^D \beta_i \le 2 \Delta_{a^*}^R \beta_i$.
    Cancelling $\Delta_{a^*}^R \beta_i$ by combining the two inequalities leads to $\beta_i = \Delta_{a^*}^D \beta_i \ge 2 \Delta_a^R \beta_i$.
    Recall that $k_a^i(y)$'s denote the values of the state variables when impression $i$ arrives, and $Y_{ai}$ denotes the subset by the panoramic interval-level assignment, should $i$ be semi-assigned or assigned to advertiser $a$ when it arrives.
    By the definition of $\Delta_a^R \beta_i$ in Eqn.~\eqref{eqn:beta-r-definition}:
    \begin{equation}
        \label{eqn:basic-beta-bound-n}
        \forall i \in N : \qquad \beta_i \ge 2 \int_{Y_{ai}} \Delta \beta \big( k_a^i(y)+1 \big) dy
        ~.
    \end{equation}

    We need to further derive a lower bound w.r.t.\ the $k_a(y)$'s at the end of the algorithm.
    This would be easy if we distribute $\beta_i$ to the points $y \in Y_{ai}$ since $k_a(y)$ is nondecreasing over time.
    Such a charging may not work, however, because $Y_{ai}$ may intersect $Y_D$.
    The next lemma resolves this.

    \begin{lemma}
        \label{lem:basic-status-comparison-n}
        For any subset $\tilde{Y}_{ai}$ with measure at most $b_{ai}$, we have:
        \[
            \int_{Y_{ai}} \Delta \beta(k_a^i(y)+1) dy \ge \int_{\tilde{Y}_{ai}} \Delta \beta(k_a(y)+1) dy
            ~.
        \]
    \end{lemma}
    
    \begin{proof}
        Since the panoramic interval-level assignment chooses a subset $Y_{ai}$ of measure $b_{ai}$ with the minimum $k_a^i(y)$'s, by the monotonicity of $\Delta \beta (\cdot)$ in Eqn.~\eqref{eqn:dual-monotonicity} we have:
        \[
            \int_{Y_{ai}} \Delta \beta(k_a^i(y)+1) dy \ge \int_{\tilde{Y}_{ai}} \Delta \beta(k_a^i(y)+1) dy
            ~.
        \]

        Further observe that $k_a(y) \ge k_a^i(y)$ for any $y \in [0, B_a)$.
        Applying the monotonicity of $\Delta \beta (\cdot)$ once again proves the lemma.
    \end{proof}

    Before explaining the definition of $Y_N$, let us recall some notations defined earlier.
    Let $y^*$ be the threshold above which $y \notin Y_D$ satisfies $k_a(y) = \kmin = \min_{z \in [0, B_a)} k_a(z)$, and below which $y \notin Y_D$ satisfies $k_a(y) = \kmin+1$ (see panoramic interval-level assignment and Lemma~\ref{lem:K-property} in Section~\ref{sec:panorama}).
    For any point $y \in [0, B_a)$, any subset $Y \subseteq [0, B_a)$, and any $0 \le b \le B_a$, let $y \oplus_Y b$ denote the point in $[0, B_a)$ such that $[y, y \oplus_Y b)$ excluding $Y$ has Lebesgue measure $b$.
    Further recall our abuse of notation which allows $y \oplus_Y b$ to be smaller than $y$, in which case $[y, y \oplus_Y b)$ denotes the union of $[y, B_a)$ and $[0, y \oplus_Y b)$.
    In the boundary case when $b \ge B_a - \mu(Y)$, define $y \oplus_Y b = y$.
    Define the inverse operator $y \ominus_Y b$ similarly.
    We write $i' < i$ if $i'$ arrives before $i$.

    For any $i \in N$, define $\tilde{Y}_{ai}$ as:

    %
    \begin{equation}
        \label{eqn:basic-beta-charging-per-impression-n}
        \tilde{Y}_{ai} \defeq \Big[ y^* \oplus_{Y_D} \sum_{i' \in N\,:\,i' < i} b_{ai'}, y^* \oplus_{Y_D} \sum_{i' \in N\,:\,i' \le i} b_{ai'} \Big) \setminus Y_D
        ~.
    \end{equation}

    In other words, we scan through the interval $[0, B_a)$ starting from $y^*$ and treating it as a circle by gluing its endpoints.
    Then, we construct $\tilde{Y}_{ai}$'s for $i \in N$ one at a time by their arrival order, letting each be a subset excluding $Y_D$ with measure up to $b_{ai}$.
    If $\sum_{i \in N} b_{ai} \le B_a - \mu(Y_D)$, which we consider to be the canonical case, these would be the panoramic interval-level assignments if these $i \in N$ arrived after the final state of the algorithm and were semi-assigned to $a$.

    Finally, the definition of the boundary case ensures that $y^* \oplus_{Y_D} \sum_{i' \in N\,:\,i' \le i} b_{ai'} = y^*$ when $\sum_{i' \in N\,:\,i' \le i} b_{ai'} \ge B_a - \mu(Y_D)$.
    Therefore, we stop scanning through $[0, B_a)$ after a full circle, and the above $\tilde{Y}_{ai}$'s are disjoint.

    Define $Y_N$ and the corresponding $\beta(y)$ as:
    \begin{equation}
        \label{eqn:basic-beta-charging-n}
        \begin{aligned}
            Y_N & \defeq \bigcup_{i \in N} \tilde{Y}_{ai} ~, \\
            \forall y \in Y_N : \qquad \beta(y) & \defeq 2 \cdot \Delta \beta \big(k_a(y)+1\big) ~.
        \end{aligned}
    \end{equation}

    \paragraph{Proof of Eqn.~\eqref{eqn:basic-distribution-N}.}
    For any $i \in N$, $\tilde{Y}_{ai}$ is a subset with measure at most $b_{ai}$ by definition.
    By Lemma~\ref{lem:basic-status-comparison-n} we have:
    \begin{equation}
        \label{eqn:basic-beta-bound-charging-n}
        \int_{Y_{ai}} \Delta \beta(k_a^i(y)+1) dy \ge \int_{\tilde{Y}_{ai}} \Delta \beta(k_a(y)+1) dy
        ~.
    \end{equation}

    Eqn.~\eqref{eqn:basic-distribution-N} then follows by Eqn.~\eqref{eqn:basic-beta-bound-n}, the above inequality in Eqn.~\eqref{eqn:basic-beta-bound-charging-n}, and the definition of $Y_N$ and the correposnding $\beta(y)$ for $y \in Y_N$ in Eqn.~\eqref{eqn:basic-beta-charging-n} through a sequence of inequalities below:
    \begin{align*}
        \sum_{i \in N} \beta_i
        &
        \ge \sum_{i \in N} \int_{Y_{ai}} 2 \cdot \Delta \beta \big( k_a^i(y) + 1 \big) dy
        &&
        \text{(Eqn.~\eqref{eqn:basic-beta-bound-n})} \\
        &
        \ge \sum_{i \in N} \int_{\tilde{Y}_{ai}} 2 \cdot \Delta \beta \big( k_a(y) + 1 \big) dy
        &&
        \text{(Eqn.~\eqref{eqn:basic-beta-bound-charging-n})} \\
        &
        = \int_{Y_N} \beta(y) dy
        ~.
        &&
        \text{(Eqn.~\eqref{eqn:basic-beta-charging-n})}
    \end{align*}

    \paragraph{Proof of Eqn.~\eqref{eqn:basic-approximate-dual-feasible-N}.}
    Writing $\alpha_a(y)$ as $\sum_{\ell=1}^{k_a(y)} \Delta \alpha(\ell)$ by the invariant in Eqn.~\eqref{eqn:alpha-invariant}, and by the definition of $\beta(y)$ in Eqn.~\eqref{eqn:basic-beta-charging-n}, Eqn.~\eqref{eqn:basic-approximate-dual-feasible-N} reduces to Eqn.~\eqref{eqn:beta-bound-not-to-a}, which is a constraint in the LP:
    \begin{align*}
        \alpha_a(y) + \beta(y)
        &
        = \sum_{\ell=1}^{k_a(y)} \Delta \alpha(\ell) + 2 \cdot \Delta \beta \big( k_a(y) + 1 \big)
        &&
        \text{(Eqn.~\eqref{eqn:alpha-invariant} and \eqref{eqn:basic-beta-charging-n})} \\
        &
        \ge \Gamma ~.
        &&
        \text{(Eqn.~\eqref{eqn:beta-bound-not-to-a})}
    \end{align*}

    \paragraph{Construction of $Y_R$ and the Corresponding $\beta(y)$'s.}
    Similar to case 3 in the warmup, we lower bound $\beta_i$ by $\Delta_a^D \beta_i$, which holds because the algorithm does not assign $i$ to advertiser $a$ deterministically.
    Further by the definition of $\Delta_a^D \beta_i$ in Eqn.~\eqref{eqn:delta-d-beta-definition}, we have:
    \begin{equation}
        \label{eqn:basic-beta-bound-r}
        \forall i \in R : \qquad \beta_i \ge \int_{Y_{ai}} \sum_{\ell = k^i_a(y)+1}^\infty \Delta \beta(\ell) dy
        ~.
    \end{equation}

    We need to further derive a lower bound w.r.t.\ the $k_a(y)$'s at the end of the algorithm.
    The next lemma is similar to Lemma~\ref{lem:basic-status-comparison-n} in the previous case, but more generally considers arbitrary $\hat{k}_a(y) \ge k_a^i(y)$ instead of only the $k_a(y)$ at the end of the algorithm.

    \begin{lemma}
        \label{lem:basic-status-comparison-r}
        Consider any $\hat{k}_a(y)$'s such that $\hat{k}_a(y) \ge k_a^i(y)$ for any $y \in [0, B_a)$.
        Then, for any subset $\tilde{Y}_{ai}$ with measure at most $b_{ai}$, we have:
        \[
            \int_{Y_{ai}} \sum_{\ell = k^i_a(y)+1}^\infty \Delta \beta(\ell) dy \ge \int_{\tilde{Y}_{ai}} \sum_{\ell = \hat{k}_a(y)+1}^\infty \Delta \beta(\ell) dy
            ~.
        \]
    \end{lemma}

    \begin{proof}
        Since the panoramic interval-level assignment chooses a subset $Y_{ai}$ of measure $b_{ai}$ with the minimum $k_a^i(y)$'s, we have:
        \[
            \int_{Y_{ai}} \sum_{\ell = k^i_a(y)+1}^\infty \Delta \beta(\ell) dy \ge \int_{\tilde{Y}_{ai}} \sum_{\ell = k^i_a(y)+1}^\infty \Delta \beta(\ell) dy
            ~.
        \]

        The lemma then follows by the assumption that $\hat{k}_a(y) \ge k_a^i(y)$ for any $y \in [0, B_a)$.
    \end{proof}

    For any $i \in R$, define $\tilde{Y}_{ai}$ as:
    %
    \begin{equation}
        \label{eqn:basic-beta-charging-per-impression-r}
        \tilde{Y}_{ai} \defeq \Big[ y^* \ominus_{Y_D} \sum_{i' \in R\,:\,i' \ge i} b_{ai'}, y^* \ominus_{Y_D} \sum_{i' \in R\,:\,i' > i} b_{ai'} \Big) \setminus Y_D
        ~.
    \end{equation}

    In other words, scan \emph{backwards} through the interval $[0, B_a)$ starting from $y^*$ and treating it as a circle by gluing its endpoints.
    Then, construct $\tilde{Y}_{ai}$'s for $i \in R$ one at a time by the opposite of their arrival order, letting each be a subset excluding $Y_D$ with measure up to $b_{ai}$.
    If $\sum_{i \in N} b_{ai} \le B_a - \mu(Y_D)$, which we consider to be the canonical case, these would be the panoramic interval-level assignments if the $i \in R$ arrived at the end of the instance, assuming the same final state of the algorithm.
    Again, by the definition of the boundary case, we stop scanning through $[0, B_a)$ after a full circle and, therefore, the above $\tilde{Y}_{ai}$'s are disjoint.

    Define $Y_R$ and the corresponding $\beta(y)$ as:
    \begin{equation}
        \label{eqn:basic-beta-charging-r}
        \begin{aligned}
            Y_R & \defeq \bigcup_{i \in R} \tilde{Y}_{ai} \setminus Y_N ~, \\
            \forall y \in Y_R : \qquad \beta(y) & \defeq \sum_{\ell = k_a(y)}^{\infty} \Delta \beta \big(\ell\big) ~.
        \end{aligned}
    \end{equation}

    \paragraph{Proof of Eqn.~\eqref{eqn:basic-distribution-R}.}
    For any $i \in R$, define $k^{-i}_a(y)$ by considering what the state variables of advertiser $a$ would have been before the arrival of $i$ \emph{if the impressions in $R$ were the latest ones in the instance}.
    More precisely, for any $i \in R$, let:
    \[
        k^{-i}_a(y) \defeq \begin{cases}
            k_a(y) - 1 & y \in \big[ y^* \ominus_{Y_D} \sum_{i' \in R\,:\,i' \ge i} b_{ai'}, y^* \big) \setminus Y_D ~; \\[1ex]
            k_a(y) & \text{otherwise.}
        \end{cases}
    \]

    Intuitively, these are the largest possible values of $k_a^i(y)$'s, as the next lemma formalizes.

    \begin{lemma}
        \label{lem:basic-k-comparison-r}
        For any $i \in R$ and any $y \in [0, B_a)$ we have:
        \[
            k_a^{-i}(y) \ge k_a^i(y)
            ~.
        \]
    \end{lemma}
    
    \begin{proof}
        Suppose for contrary $k_a^{-i}(y) < k_a^{i}(y)$ for some impression $i \in R$ some point $y \in [0,B_a)$.
        By the definition of $k_a^{-i}(y)$, this means $k_a^{i}(y) = k_a(y) < \infty $, and $k_a^{-i}(y) = k_a(y)-1$.
        Importantly, by $k_a^{i}(y) = k_a(y)$, the panoramic assignment can only assign or semi-assign impressions after and including $i$ to points in $(y, y^*)$, and at most once per point.
        This is the backbone of the proof.

        Its first implications is that all impressions $i' \in R$ after and including $i$ are semi-assigned to \emph{disjoint} subsets of $(y, y^*)$:
        \[
            \bigcup_{i' \in R\,:\,i' \ge i} Y_{ai'} \subset (y, y^*)
            ~.
        \]

        Its second implication, which may be less obvious, is that points in $Y_D$ cannot be semi-assigned since the arrival of impression $i$.
        Consider any point $y \in Y_D$.
        If it has already been deterministically assigned by the time impression $i$ arrives, the claim holds trivially.
        Otherwise, it can be assigned or semi-assigned at most once since the arrival of $i$.
        This last opportunity must be used for a deterministic assignment or else $y$ would not be in $Y_D$.
        Together with the first implication, we get that all impressions $i' \in R$ after and including $i$ are semi-assigned to \emph{disjoint} subsets of $(y, y^*) \setminus Y_D$:
        \[
            \bigcup_{i' \in R\,:\,i' \ge i} Y_{ai'} \subset (y, y^*) \setminus Y_D
            ~.
        \]

        However, the LHS has total measure $\sum_{i' \in R : i' \ge i} b_{ai'}$ by definition, and the RHS has total measure strictly less than that due to $k_a^{-i}(y) = k_a(y) - 1$.
        We have a contradiction.
    \end{proof}

    Consider any $i \in R$.
    By definition, $\tilde{Y}_{ai}$ is a subset with measure at most $b_{ai}$.
    Further, Lemma~\ref{lem:basic-k-comparison-r} above allows us to apply Lemma~\ref{lem:basic-status-comparison-r}:
    \[
        \int_{Y_{ai}} \sum_{\ell = k_a^i(y) + 1}^{\infty} \Delta \beta(\ell) dy \ge \int_{\tilde{Y}_{ai}} \sum_{\ell = k_a^{-i}(y)+1}^\infty \Delta \beta(\ell) dy
        ~.
    \]

    Finally, observe that $k_a^{-i}(y) = k_a(y) - 1$ for any $y \in \tilde{Y}_{ai}$.
    We have:
    \begin{equation}
        \label{eqn:basic-beta-bound-charging-r}
        \int_{Y_{ai}} \sum_{\ell = k_a^i(y) + 1}^{\infty} \Delta \beta(\ell) dy \ge \int_{\tilde{Y}_{ai}} \sum_{\ell = k_a(y)}^\infty \Delta \beta(\ell) dy
        ~.
    \end{equation}

    Eqn.~\eqref{eqn:basic-distribution-R} then follows by Eqn.~\eqref{eqn:basic-beta-bound-r}, the above inequality in Eqn.~\eqref{eqn:basic-beta-bound-charging-r}, and the definition of $Y_R$ and the correposnding $\beta(y)$ for $y \in Y_R$ in Eqn.~\eqref{eqn:basic-beta-charging-r}, through a sequence of inequalities as follows:
    \begin{align*}
        \sum_{i \in R} \beta_i
        &
        \ge \sum_{i \in R} \int_{Y_{ai}} \sum_{\ell = k_a^i(y)+1}^\infty \Delta \beta ( \ell ) dy
        &&
        \text{(Eqn.~\eqref{eqn:basic-beta-bound-r})} \\
        &
        \ge \sum_{i \in R} \int_{\tilde{Y}_{ai}} \sum_{\ell = k_a(y)}^\infty \Delta \beta ( \ell ) dy
        &&
        \text{(Eqn.~\eqref{eqn:basic-beta-bound-charging-r})} \\
        &
        \ge \int_{Y_R} \beta(y) dy
        ~.
        &&
        \text{(Eqn.~\eqref{eqn:basic-beta-charging-r})}
    \end{align*}

    \paragraph{Proof of Eqn.~\eqref{eqn:basic-approximate-dual-feasible-R}.}
    Writing $\alpha_a(y)$ as $\sum_{\ell=1}^{k_a(y)} \Delta \alpha(\ell)$ by the invariant in Eqn.~\eqref{eqn:alpha-invariant}, and by the definition of $\beta(y)$ in Eqn.~\eqref{eqn:basic-beta-charging-r}, Eqn.~\eqref{eqn:basic-approximate-dual-feasible-R} reduces to Eqn.~\eqref{eqn:beta-bound-half-to-a}, which is a constraint in the LP:
    \begin{align*}
        \alpha_a(y) + \beta(y)
        &
        = \sum_{\ell=1}^{k_a(y)} \Delta \alpha(\ell) + \sum_{\ell = k_a(y)}^\infty \Delta \beta ( \ell )
        &&
        \text{(Eqn.~\eqref{eqn:alpha-invariant} and \eqref{eqn:basic-beta-charging-r})} \\
        &
        \ge \Gamma ~.
        &&
        \text{(Eqn.~\eqref{eqn:beta-bound-half-to-a})}
    \end{align*}

    \paragraph{Disjointness.}
    This follows directly by the constructions of $Y_N$ and $Y_R$.
    The construction of $Y_N$ explicitly rules out points in $Y_D$ in Eqn.~\eqref{eqn:basic-beta-charging-per-impression-n}.
    The construction of $Y_R$ explicitly rules out points in $Y_D$ in Eqn.~\eqref{eqn:basic-beta-charging-per-impression-r}, and rules out points in $Y_R$ in Eqn.~\eqref{eqn:basic-beta-charging-r}.

    \paragraph{Bounding the Total Measure: Proof of Eqn.~\eqref{eqn:basic-total-measure}.}
    Observe that:
    \[
        \bigcup_{i \in N} \tilde{Y}_{ai} = \Big[ y^* ~,~ y^* \oplus_{Y_D} \sum_{i \in N} b_{ai} \Big) \setminus Y_D
        \quad,\quad
        \bigcup_{i \in R} \tilde{Y}_{ai} = \Big[ y^* \ominus_{Y_D} \sum_{i \in R} b_{ai} ~,~ y^* \Big) \setminus Y_D
        ~.
    \]

    If $\sum_{i \in N} b_{ai} + \sum_{i \in R} b_{ai} \ge B_a - \mu(Y_D)$, the union of $Y_D$, $Y_N$, and $Y_R$ covers $[0, B_a)$.
    Eqn.~\eqref{eqn:basic-total-measure} then holds trivially because the LHS equals $B_a$ and the RHS is upper bounded by $B_a$.

    Otherwise, the definition of $Y_R$ simplifies as $Y_R = \cup_{i \in R} \tilde{Y}_{ai}$, and we have:
    \[
        \mu(Y_R) = \sum_{i \in R} b_{ai}
        \quad,\quad
        \mu(Y_N) = \sum_{i \in N} b_{ai}
        ~.
    \]

    Finally, any $i \in S \setminus (N \cup R)$ is deterministically assigned to $a$ and therefore:
    \[
        \mu(Y_D) \ge \sum_{i \in S \setminus (N \cup R)} b_{ai}
        ~.
    \]

    Putting together proves Eqn.~\eqref{eqn:basic-total-measure}: 
    \[
        \mu(Y_N) + \mu(Y_R) + \mu(Y_D) \ge \sum_{i \in S} b_{ai} \ge b_a(S)
        ~.
    \]
\end{proof}

\subsection{Optimizing the Parameters}
\label{sec:basic-lp-solution}

It remains to optimize the parameters $\Delta \alpha(k)$'s and $\Delta \beta(k)$'s and the competitive ratio $\Gamma$ by solving the LP in Eqn.~\eqref{eqn:basic-factor-lp}.
Observe that the LP has countably infinitely many variables and constraints and therefore, cannot be directly solved with an LP solver.
One possible strategy is to solve a finite restriction by setting $\Delta \alpha(k) = \Delta \beta (k) = 0$ for $k > \kmax$ with some sufficiently large $\kmax$.
This is indeed the strategy we use for the hybrid algorithm in Section~\ref{sec:hybrid}.
Fortunately, the LP in Eqn.~\eqref{eqn:basic-factor-lp} admits benign structures.
As a result, we can provide an explicit solution.

\begin{lemma}
    \label{lem:basic-lp-solution}
    For any $0 \le \gamma \le 1$, the following is a solution of the LP in Eqn.~\eqref{eqn:basic-factor-lp}:
    \begin{equation}
        \label{eqn:basic-lp-solution}
        \begin{aligned}
            \Gamma & = \frac{3 + 2\gamma}{6+3\gamma} ~. \\[2ex]
            \Delta \alpha(k) & =
            \begin{cases}
                \displaystyle
                \frac{3+\gamma}{6+3\gamma} \Delta x(1) & k = 1 ~; \\[2ex]
                \displaystyle
                \frac{1+\gamma}{2+\gamma} \Delta x(k) & k \ge 2 ~. \\
            \end{cases} \\[2ex]
            \Delta \beta(k) & =
            \begin{cases}
                \displaystyle
                \frac{3+2\gamma}{6+3\gamma} \Delta x(1) & k = 1 ~; \\[2ex]
                \displaystyle
                \frac{1}{2+\gamma} \Delta x(k) & k \ge 2 ~.
            \end{cases}
        \end{aligned}
    \end{equation}
\end{lemma}

\begin{proof}
    Below we verify the constraints of the LP, i.e., Eqn.~\eqref{eqn:beta-r-definition}, \eqref{eqn:beta-bound-not-to-a}, \eqref{eqn:random-vs-deter}, \eqref{eqn:beta-bound-half-to-a}, \eqref{eqn:bound-at-limit}, and \eqref{eqn:dual-monotonicity}.
    
    \paragraph{Eqn.~\eqref{eqn:beta-r-definition}:}
    The invariant $\Delta \alpha(k) + \Delta \beta(k) = \Delta x(k)$ is guaranteed explicitly by the definition of $\Delta \alpha(k)$'s and $\Delta \beta(k)$'s above.

    \paragraph{Eqn.~\eqref{eqn:beta-bound-not-to-a}:}
    It holds with equality as shown below.
    \begin{align*}
        \sum_{\ell = 1}^k \Delta \alpha(\ell) + 2 \cdot \Delta \beta(k+1)
        &
        = \Delta \alpha(1) + \sum_{\ell=2}^k \Delta \alpha(\ell) + 2 \cdot \Delta \beta(k+1) \\
        &
        = \frac{3+\gamma}{6+3\gamma} \Delta x(1) + \frac{1+\gamma}{2+\gamma} \sum_{\ell=2}^k \Delta x(\ell) + \frac{2}{2+\gamma} \Delta x(k+1) \\
        &
        = \frac{3+\gamma}{12+6\gamma} + \frac{1+\gamma}{4+2\gamma} \Big(1 - \big(\frac{1-\gamma}{2}\big)^{k-1} \Big) + \frac{1+\gamma}{4+2\gamma} \big(\frac{1-\gamma}{2}\big)^{-k-1} \\[1ex]
        &
        = \frac{3+2\gamma}{6+3\gamma}
        ~.
    \end{align*}

    \paragraph{Eqn.~\eqref{eqn:random-vs-deter}:}
    By definition, we have:
    \[
        \sum_{\ell=k+1}^\infty \Delta \beta(\ell) = \frac{1}{2+\gamma} \sum_{\ell=k+1}^{\infty} \Delta x(\ell)
        ~.
    \]

    Next we show:
    \[
        \sum_{\ell = k+1}^\infty \Delta x(\ell) \le \Delta x(k)
        ~.
    \]

    It holds with equality when $k = 1$ because both sides equal $\frac{1}{2}$.
    When $k \ge 2$, it follows by the observation that $\Delta x(\ell+1) \le \frac{1}{2} \Delta x(\ell)$ for any $\ell \ge 2$.

    Together we have:
    %
    \[
        \sum_{\ell=k+1}^\infty \Delta \beta(\ell) \le \frac{1}{2+\gamma} \Delta x(k)
    \]

    The RHS equals $\Delta \beta(k)$ by definition when $k \ge 2$.
    For $k = 1$, this is less than $\Delta \beta(k)$ due to $\frac{3+2\gamma}{6+3\gamma} \ge \frac{1}{2+\gamma}$ for any $\gamma \ge 0$.
    In both cases we get Eqn.~\eqref{eqn:random-vs-deter}.

    \paragraph{Eqn.~\eqref{eqn:beta-bound-half-to-a}:}
    First consider $k = 1$.
    The constraint holds with strict inequality:
    \begin{align*}
        \sum_{\ell = 1}^k \Delta \alpha(\ell) + \sum_{\ell=k}^\infty \Delta \beta(\ell)
        &
        = \Delta \alpha(1) + \Delta \beta(1) + \sum_{\ell=2}^\infty \Delta \beta(\ell) \\
        &
        = \Delta x(1) + \frac{1}{2+\gamma} \sum_{\ell=2}^\infty \Delta x(\ell) \\
        &
        = \frac{1}{2} + \frac{1}{4+2\gamma} \\[1ex]
        &
        > \frac{3+2\gamma}{6+2\gamma}
        ~.
    \end{align*}

    The last inequality holds for any $\gamma \le 1$.

    Next, consider $k \ge 2$.
    The constraint still holds with strict inequality:
    \begin{align*}
        \sum_{\ell = 1}^k \Delta \alpha(\ell) + \sum_{\ell=k}^\infty \Delta \beta(\ell)
        &
        = \Delta \alpha(1) + \sum_{\ell=2}^{k-1} \Delta \alpha(\ell) + \big( \Delta \alpha(k) + \Delta \beta(k) \big) + \sum_{\ell=k+1}^\infty \Delta \beta(\ell) \\
        &
        = \frac{3+\gamma}{6+3\gamma} \Delta x(1) + \frac{1+\gamma}{2+\gamma} \sum_{\ell=2}^{k-1} \Delta x(\ell) + \Delta x(k) + \frac{1}{2+\gamma} \sum_{\ell=k+1}^\infty \Delta x(\ell) \\[1ex]
        &
        = \frac{3+\gamma}{12+6\gamma} + \frac{1+\gamma}{4+2\gamma} \Big(1 - \big(\frac{1-\gamma}{2}\big)^{k-2} \Big) + \frac{1+\gamma}{4} \big(\frac{1-\gamma}{2}\big)^{k-2} + \frac{1}{2} \big(\frac{1-\gamma}{2}\big)^{-k-1} \\[2ex]
        &
        = \frac{3+2\gamma}{6+3\gamma} + \frac{1}{4+2\gamma} \big(\frac{1-\gamma}{2}\big)^{-k-1} \\[2ex]
        &
        > \frac{3+2\gamma}{6+3\gamma}
        ~.
    \end{align*}

    \paragraph{Eqn.~\eqref{eqn:bound-at-limit}:}
    It holds with equality.
    In fact, it can be seen as the limit case of Eqn.~\eqref{eqn:beta-bound-not-to-a} or Eqn.~\eqref{eqn:beta-bound-half-to-a} when $k$ goes to infinity.
    We include the calculation below for completeness.
    \begin{align*}
        \sum_{\ell = 1}^\infty \Delta \alpha(\ell)
        &
        = \Delta \alpha(1) + \sum_{\ell=2}^{\infty} \Delta \alpha(\ell) \\
        &
        = \frac{3+\gamma}{6+3\gamma} \Delta x(1) + \frac{1+\gamma}{2+\gamma} \sum_{\ell=2}^{\infty} \Delta x(\ell) \\
        &
        = \frac{3+\gamma}{12+6\gamma} + \frac{1+\gamma}{4+2\gamma} \\[1ex]
        & 
        = \frac{3+2\gamma}{6+3\gamma}
        ~.
    \end{align*}

    \paragraph{Eqn.~\eqref{eqn:dual-monotonicity}:}
    By $\Delta x(k) > \Delta x(k+1)$ and $\frac{3+2\gamma}{6+3\gamma} \ge \frac{1}{2+\gamma}$, we have $\Delta \beta(k) > \Delta \beta(k+1)$ from the definition of $\Delta \beta(k)$'s.
    The strict inequality substantiates our earlier remark that the constraint would be satisfied automatically by the solution of the LP even if it was not stated explicitly.
\end{proof}

\subsubsection*{Large Bids: the Crux of AdWords}

In light of the positive results by \citet{MehtaSVV/JACM/2007} for small bids that are at most half the budgets, the case of large bids, i.e., $\frac{1}{2} B_a < b_{ai} \le B_a$, can be viewed as the crux of the AdWords problem.
As a direct corollary of Lemma~\ref{lem:basic-lp-solution} and the $0.05144$-PanOCS for large bids in Theorem~\ref{thm:panocs-large-bid}, we get the first online algorithm that breaks the $0.5$ barrier in the crux.

\begin{theorem}
    \label{thm:basic-large-bid}
    Suppose all nonzero bids are large, i.e., $\frac{1}{2} B_a < b_{ai} \le B_a$ or $b_{ai} = 0$ for any advertiser $a \in A$ and any impression $i \in I$.
    Then, Algorithm~\ref{alg:pd-algorithm} with the $\gamma = 0.05144$-PanOCS in Theorem~\ref{thm:panocs-large-bid} is $\Gamma$-competitive for:
    \[
        \Gamma = \frac{3+2\gamma}{6+3\gamma} > 0.5041
        ~.
    \]
\end{theorem}

\subsubsection*{General Bids: a Weaker Version of Theorem~\ref{thm:main}}

Next, consider a restricted version of Algorithm~\ref{alg:pd-algorithm} such that for any advertiser $a$ and any point $y \in [0, B_a)$, $y$ is semi-assigned at most $\kmax$ times for some positive integer $\kmax$.
This can be achieved by letting $\Delta x(k) = \Delta \alpha(k) = \Delta \beta(k) = 0$ for any $k > \kmax$.
The restriction allows us to use the PanOCS for general bids in Theorem~\ref{thm:panocs-general-bid}.
We need, however, a solution to the LP in Eqn.~\eqref{eqn:basic-factor-lp} under the restriction.
A natural choice is adopting the solution in Lemma~\ref{lem:basic-lp-solution} directly for $k \le \kmax$, and decreasing the competitive ratio $\Gamma$ accordingly to preserve feasibility.

\begin{lemma}
    \label{lem:basic-lp-solution-general}
    For any $0 \le \gamma \le 1$, the following is a solution of the LP in Eqn.~\eqref{eqn:basic-factor-lp}:
    \begin{equation}
        \label{eqn:basic-lp-solution-general}
        \begin{aligned}
            \Gamma & = \frac{3 + 2\gamma}{6+3\gamma} - 2^{-\kmax} (1-\gamma)^{\kmax-1} ~. \\[2ex]
            \Delta \alpha(k) & =
            \begin{cases}
                \displaystyle
                \frac{3+\gamma}{6+3\gamma} \cdot \Delta x(1) & k = 1 ~; \\[2ex]
                \displaystyle
                \frac{1+\gamma}{2+\gamma} \cdot \Delta x(k) & 2 \le k \le \kmax ~; \\[2ex]
                0 & k > \kmax ~. \\
            \end{cases} \\[2ex]
            \Delta \beta(k) & =
            \begin{cases}
                \displaystyle
                \frac{3+2\gamma}{6+3\gamma} \cdot \Delta x(1) & k = 1 ~; \\[2ex]
                \displaystyle
                \frac{1}{2+\gamma} \cdot \Delta x(k) & 2 \le k \le \kmax ~; \\[2ex]
                0 & k > \kmax ~.
            \end{cases}
        \end{aligned}
    \end{equation}
\end{lemma}

\begin{proof}
    It follows by Lemma~\ref{lem:basic-lp-solution} that the contributions of $\Delta \alpha(k)$'s and $\Delta \beta(k)$'s, $k > \kmax$, to the approximate dual feasibility constraints, i.e., Eqn.~\eqref{eqn:beta-bound-not-to-a}, \eqref{eqn:beta-bound-half-to-a}, and \eqref{eqn:bound-at-limit}, are at most:
    \[
        \sum_{\ell = \kmax+1}^\infty \Delta x (\ell) = 2^{-\kmax} (1-\gamma)^{\kmax-1}
        ~.
    \]

    Hence, decreasing the competitive ratio $\Gamma$ by this amount restores approximate dual feasibility even after setting $\Delta \alpha(k) = \Delta \beta(k) = 0$ for $k > \kmax$.

    Finally, observe that the other constraints, i.e., Eqn.~\eqref{eqn:beta-r-definition}, \eqref{eqn:random-vs-deter}, and \eqref{eqn:dual-monotonicity}, are trivially preserved after letting $\Delta x(k) = \Delta \alpha(k) = \Delta \beta(k) = 0$ for $k > \kmax$.
\end{proof}

Consider the $\gamma$-PanOCS in Theorem~\ref{thm:panocs-general-bid} where $\gamma = 0.01245 \cdot \kmax^{-1}$.
The competitive ratio from the above solution is:
\[
    \Gamma = \frac{3 + 2\gamma}{6+3\gamma} - 2^{-\kmax} (1-\gamma)^{\kmax-1} = \frac{1}{2} + \Omega \big( \kmax^{-1} \big) - 2^{-\kmax} (1-\gamma)^{\kmax-1}
    ~.
\]

The second term is inverse proportional to $\kmax$ while the third term decreases exponentially in $\kmax$.
Hence, by choosing a sufficiently large $\kmax$, the competitive ratio is strictly larger than half.
Indeed, letting $\kmax = 18$ gives $\Gamma > 0.50005$.

\clearpage

\section{Panoramic Online Correlated Selection}
\label{sec:panocs}

This section details the design and analysis of the PanOCS algorithms used in the last section.
We first restate the definition of PanOCS below.

\defpanocs*

We start with a warmup algorithm in Subsection~\ref{sec:panocs-warmup} which gives a simple yet weaker $\frac{1}{64}$-PanOCS for large bids, substantiating the proof sketch in the previous section.
Then, we explain how to improve and generalize the warmup algorithm to prove Theorem~\ref{thm:panocs-large-bid} and Theorem~\ref{thm:panocs-general-bid} in Subsection~\ref{sec:panocs-large-bid} and Subsection~\ref{sec:panocs-general-bid} respectively.

\subsection{Warmup: \texorpdfstring{$\frac{1}{64}$}{1/64}-PanOCS for Large Bids}
\label{sec:panocs-warmup}

This subsection explains the basics of PanOCS algorithms and their analyses through a proof of the following theorem.

\begin{theorem}[Weaker Version of Theorem~\ref{thm:panocs-large-bid}]
    \label{thm:panocs-large-bid-weak}
    Suppose all nonzero bids are large, i.e., we have $\frac{1}{2} B_a < b_{ai} \le B_a$ or $b_{ai} = 0$ for any $a \in A$ and any $i \in I$.
    Then, there is a $\frac{1}{64} \approx 0.0156$-PanOCS.
\end{theorem}

We adopt the concepts and \emph{ex-ante} and \emph{ex-post} dependence graphs from the research on OCS by Huang and Tao~\cite{Huang/arXiv/2019, HuangT/arXiv/2019}.
To avoid confusion with the vertices and edges in the bipartite graph of AdWords, we shall refer to the counterparts in the dependence graphs as nodes and arcs respectively.

\paragraph{Ex-ante Dependence Graph.}
Let $I^R$ denote the set of impressions in randomized rounds.
The \emph{ex-ante} dependence graph $D$ is a directed graph with a node for every impression $i \in I^R$;
we abuse notation and refer to the node also as $i$.
Recall that we write $i < i'$ if $i$ arrives before $i'$.

\begin{definition}[Correlation among Randomized Rounds]
    \label{def:randomized-round-correlation}
    Suppose $i < i'$ are two impressions semi-assigned to an advertiser $a$ and to subsets $Y_{ai}$ and $Y_{ai'}$ respectively.
    \begin{enumerate}
        \item They are \emph{related} w.r.t.\ advertiser $a$ if the subsets overlap, i.e., if there exists $y \in Y_{ai} \cap Y_{ai'}$.
        \item If further there is no impression $i''$ between them, i.e., $i < i'' < i'$, which is also semi-assigned to advertiser $a$ and a subset containing $y$, we say that $i < i'$ are \emph{adjacent} w.r.t.\ advertiser $a$.
        \item Otherwise, we say that $i < i'$ are \emph{unrelated} w.r.t.\ advertiser $a$.
    \end{enumerate}
\end{definition}


For large bids, two impressions semi-assigned to the same advertiser $a$ are always related.
We make the above definition more general so that it applies to arbitrary bids.

Let there be an arc $(i, i')_a$ in the \emph{ex-ante} dependence graph $D$ if $i < i'$ are adjacent w.r.t.\ $a$.
The subscript helps distinguish parallel arcs, when $i$ and $i'$ are semi-assigned to the same pairs of advertisers (yet potentially distinct subsets).
See Figure~\ref{fig:dependence-graph} for an illustrative example.

Further, we sometimes say that $i < i'$ are related or adjacent without specifying an advertiser, which means the relation holds \emph{for some} advertiser $a$.
Similarly, we say that $i < i'$ are unrelated without specifying an advertiser, which means they are unrelated w.r.t.\ \emph{any} advertiser $a$.

Informally, our PanOCS ensures that for any pair of adjacent nodes $i < i'$, with probability $\gamma$ the $\gamma$-PanOCS correlates the decisions perfectly negatively:
it selects advertiser $a$ in round $i'$ if it does not select $a$ in round $i$, and vice versa.
Further, if two nodes are related, the selections therein are either independent or negatively correlated.
Finally, if two nodes are unrelated, the selections could be arbitrarily correlated.
In this subsection and the next, the PanOCS algorithms for large bids make pairwise independent decisions for unrelated nodes.
The PanOCS for general bids in the last subsection, however, crucially utilizes the freedom of correlating unrelated nodes positively.

\begin{lemma}
    \label{lem:large-bid-bounded-degree}
    If all nonzero bids are large, any $i \in I^R$ has at most $4$ out-arcs and at most $4$ in-arcs.
\end{lemma}

\begin{proof}
    Fix any impression $i \in I^R$.
    Let $a$ and $a'$ be the advertisers chosen in this randomized round.
    We will show that there are at most two out-arcs $(i, i')_a$ w.r.t.\ advertiser $a$.
    Then, by symmetric arguments, there are at most two out-arcs w.r.t.\ advertiser $a'$, and at most two in-arcs w.r.t.\ each of $a$ and $a'$.
    Putting together proves the lemma.

    Let $i_1$ and $i_2$ be the next two randomized rounds after $i$ which semi-assign to advertiser $a$.
    We claim that $i$ has no out-arcs to any node other than $i_1$ and $i_2$ w.r.t.\ advertiser $a$, because when bids are large every point $y \in [0, B_a)$ is semi-assigned at least once in rounds $i_1$ and $i_2$.
    Hence, $i < i'$ cannot be adjacent for any later impression $i' \ne i_1, i_2$ because for any choice of $y \in [0, B_a)$ there always exists $i'' = i_1$ or $i_2$ such that $i < i'' < i'$ and $y \in Y_{ai''}$.
\end{proof}


\begin{figure}
    \centering
    \begin{subfigure}{.48\textwidth}
        \includegraphics[width=\textwidth]{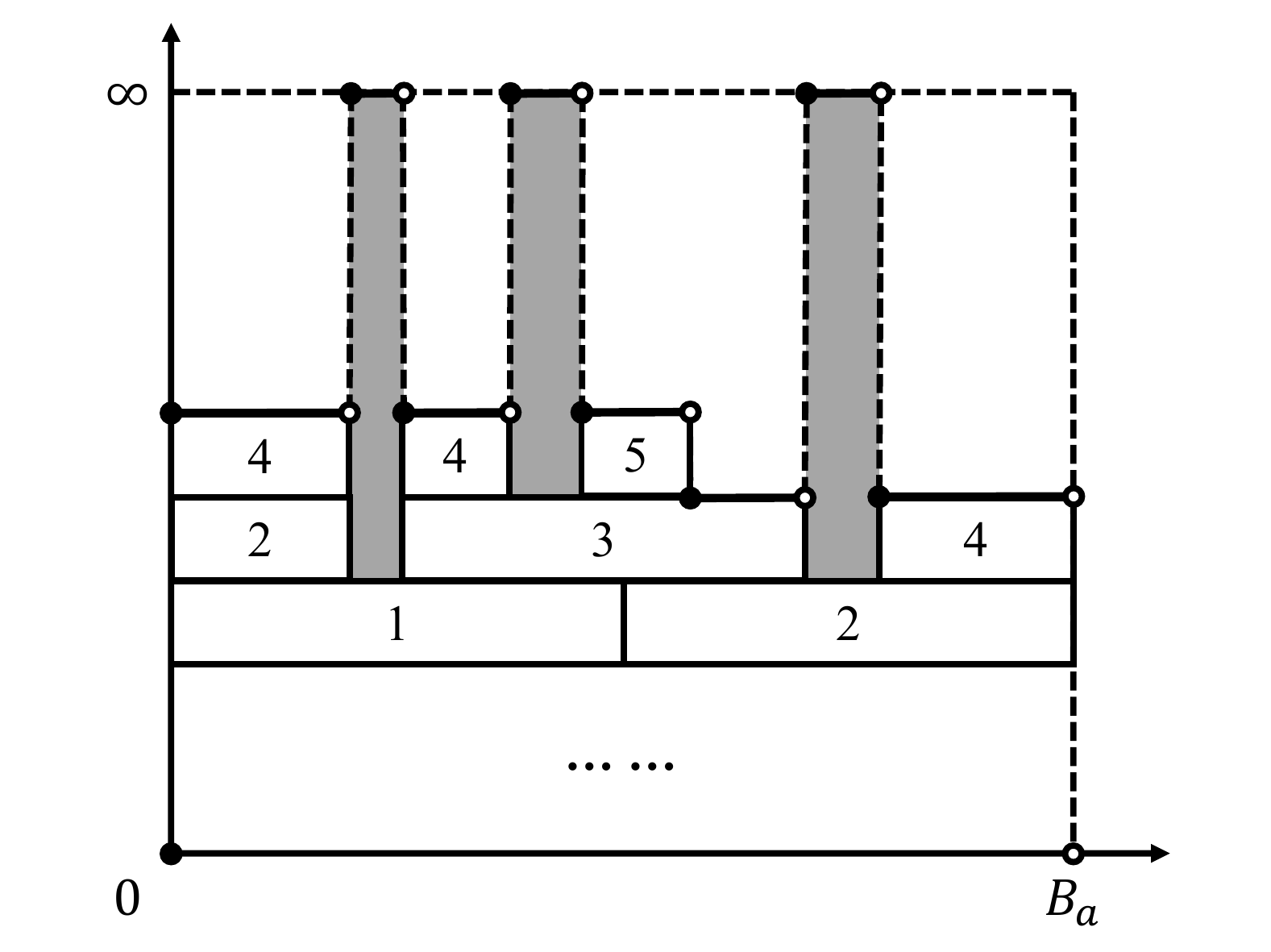}
        \caption{Semi-assignments to advertiser $a$}
    \end{subfigure}
    \begin{subfigure}{.48\textwidth}
        \includegraphics[width=\textwidth]{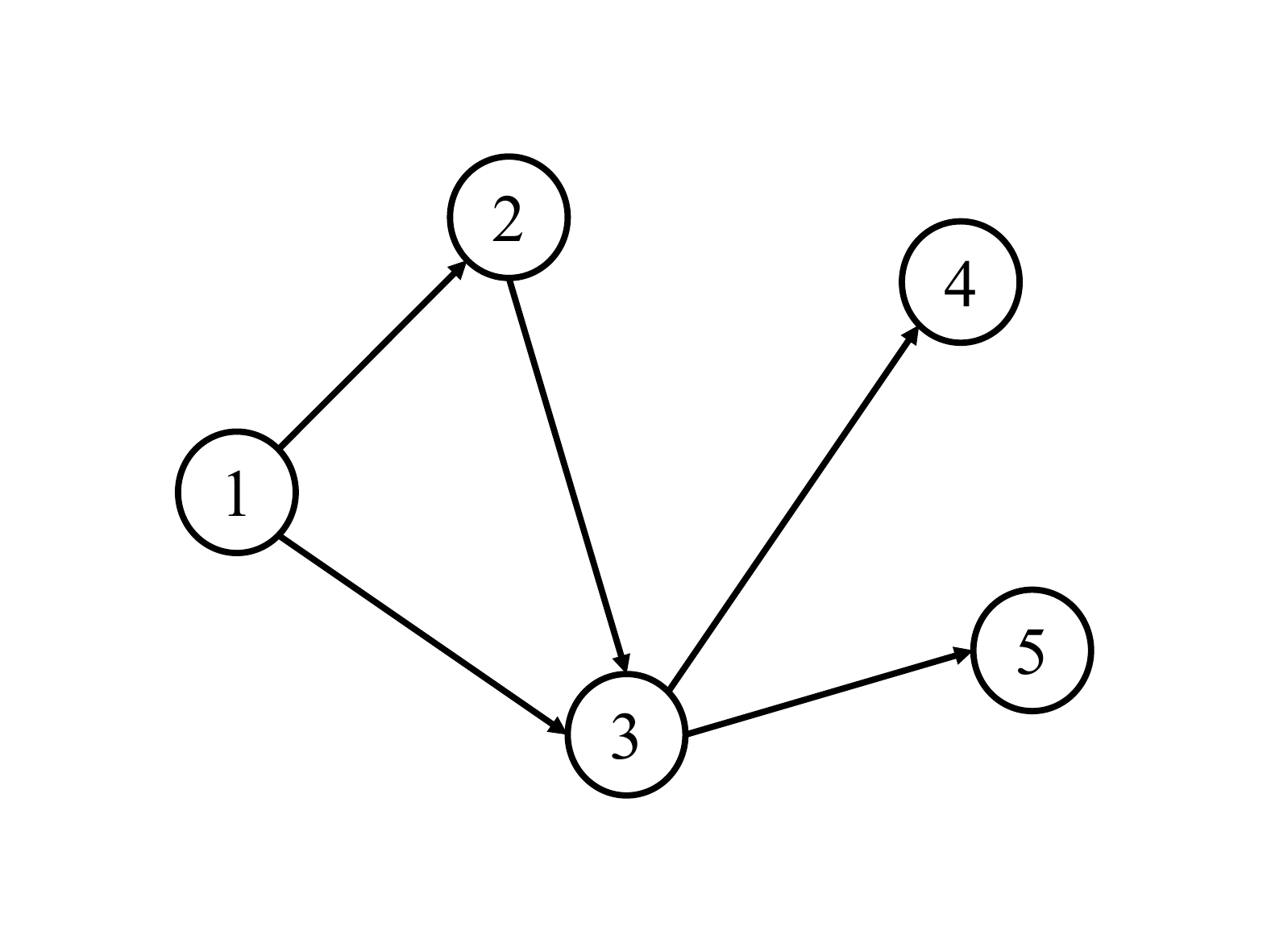}
        \caption{Arcs w.r.t.\ advertiser $a$}
    \end{subfigure}
    \caption{Example of \emph{ex-ante} dependence graph}
    \label{fig:dependence-graph}
\end{figure}

\paragraph{Ex-post Dependence Graph.}
The PanOCS algorithms in this paper follow the same recipe:
construct a random subgraph $D^*$ of $D$;
then, for every arc $(i, i')_a$ in the subgraph $D^*$, it correlates the selections in $i$ and $i'$ perfectly negatively.
We call $D^*$ the \emph{ex-post} dependence graph.
The subgraph $D^*$ must not introduce positive correlation among related nodes.
For instance, we cannot include both $(i, i')_a$ and $(i', i'')_a$ in $D^*$ if $i$ and $i''$ are related, or else the decisions in $i$ and $i''$ will be perfectly positively correlated.
We construct $D^*$ to be a random matching.
Hence, conditioned on any realization of $D^*$, any pair of nodes are either independent or perfectly negatively correlated.

Concretely, on the arrival of each node $i$, we pick an incident arc randomly each with probability $\frac{1}{8}$;%
\footnote{Even though the out-arcs have yet to reveal themselves, the PanOCS may reference them as the first and second out-arcs w.r.t\ each of the two chosen advertisers.}
an arc of $D$ is included in $D^*$ if both incident nodes pick it.
See Algorithm~\ref{alg:large-bid-panocs}.

\begin{algorithm}[t]
    \caption{Panoramic Online Correlated Selection (Large Bids, $\gamma = \frac{1}{64}$)}
    \label{alg:large-bid-panocs}
    \begin{algorithmic}
        \FORALL{impression $i \in I^R$}
            \STATE add arcs from $i'$ to $i$ to $D$ for every existing $i'$ that correlates with $i$
            \STATE randomly pick one of its at most $8$ incident arc in $D$, each with probability $\frac{1}{8}$
            \IF{it picks an in-arc, say, $(i', i)_a$, and $i'$ also picks the arc}
                \STATE add arc $(i', i)_a$ to $D^*$
                \STATE select $a$ if it does \emph{not} select $a$ for $i'$, and select $a'$ otherwise (i.e., the opposite selection)
            \ELSE
                \STATE select with a fresh random bit
            \ENDIF
        \ENDFOR
    \end{algorithmic}
\end{algorithm}

\begin{proof}[Proof of Theorem~\ref{thm:panocs-large-bid-weak}]
    Fix any advertiser $a$ and any point $y \in [0, B_a)$.
    Let $i_1 < i_2 < \dots < i_k$ be the subset of impressions that are semi-assigned to advertiser $a$ and subsets containing $y$.
    Suppose there exists an arc in $D^*$ between a pair of these nodes.
    Then, Algorithm~\ref{alg:large-bid-panocs} picks advertiser $a$ in exactly one of these two rounds.%
    \footnote{Importantly, this holds even if the arc was due to an advertiser other than $a$.}
    If no such arc exists, Algorithm~\ref{alg:large-bid-panocs} picks $a$ with probability half independently in each of the $k$ rounds;
    the probability that $a$ is never chosen is $2^{-k}$.
    Hence, it suffices to show that the probability that there is no arc in $D^*$ among $i_1$ to $i_k$ is at most $(1 - \frac{1}{64})^{k-1}$.

    In fact, we will show a slightly stronger result.
    For any $0 \le m \le k$, let $F_m$ denote the event that no arc of the form $(i_j, i_{j+1})_a$ exists among the first $m$ nodes.
    Let $f_m$ denotes its probability.
    Observe that when $m = k$, this is a necessary condition of the event we seek to analyze.
    Hence, it is sufficient to bound the probability of $F_k$.
    We claim that $f_m$ is recursively defined as follows:
    \begin{equation}
        \label{eqn:large-bid-panocs-recursion}
        f_0 = f_1 = 1, f_m = f_{m-1} - \frac{1}{64} f_{m-2}
        ~.
    \end{equation}

    Observe that $f_{m-2} \ge f_{m-1}$ because $F_{m-1}$ is a subevent of $F_{m-2}$.
    The desired bound follows by:
    \[
        f_m = f_{m-1} - \frac{1}{64} f_{m-2} \le \big( 1 - \frac{1}{64} \big) f_{m-1}
        ~.
    \]

    The base cases of Eqn.~\eqref{eqn:large-bid-panocs-recursion} are trivial.
    It remains to show the recurrence.
    To do so, we further introduce an auxiliary subevent $G_m$ for any $1 \le m \le k-1$, which not only requires $F_m$, but also that the last impression $i_m$ picks arc $(i_m, i_{m+1})_a$.
    Let $g_m$ denote the probability of the subevent.

    For the subevent to happen, the last impression $i_m$ must pick a specific arc $(i_m, i_{m+1})_a$, which happens with probability $\frac{1}{8}$.
    This choice of $i_m$ further ensures that arc $(i_{m-1}, i_m)$ cannot be realized.
    Therefore, the rest of the restriction is simply that no arc $(i_j, i_{j+1})_a$ is realized in the first $m-1$ rounds.
    Hence:
    \[
        g_m = \frac{1}{8} f_{m-1}
        ~.
    \]

    For the original event to happen, on the other hand, there are two cases.
    The first case is when $i_m$ picks arc $(i_{m-1}, i_m)_a$, which happens with probability $\frac{1}{8}$.
    In this case the rest of the restriction requires not only that no arc $(i_j, i_{j+1})_a$ is realized in the first $m-1$ rounds, but also that impression $i_{m-1}$ cannot pick arc $(i_{m-1}, i_m)_a$.
    In other word, it corresponds to the event w.r.t.\ the first $m-1$ rounds, \emph{excluding the subevent}.
    The other case is when $i_m$ picks one of the other $7$ arcs, which happens with probability $\frac{7}{8}$.
    In this case it once again reduces to having no arc $(i_j, i_{j+1})_a$ realized in the first $m-1$ rounds.
    Putting together we have:
    \[
        f_m = \frac{1}{8} \big(f_{m-1} - g_{m-1} \big) + \frac{7}{8} f_{m-1}
        ~.
    \]

    Cancelling $g_m$ using the previous equation proves Eqn.~\eqref{eqn:large-bid-panocs-recursion}.
\end{proof}

\subsection{Large Bids: Proof of Theorem~\ref{thm:panocs-large-bid}}
\label{sec:panocs-large-bid}

We further adopt the nomenclature from the research on OCS~\cite{Huang/arXiv/2019, HuangT/arXiv/2019}.
In the warmup PanOCS in the last subsection, whenever a node $i$ picks an out-arc, say, $(i, i')_a$, the PanOCS uses a fresh random bit to select an advertiser-subset combination in the round $i$;
further, the selection of round $i$ is ready to be received by node $i'$ should $i'$ also picks arc $(i, i')_a$.
To this end, we call node $i$ a \emph{sender} if it picks an out-arc, and call it a \emph{receiver} otherwise.
In the warmup algorithm, a sender's random bit is sent to one other node, which is received only if the other node chooses to be a receiver and further \emph{the randomly chosen in-arc happens to be the one from the sender}.

Next, we refine the italic part above to obtain the improved ratio in Theorem~\ref{thm:panocs-large-bid}.
We still let each node be a sender or a receiver randomly.
Further, each sender still sends its selection to a random out-neighbor.
Each receiver, however, proactively checks whether any in-neighbors are senders who pick its corresponding in-arcs.
If so, it randomly picks one such in-arc.
Finally, we optimize the probability of letting a node be a sender to obtain the ratio stated in Theorem~\ref{thm:panocs-large-bid}.

We further define the algorithm more generally for arbitrary bids, where the correlation occurs only among large ones.
The more general definition will be useful in the hybrid algorithm in the next section.
See Algorithm~\ref{alg:large-bid-panocs-improved}.

\begin{algorithm}[t]
    \caption{Panoramic Online Correlated Selection (Large Bids)}
    \label{alg:large-bid-panocs-improved}
    \begin{algorithmic}
        \smallskip
        \STATE \textbf{parameter:~} $0 \le p \le 1$, the probability of being a sender\\[1ex]
        \FORALL{impression $i \in I^R$}
            \STATE add arc $(i', i)_a$ to $D$ for any existing $i'$ adjacent to $i$ w.r.t.\ some advertiser $a$ s.t.\ \emph{$b_{ai}, b_{ai'} > \frac{B_a}{2}$}
            \medskip
            \STATE $i$ is a \textbf{sender} with probability $p$:
                \INDSTATE select with a fresh random bit
                \INDSTATE randomly pick an out-arc in $D$, each with probability $\frac{1}{4}$
            \medskip
            \STATE $i$ is a \textbf{receiver} with probability $1-p$:
                \INDSTATE \textbf{if} some sender $i'$ picks an in-arc $(i', i)_a$ of $i$ in $D$ \textbf{then}
                    \INDINDSTATE randomly pick such an in-arc $(i', i)_a$ and add it to $D^*$
                    \INDINDSTATE select $a$ if it does \emph{not} select $a$ for $i'$, and vice versa (i.e., the opposite selection)
                \INDSTATE \textbf{else}
                    \INDINDSTATE select with a fresh random bit
            \INDSTATE \textbf{end if}
        \ENDFOR
    \end{algorithmic}
\end{algorithm}

The next lemma analyzes Algorithm~\ref{alg:large-bid-panocs-improved} for any choice of sender probability $0 < p < 1$.
Theorem~\ref{thm:panocs-large-bid}, i.e., $\gamma = 0.05144$, is a corollary with the optimal $p = \frac{4}{9}$. 

\begin{lemma}
    \label{lem:panocs-large-bid}
    If all nonzero bids are large, Algorithm~\ref{alg:large-bid-panocs-improved} is a $\gamma$-PanOCS for:
    \begin{equation}
        \label{eqn:panocs-large-bid-gamma}
        \gamma = \frac{1}{4} \big(1-p\big) p \big(1-\frac{3p}{8} \big)
        ~.
    \end{equation}
\end{lemma}

Recall the intuition behind a $\gamma$-PanOCS:
any two adjacent impressions (which correspond to two neighboring nodes in $D$) are perfectly negatively correlated (which correspond to being in $D^*$) with probability $\gamma$;
moreover, the events are as least as good as independent.
Before diving into a formal proof of Lemma~\ref{lem:panocs-large-bid}, we explain the intuition why the marginal probability of realizing an arc $(i', i)_a$ in $D^*$ may be the above value of $\gamma$.
For an arc to be in $D^*$, $i$ must be a receiver, which happens with probability $1 - p$, and must pick arc $(i', i)_a$, which happens with probability $\frac{1}{4}$.
Further, $i'$ must be a receiver, which happens with probability $p$.
Conditioned on all of the above, what is the chance that $(i', i)_a$ is in $D^*$?
Each in-neighbor is a sender with probability $p$, picks the in-arc of $i$ with probability $\frac{1}{4}$, and finally wins over $(i', i)_a$ with probability $\frac{1}{2}$.
Hence, the three in-arcs of $i$ other than $(i', i)_a$ together prevent $i$ from choosing $(i', i)_a$ with probability at most $\frac{3p}{8}$.
This overestimation of the failure probability gives the stated value of $\gamma$.

We will in fact prove the following lemma that further applies to the more general case with a mixture of large and small bids.
Lemma~\ref{lem:panocs-large-bid} follows as a direct corollary.
The more general lemma will be useful in the hybrid algorithm in the next section.
For any advertiser $a$ and any any point $y \in [0, B_a)$, let $k_a^L(y) \le k_a(y)$ denote the number of impressions semi-assigned to advertiser $a$ and subsets containing $y$ whose bids are at least $\frac{B_a}{2}$, before the first time $y$ is semi-assigned to a small bid, i.e., smaller than $\frac{B_a}{2}$.
If $y$ has never been semi-assigned to a small bid, let $k_a^L(y) = k_a(y)$.

\begin{lemma}
    \label{lem:panocs-large-bid-general}
    For the $\gamma$ in Eqn.~\eqref{eqn:panocs-large-bid-gamma}, Algorithm~\ref{alg:large-bid-panocs-improved} satisfies that for any advertiser $a$ and any point $y \in [0, B_a)$, $y$ is assigned at least once with probability at least:
    \[
        1 - 2^{-k_a(y)} \cdot (1-\gamma)^{\max\{k_a^L(y)-1, 0\}}
        ~.
    \]
\end{lemma}

\begin{proof}
    Fix any advertiser $a$ and any point $y \in [0, B_a)$.
    Let $i_1 < i_2 < \dots < i_{k_a(y)}$ be the impressions semi-assigned to advertiser $a$ and subsets containing $y$.
    If there is an arc in the \emph{ex-post} dependence graph $D^*$ between two of them, point $y$ is assigned in exactly one of the two rounds.
    Otherwise, each of the $k$ rounds independently has probability half of assigning $y$;
    there is only a $2^{-k_a(y)}$ probability that $y$ is never assigned.
    Hence, it remains to analyze the former event and show that the probability of having no arc among $i_1 < i_2 < \dots < i_k$ in the \emph{ex-post} dependence graph $D^*$ is at most $(1 - \gamma)^{k_a^L(y)-1}$ for the stated value of $\gamma$ in the lemma.
    We will upper bound this by the probability of having no arcs among the first $k_a^L(y)$ impressions $i_1 < i_2 < \dots < i_{k_a^L(y)}$.

Concretely, for any $0 \le m \le k_a^L(y)$, let $F_m$ denote the event that there is no arc in $D^*$ among $i_\ell$, $1 \le \ell \le m$.
    Let $f_m$ denote the probability of $F_m$.
    Trivially we have $f_0 = f_1 = 1$.
    
    Next, we inductively derive the following upper bound of $f_m$ for $2 \le m \le k_a^L(y)$:
    \[
        f_m \le (1 - \gamma) \cdot f_{m-1}
        ~.
    \]

    To do so, further consider an auxiliary subevent $G_m$ of $F_m$ for any $1 \le m \le k_a^L(y)-1$, which further requires that $i_m$ is a sender who picks arc $(i_m, i_{m+1})_a$.
    Let $g_m$ be the probability of $G_m$.

    \paragraph{Auxiliary Event.}
    In order to have event $G_m$, we need:
    \begin{enumerate}
        \item Node $i_m$ is a sender (probability $p$);
        \item Node $i_m$ picks $(i_m, i_{m+1})_a$ (probability $\frac{1}{4}$); and
        \item Event $F_{m-1}$ (probability $f_{m-1})$.
    \end{enumerate}
    These conditions are independent since they rely on disjoint subsets of random bits.
    We have:
    \begin{equation}
        \label{eqn:panocs-large-bid-auxiliary}
        g_m = \frac{p}{4} f_{m-1}
        ~.
    \end{equation}

    \paragraph{Main Event.}
    Next we turn to event $F_m$.
    There are two subcases: $i_m$ is a sender, or a receiver.

    \paragraph{Case 1: Sender.}
    If $i_m$ is a sender, $F_m$ cannot fail due to a pair of nodes including $i_m$.
    Hence, it remains to ensure $F_{m-1}$.
    The contribution of this case to the probability of $F_m$ is:
    \begin{equation}
        \label{eqn:panocs-large-bid-1}
        p f_{m-1}
    \end{equation}

    \paragraph{Case 2: Receiver.}
    In this case, we need to further ensure that $F_m$ does not fail due to a pair of nodes including $i_m$.
    There are at most $4$ in-arcs of concern.
    First, there is always an arc $(i_{m-1}, i_m)_a$ in $D$.
    Further, there may some other arcs, either from $i_{m-2}$ to $i_m$ w.r.t.\ advertiser $a$, or from some $i_\ell$, $1 \le \ell \le m-1$ w.r.t.\ other advertisers.
    Let $n \le 3$ be the number of in-arcs of the latter form.
    The binding case of the analysis is when $n = 0$, which we shall demonstrate first.

    \paragraph{Case 2a: $n = 0$.}
    The only arc incident to $i_m$ of concern is $(i_{m-1}, i_m)_a$.
    As a result, it is sufficient (but not necessary in general) if the first $m-1$ rounds are in $F_{m-1} \setminus G_{m-1}$.
    This part contributes:
    \begin{equation}
        \label{eqn:panocs-large-bid-1a-1}
        (1-p) \big( f_{m-1} - g_{m-1} \big)
        ~.
    \end{equation}

    Even if the first $m-1$ rounds are in $G_{m-1}$, there may not be a perfect negative correlation between the selections of impressions $i_{m-1}$ and $i_m$, due to the competition from the other in-arcs of $i_m$ in $D$.
    Concretely, there are $3$ in-arcs other than $(i_{m-1}, i_m)_a$.
    Moreover, by the assumption of the subcase, these in-arcs are not from $i_\ell$ for any $1 \le \ell \le m-2$.

    Suppose the in-arcs are further from $3$ distinct nodes.
    Then, each is a sender who picks $i_m$ with probability $\frac{p}{4}$ independently.
    The probability that $i_m$ picks $i_{m-1}$ equals:
    \begin{align*}
        \sum_{i=0}^{3} \frac{1}{i+1} \bigg(\frac{p}{4}\bigg)^i \bigg(1-\frac{p}{4}\bigg)^{3-i} \binom{3}{i}
        & 
        = \sum_{i=0}^{3} \frac{1}{4} \bigg(\frac{p}{4}\bigg)^i \bigg(1-\frac{p}{4}\bigg)^{3-i} \binom{4}{i+1} \\
        &
        = \frac{1}{p} \bigg( 1 - \big(1 - \frac{p}{4}\big)^{4} \bigg)
        ~.
    \end{align*}

    Through a similar calculation, we further conclude that if there are two in-arcs from the same node, the probability that $i_m$ picks $i_{m-1}$ is $1 - \frac{3}{8} p + \frac{1}{24} p^2$.
    The bounds in both cases are greater than $1 - \frac{3}{8} p$ for any $0 \le p \le 1$.

    Hence, the contribution of this part to the probability of $F_m$ is at most:
    \begin{equation}
        \label{eqn:panocs-large-bid-1a-2}
        (1-p) \bigg( 1 - \big( 1 - \frac{3p}{8} \big) \bigg) g_{m-1}
        ~.
    \end{equation}

    Summing up the contributions in Equations~\eqref{eqn:panocs-large-bid-1}, \eqref{eqn:panocs-large-bid-1a-1}, and \eqref{eqn:panocs-large-bid-1a-2}, cancelling $g_{m-1}$ using Eqn.~\eqref{eqn:panocs-large-bid-auxiliary}, and plugging in the value of $\gamma$, we have:
    \[
        f_m \le f_{m-1} - \gamma \cdot f_{m-2} \le (1 - \gamma) f_{m-1}
        ~.
    \]

    \paragraph{Case 2b: $1 \le n \le 3$.}
    We will show that the probability of $F_m$ in this case is upper bounded by the previous case.
    In this case, having the first $m-1$ rounds in $F_{m-1} \setminus G_{m-1}$ is no longer a sufficient condition since there may still be an arc from $i_\ell$, $1 \le \ell \le m-2$ to $i_m$.
    Nonetheless, Eqn.~\eqref{eqn:panocs-large-bid-1} is still an upper bound of the contribution of this part.

    Next, consider the case when the first $m-1$ rounds are in $G_{m-1}$.
    Comparing with the previous case, the main difference is that $i_m$ may have in-arcs in $D$ other than $(i_{m-1}, i_m)_a$ from some $i_\ell$, $1 \le \ell \le m-1$.
    However, should $i_m$ pick one of these in-arcs, we still get that $y$ has been selected.
    Intuitively, it is less likely than the previous case that $i_m$ picks an in-arc \emph{not} from $i_\ell$, $1 \le \ell \le m-1$.

    Formally, fix any realization of randomness in the first $m-1$ rounds in $G_{m-1}$.
    Suppose $n' \le n$ in-neighbors of $i_m$ of the form $i_\ell$, $1 \le \ell \le m-2$, are senders who pick the arc to $i_m$.
    If each of the $3-n$ other in-arcs from from distinct nodes, each of the $3-n$ other in-neighbors is a sender who picks the arc to $i_m$ independently with probability $\frac{p}{4}$.
    The probability that $i_m$ picks one of the in-arcs from in-neighbors of the form $i_\ell$ is:
    \begin{align*}
        \sum_{i=0}^{3-n} \frac{1+n'}{i+n'+1} \bigg(\frac{p}{4}\bigg)^i \bigg(1-\frac{p}{4}\bigg)^{3-n-i} \binom{3-n}{i}
        &
        \ge \sum_{i=0}^{3-n} \frac{1}{i+1} \bigg(\frac{p}{4}\bigg)^i \bigg(1-\frac{p}{4}\bigg)^{3-n-i} \binom{3-n}{i} \\
        &
        = \sum_{i=0}^{3-n} \frac{1}{4-n} \bigg(\frac{p}{4}\bigg)^i \bigg(1-\frac{p}{4}\bigg)^{3-n-i} \binom{4-n}{i+1} \\
        &
        = \frac{4}{(4-n)p} \bigg( 1 - \big(1-\frac{p}{4}\big)^{4-n} \bigg)
        ~.
    \end{align*}

    Observe that $x^{-1} \big( 1 - (1-\frac{p}{4})^x \big)$ is decreasing in $x$, the above is greater than the bound in the previous case which corresponds to $n = 0$.

    Finally, if there are two of the $3-n$ other in-arcs are from the same neighbor, it must be $n = 1$.
    This in-neighbor is a sender who picks one of these two arcs with probability $\frac{p}{2}$.
    We have a similar calculation for the probability that $i_m$ picks one of the in-arcs from in-neighbors of the form $i_\ell$:
    \[
        \frac{p}{2} \cdot \frac{1+n'}{2+n'} + \big(1 - \frac{p}{2}\big) \ge \frac{p}{2} \cdot \frac{1}{2} + \big(1 - \frac{p}{2}\big) = 1 - \frac{p}{4} 
        ~,
    \]
    which is also greater than the $1 - \frac{3p}{8}$ bound needed in the analysis.
\end{proof}

\subsection{General Bids: Proof of Theorem~\ref{thm:panocs-general-bid}}
\label{sec:panocs-general-bid}

\paragraph{Challenge.}
In the presence of both large and small bids, a semi-assignment of an impression $i$ to an advertiser-subset combination $(a, Y_{ai})$ may be adjacent to an arbitrary number of subsequent semi-assignments of smaller bids.
For instance, consider an impression with a large bid $b_{ai} = B_a$ followed by $n$ impressions with small bids $b_{ai'} = \frac{B_a}{n}$.
Therefore, the previous approach of letting each impression randomly picks an out-arc in the \emph{ex-ante} dependence graph $D$ no longer works because the probability that an arc $(i, i')_a$ in the \emph{ex-ante} dependence graph $D$ is included in the \emph{ex-post} dependence graph $D^*$ may be arbitrarily small.

\paragraph{Solution.}
For any advertiser $a$, we will partition the impressions semi-assigned to $a$ with small bids into groups.
Then, define group-level correlation similar to the case of large bids, treating the union of impressions within the same group as a single large bid.
We will argue that each group is correlated with a bounded number of other groups.
Finally, recall that any impression in a randomized round is associated with two advertisers and thus, belongs to two groups, one for each advertiser.
Pick one of the two groups randomly and follow its decision.

It is worth remarking that the impressions in the same group are positively correlated.
Hence, two impressions in the same group, say, w.r.t.\ an advertiser $a$ must be unrelated.
In other words, other than the common advertiser $a$, the two impressions are semi-assigned to either two distinct advertisers, or the same advertiser but disjoint subsets.

\bigskip

The rest of the subsection will substantiate the above intuition with a formal definition of the algorithm and its analysis.
We will build on the notions of two nodes' being related, adjacent, and unrelated (Definition~\ref{def:randomized-round-correlation}).
\paragraph{First-level Partition.}
For each advertiser $a$, let $I^R_a$ denote the set of randomized round which semi-assign the impressions to $a$.
We shall greedily partition $I^R_a$ into subsets of impressions that are pairwise unrelated w.r.t.\ advertiser $a$, denoted as $I_a^j$, $j \ge 1$.
\begin{enumerate}
    \item Initialize $j_a = 1$ and $I_a^j = \emptyset$ for any $j \ge 1$.
    \item For each impression $i \in I^R_a$:
    \begin{enumerate}
        \item Let $j_a \leftarrow j_a + 1$ if impression $i$ is adjacent to the first impression in $I_a^{j_a}$ w.r.t.\ $a$.
        \item Let $I_a^{j_a} \leftarrow I_a^{j_a} \cup \{ i \}$.
    \end{enumerate}
\end{enumerate}

Next we establish several properties of the above greedy partition.
By the above definition and the panoramic interval-level assignment in Section~\ref{sec:panorama}, we have:

\begin{lemma}
    \label{lem:same-block-unconstrained}
    Any two impressions in the same subset $I_a^j$ are unrelated w.r.t.\ advertiser $a$.
\end{lemma}

The next property shows that the subsets restore the main structural property of large bids, i.e., Lemma~\ref{lem:large-bid-bounded-degree}.
Indeed, if all bids were large, each impression would be a subset on its own.

\begin{lemma}
    \label{lem:two-cover-all}
    For any neighboring subsets $I_a^j$ and $I_a^{j+1}$, and any point $y \in [0, B_a)$ which is \emph{not} deterministically assigned, $y$ is semi-assigned at least once in the rounds in $I_a^j$ and $I_a^{j+1}$.
\end{lemma}

\begin{proof}
    In fact, we will show a stronger claim that the rounds in subset $I_a^{j}$ and the \emph{first round} in $I_a^{j+1}$ suffice.
    By the definition of the greedy partition algorithm, the subset of $[0, B_a)$ chosen in the first round in $I_a^{j+1}$ intersects with the subset of the first round in $I_a^{j}$;
    otherwise, it would have been added to $I_a^{j}$ instead.
    Hence, the points in this intersection have already been semi-assigned twice.
    Finally, by the panoramic interval-level assignment in Section~\ref{sec:panorama}, any point $y \in [0, B_a)$ that has not been deterministically assigned thus far must have been semi-assigned at least once.
\end{proof}

This lemma has two direct corollaries.

\begin{corollary}
    \label{cor:setset-num}
    For any advertiser $a$ there are at most $2 \kmax$ nonempty subsets $I_a^j$.
\end{corollary}

\begin{corollary}
    \label{cor:plus-minus-two-correlation}
    Suppose two impressions belong to subsets $I_a^j$ and $I_a^k$ respectively, where $k > j + 2$.
    Then, they are \emph{not} adjacent w.r.t.\ $a$.
\end{corollary}

\paragraph{Second-level Partition.}
An pair of impressions in the same subset $I_a^j$ could still be related or even adjacent w.r.t.\ an advertiser other than $a$.
To resolve this we further introduce another layer of partition of each $I_a^j$ into $\cup_k I_a^{j,k}$ as follows:
\begin{equation}
    \label{eqn:panocs-partition}
    \forall k \in \mathbb{Z}_+ : \quad
    I_a^{j,k} = I_a^j \cap \bigg( \bigcup_{a' \ne a} I_{a'}^k \bigg)
    ~.
\end{equation}

In other words, an impression $i$ is in $I_a^{j,k}$ if it belongs to the $j$-th subset $I_a^j$ w.r.t.\ advertiser $a$ in the first-level partition, and further belongs to the $k$-th subset $I_{a'}^k$ of the other advertiser $a'$ to which $i$ is semi-assigned.
We shall refer to each $I_a^{j,k}$ as a \emph{group} of impressions.

As a corollary of Lemma~\ref{lem:same-block-unconstrained}, we have:

\begin{corollary}
    \label{cor:panocs-partition-unconstrained}
    Any two impressions in the same group $I_a^{j,k}$ are unrelated.
\end{corollary}

\paragraph{Group-level Decision.}
Fix any advertiser $a \in A$.
We say that two groups $I_a^{j,k}$ and $I_a^{j',k'}$ are adjacent if there exist impressions $i \in I_a^{j,k}$ and $i' \in I_a^{j',k'}$ such that $i$ and $i'$ are adjacent w.r.t.\ $a$.
Then, as a further corollary of Lemma~\ref{lem:same-block-unconstrained}, Corollary~\ref{cor:setset-num}, and Corollary~\ref{cor:plus-minus-two-correlation}, we have:

\begin{corollary}
    \label{cor:group-level-degree-bound}
    Any group $I_a^{j,k}$ is adjacent to at most $8 \kmax$ other groups $I_a^{j',k'}$, whose superscripts are $j'\in\{j-2, j-1, j+1, j+2\}$ and $1 \le k' \le 2 \kmax$.
\end{corollary}

The group-level (negative) correlation is achieved using the following algorithm similar to the PanOCS for large bids, treating each group as a large bid.
For each group $I_a^{j,k}$, it returns either $a$ or $\neg a$ with 50-50 marginal probability.

Concretely, define an \emph{ex-ante} dependence graph $D_a$ for each advertiser $a$.
Let there be a node for each group $I_a^{j,k}$.
Further, let there be an arc from $I_a^{j,k}$ to $I_a^{j',k'}$ if they are adjacent, and $j < j'$.
The second condition indicates that arcs are from earlier groups to later ones.
The algorithm constructs an \emph{ex-post} dependence graph $D_a^*$ similar to the PanOCS for large bids.
It is parameterized by $0 < p < 1$, the probability of letting each group be a sender.
For each group:
\begin{enumerate}
    \item With probability $p$, let it be a \emph{sender}:
    \begin{enumerate}
        \item Pick a subsequent adjacent group $I_a^{j',k'}$, $j'\in\{j+1, j+2\}$, $1 \le k' \le 2 \kmax$, randomly.
        \item Return $a$ or $\neg a$ uniformly at random with a fresh random bit.
    \end{enumerate}
    \item Otherwise, let it be a \emph{receiver}:
    \begin{enumerate}
        \item If there exists a previous adjacent group $I_a^{j',k'}$ which is a sender and picks $I_a^{j,k}$, makes the opposite decision, i.e., return $a$ if group $I_a^{j',k'}$ returns $\neg a$, and vice versa.
        \item Otherwise, return $a$ or $\neg a$ uniformly at random with a fresh random bit.
    \end{enumerate}
\end{enumerate}

\paragraph{Impression-level Decision.}
Recall that the impression of each randomized round is associated with two advertisers $a$ and $a'$ and thus, two corresponding groups $I_a^{j,k}$ and $I_{a'}^{j',k'}$.
Follow the decision of one of the groups, chosen uniformly at random.
By following the decision of a group, say, $I_a^{j,k}$, the PanOCS picks $a$ if the group picks $a$, and picks $a'$ if the group picks $\neg a$.

\bigskip

The PanOCS for general bids is summarized in Algorithm~\ref{alg:panocs-weak}.
We next show a general analysis of the algorithm for any value of $0 < p < 1$.

\begin{algorithm}[t]
    \caption{Panoramic Online Correlated Selection (General Bids, Parameter $0 < p < 1$)}
    \label{alg:panocs-weak}
    \begin{algorithmic}
        \STATE initialize $j_a = 1$, $I_a^j = \emptyset$, and $I_a^{j,k} = \emptyset$ for any $a \in A$, any $j \ge 1$, and any $1 \le k \le 2 \kmax$.\\[1ex]
        \FORALL{group $I_a^{j,k}$}
            \smallskip
            \STATE \textbf{\emph{\# group-level decision}}
            \smallskip
            \STATE with probability $p$, let it be a \textbf{sender}:
                \INDSTATE let its decision be $a$ or $\neg a$ with a fresh random bit
                \INDSTATE randomly pick an $I_a^{j',k'}$, $j' \in \{j+1, j+2\}$, $1 \le k' \le 2 \kmax$, as the potential receiver
            \smallskip
            \STATE otherwise, i.e., with probability $1 - p$, let it be a \textbf{receiver}:
                \INDSTATE randomly pick a sender $I_a^{j',k'}$, $j' \in \{j-1, j-2\}$, $1 \le k' \le 2 \kmax$, who picks $I_a^{j,k}$
                \INDSTATE let $I_a^{j,k}$'s decision be $a$ if $I_a^{j',k'}$'s decision is $\neg a$, and vise versa
                \INDSTATE if no such $I_a^{j',k'}$ exists, let its decision be $a$ or $\neg a$ with a fresh random bit
        \ENDFOR
        \medskip
        \FORALL{impression $i \in I^R$, say, semi-assigned to $(a, S)$ and $(a', S')$}
            \medskip
            \STATE \textbf{\emph{\# 1st-level partition}}
            \STATE let $j_a \leftarrow j_a + 1$ if $i$ is adjacent to the first impression in $I_a^{j_a}$ w.r.t.\ $a$
            \STATE let $j_{a'} \leftarrow j_{a'} + 1$ if $i$ is adjacent to the first impression in $I_{a'}^{j_{a'}}$ w.r.t.\ $a'$
            \STATE let $I_a^{j_a} \leftarrow I_a^{j_a} \cup \{i\}$, and $I_{a'}^{j_{a'}} \leftarrow I_{a'}^{j_{a'}} \cup \{i\}$
            \medskip
            \STATE \textbf{\emph{\# 2nd-level partition}}
            \STATE let $I_a^{j_a,j_{a'}} \leftarrow I_a^{j_a,j_{a'}} \cup \{i\}$ and $I_{a'}^{j_{a'},j_a} \leftarrow I_{a'}^{j_{a'},j_a} \cup \{i\}$
            \medskip
            \STATE \textbf{\emph{\# impression-level decision}}
            \STATE follow either $I_a^{j_a,j_{a'}}$ or $I_{a'}^{j_{a'},j_a}$'s decision, each with probability half
            \medskip
        \ENDFOR
    \end{algorithmic}
\end{algorithm}

\begin{lemma}
    Algorithm~\ref{alg:panocs-weak} is a $\gamma$-PanOCS, where:
    \[
        \gamma = \frac{1}{16\kmax} \big(1-p\big) \left( 1 - \big(1 - \frac{p}{4\kmax} \big)^{4\kmax} \right)
        ~.
    \]
\end{lemma}

Then, Theorem~\ref{thm:panocs-general-bid} follows as a corollary, observing the stated value of $\gamma$ is at least:
\[
    \frac{1}{16\kmax} (1-p) (1 - e^{-p})
    ~,
\]
and letting $p = 2 - W(e^2) \approx 0.44285$ to maximize it, where $W(\cdot)$ is the product logarithm.

\begin{proof}
    Fix any advertiser $a$ and any point $y \in [0, B_a)$.
    Let $i_1 < i_2 < \dots < i_k$ be the impressions semi-assigned to advertiser $a$ and subsets containing $y$.
    Let $j_1 < j_2 < \dots < j_k$ be the subsets such that $i_\ell \in I_a^{j_\ell}$.
    Recall that each subset $I_a^{j_\ell}$ is further partitioned into groups and the impression $i_\ell$ belongs to exactly one group.
    Nonetheless, the index in the second-level partition is unimportant for the argument;
    we write the group as $I_a^{j_\ell,*}$.

    Suppose there exists an arc in the \emph{ex-post} dependence graph $D_a^*$ between two of $I_a^{j_\ell,*}$, $1 \le \ell \le k$, and further the PanOCS chooses to follows $I_a^{j_\ell,*}$'s decision in these two rounds.
    Then, point $y$ is assigned in exactly one of the two rounds.
    Otherwise, each of the $k$ rounds independently has probability half of assigning $y$.
    Hence, it remains to analyze the former event and upper bound the probability that it does \emph{not} happen by $(1 - \gamma)^{k-1}$, for the stated value of $\gamma$ in the lemma.

    Concretely, for any $0 \le m \le k$, let $F_m$ denote the event that for any two nodes $I_a^{j_\ell,*}$, $1 \le \ell \le m$, either there is no arc in $D_a^*$ between them, or the PanOCS does not follow their decisions in at least one of the two rounds.
    Let $f_m$ denote the probability of $F_m$.
    Trivially we have $f_0 = f_1 = 1$.
    
    Next, we inductively derive the following upper bounds of $f_m$ for $2 \le m \le k$:
    \[
        f_m \le (1 - \gamma) \cdot f_{m-1}
        ~.
    \]

To do so, further consider an auxiliary subevent $G_m$ of $F_m$ for any $1 \le m \le k-1$, which further requires that $I_a^{j_m,*}$ is a sender who picks $I_a^{j_{m+1},*}$ and the PanOCS follows $I_a^{j_m,*}$'s decision.
    Let $g_m$ be the probability of $G_m$.

    \paragraph{Auxiliary Event.}
    In order to have event $G_m$, we need:
    \begin{enumerate}
        \item The PanOCS follows $I_a^{j_m,*}$'s decision (probability $\frac{1}{2}$);
        \item Node $I_a^{j_m,*}$ is a sender (probability $p$);
        \item Node $I_a^{j_m,*}$ picks $I_a^{j_{m+1},*}$ (probability $\frac{1}{4\kmax}$); and
        \item Event $F_{m-1}$ (probability $f_{m-1})$.
    \end{enumerate}

    Further observe that these conditions are independent since they rely on disjoint subsets of random bits.
    We have:
    \begin{equation}
        \label{eqn:panocs-general-bid-auxiliary}
        g_m = \frac{p}{8\kmax} f_{m-1}
        ~.
    \end{equation}

    \paragraph{Main Event.}
    Next we turn to event $F_m$.
    There are two cases depending on the conditions below, which are independent because they rely on disjoint subsets of random bits.
    \begin{enumerate}
        \item The PanOCS follows $I_a^{j_m,*}$'s decision (probability $\frac{1}{2}$);
        \item Node $I_a^{j_m,*}$ is a receiver (probability $1-p$).
    \end{enumerate}

    \paragraph{Case 1: Last Node Does Not Matter.}
    If at least one of the above conditions do not hold, $F_m$ cannot fail due to a pair of nodes including $I_a^{j_m,*}$.
    Hence, it remains to ensure $F_{m-1}$.
    The contribution of this case to the probability of $F_m$ is:
    \begin{equation}
        \label{eqn:panocs-general-last-node-not-matter}
        \big( 1 - \frac{1-p}{2} \big) f_{m-1}
    \end{equation}

    \paragraph{Case 2: Last Node Matters.}
    If both conditions hold, we need to further ensure that $F_m$ does not fail due to a pair of nodes including $I_a^{j_m,*}$.
    There are two pairs of nodes of concern.
    First, there is always an arc from $I_a^{j_{m-1},*}$ to $I_a^{j_m,*}$ in $D_a$.
    We need either that the arc is \emph{not} realized in $D_a^*$, or that the PanOCS does not follows $I_a^{j_{m-1},*}$'s decision.
    Here, recall the assumption that the PanOCS does follow $I_a^{j_m,*}$'s decision.
    The second pair of nodes is $I_a^{j_{m-2},*}$ and $I_a^{j_m,*}$.
    There may or may not be an arc between them in the \emph{ex-ante} dependence graph $D_a$, depending on whether $j_m = j_{m-2} + 2$.
    This further divides the rest of the analysis into two subcases

    \paragraph{Case 2a.}
    The first subcase is when $j_m > j_{m-2} + 2$, which is also the bottleneck of the analysis.
    Then, there cannot be an arc from $I_a^{j_{m-2},*}$ to $I_a^{j_m,*}$ in the dependence graphs.
    The only arc of concern is from $I_a^{j_{m-1},*}$ to $I_a^{j_m,*}$.
    As an immediate implication, it is sufficient (but not necessary in general) if the first $m-1$ rounds are in $F_{m-1} \setminus G_{m-1}$.
    The contribution of this part is:
    \begin{equation}
        \label{eqn:panocs-general-last-node-matter-a1}
        \frac{1-p}{2} \big( f_{m-1} - g_{m-1} \big)
        ~.
    \end{equation}

    Even if the first $m-1$ rounds are in $G_{m-1}$, there may not be a perfect negative correlation between the decisions of impressions $i_{m-1}$ and $i_m$, due to the competition from other in-neighbors of $I_a^{j_m,*}$ in $D_a$.
    Concretely, there are $4\kmax-1$ in-neighbors other than $I_a^{j_{m-1},*}$.
    Further, by the assumption of the subcase, these in-neighbors are not $I_a^{j_\ell}$ for any $1 \le \ell \le m-2$
    Hence, each of them independently has probability $\frac{p}{4\kmax}$ of being a sender who picks $I_a^{j_m,*}$.
    The probability that $I_a^{j_m,*}$ picks $I_a^{j_{m-1},*}$ instead of one of these competitors is equal to:
    \begin{align*}
        \sum_{i=0}^{4\kmax-1} & \frac{1}{i+1} \bigg(\frac{p}{4\kmax}\bigg)^i \bigg(1-\frac{p}{4\kmax}\bigg)^{4\kmax-1-i} \binom{4\kmax-1}{i} \\
        & 
        = \sum_{i=0}^{4\kmax-1} \frac{1}{4\kmax} \bigg(\frac{p}{4\kmax}\bigg)^i \bigg(1-\frac{p}{4\kmax}\bigg)^{4\kmax-1-i} \binom{4\kmax}{i+1} \\
        &
        = \frac{1}{p} \bigg( 1 - \big(1 - \frac{p}{4\kmax}\big)^{4\kmax} \bigg)
        ~.
    \end{align*}

    Hence, the contribution of this part to the probability of $F_m$ is:
    \begin{equation}
        \label{eqn:panocs-general-last-node-matter-a2}
        \frac{1-p}{2} \bigg( 1 - \frac{1}{p} \bigg( 1 - \big(1 - \frac{p}{4\kmax}\big)^{4\kmax} \bigg) \bigg) g_{m-1}
        ~.
    \end{equation}

    Summing up the contributions in Equations~\eqref{eqn:panocs-general-last-node-not-matter}, \eqref{eqn:panocs-general-last-node-matter-a1}, and \eqref{eqn:panocs-general-last-node-matter-a2}, cancelling $g_{m-1}$ using Eqn.~\eqref{eqn:panocs-general-bid-auxiliary}, and plugging in the value of $\gamma$, we have:
    \[
        f_m = f_{m-1} - \gamma \cdot f_{m-2} \le (1 - \gamma) f_{m-1}
        ~.
    \]

    \paragraph{Case 2b.}
    The second subcase is when $j_m = j_{m-2} + 2$.
    We will show that the probability of $F_m$ is upper bounded by the previous case.
    In this case, having the first $m-1$ rounds in $F_{m-1} \setminus G_{m-1}$ is no longer a sufficient condition since there may still be an arc $I_a^{j_{m-2},*}$ to $I_a^{j_m,*}$.
    Nonetheless, Eqn.~\eqref{eqn:panocs-general-last-node-matter-a1} is still an upper bound of the contribution of this part.

    Next, consider the case when the first $m-1$ rounds are in $G_{m-1}$.
    Comparing with the previous case, the main difference is that $I_a^{j_{m-2},*}$ is also an in-neighbor of $I_a^{j_m,*}$ in $D_a$.
    However, $I_a^{j_{m-2},*}$ serve as an competitor only when the PanOCS does \emph{not} follow its decision;
    otherwise, having an arc from $I_a^{j_{m-2},*}$ to $I_a^{j_m,*}$ in the \emph{ex-post} dependence graph $D_a^*$ also precludes event $G_{m-1}$.

    Next, we argue that Eqn.~\eqref{eqn:panocs-general-last-node-matter-a2} continues to serve as an upper bound of the contribution from this part.
    Formally, let $H$ denote the event that $I_a^{j_{m-2},*}$ is a sender who picks $I_a^{j_m,*}$, and further the PanOCS does \emph{not} follow $I_a^{j_{m-2},*}$'s decision.
    We show that \emph{conditioned on $G_{m-1}$}, event $H$ holds with probability at most $\frac{p}{4\kmax}$.
    Observe the intersection of events $G_{m-1}$ and $H$ is equivalent to the following collection of independent conditions:
    \begin{enumerate}
        \item The PanOCS follows $I_a^{j_{m-1},*}$'s decision (probability $\frac{1}{2}$);
        \item Node $I_a^{j_{m-1},*}$ is a sender (probability $p$);
        \item Node $I_a^{j_{m-1},*}$ picks $I_a^{j_m,*}$ (probability $\frac{1}{4\kmax}$);
        \item The PanOCS does \emph{not} follows $I_a^{j_{m-2},*}$'s decision (probability $\frac{1}{2}$);
        \item Node $I_a^{j_{m-2},*}$ is a sender (probability $p$);
        \item Node $I_a^{j_{m-2},*}$ picks $I_a^{j_m,*}$ (probability $\frac{1}{4\kmax}$); and
        \item Event $F_{m-3}$ (probability $f_{m-3}$).
    \end{enumerate}

    Putting together, the joint event happens with probability:
    \[
        \frac{p^2}{32 \kmax^2} f_{m-3}
        ~.
    \]

    Then, the conditional probability bound is:
    \begin{align*}
        \Pr \big[ H \,|\, G_{m-1} \big] 
        &
        = \frac{p^2}{64 \kmax^2} \frac{f_{m-3}}{g_{m-1}} && \text{(Bayes's rule)} \\
        &
        = \frac{p}{8 \kmax} \frac{f_{m-3}}{f_{m-2}} && \text{(Eqn.~\eqref{eqn:panocs-general-bid-auxiliary})} \\
        &
        \le \frac{p}{4\kmax} 
        ~.
    \end{align*}

    The last inequality is due to the observation that having event $F_{m-3}$ and having PanOCS not follow $I_a^{j_{m-2},*}$'s decision is sufficient for $F_{m-2}$.
\end{proof}

\section{Hybrid Algorithm}
\label{sec:hybrid}

This section gives a $0.5016$-competitive algorithm to prove Theorem~\ref{thm:main}.
This is a hybrid algorithm which combines the basic algorithm in Section~\ref{sec:basic-algorithm} and the algorithm of \citet{MehtaSVV/JACM/2007} (see also Appendix~\ref{app:small-bid}) to handle large and small bids with different strategies.
For large bids, we continue to utilize the negative correlation enabled by PanOCS.
For small bids, however, we fall back to either deterministic matches, or randomized matches with independent randomness.
By doing so, we can exploit the PanOCS for large bids (Algorithm~\ref{alg:large-bid-panocs-improved}, Theorem~\ref{thm:panocs-large-bid}, and Lemma~\ref{lem:panocs-large-bid-general}), and enjoy its superior performance compared to its counterpart for general bids (Algorithm~\ref{alg:panocs-weak} and Theorem~\ref{thm:panocs-general-bid}).




\subsection{Online Primal Dual Algorithm}

\subsubsection{Overview}

The hybrid algorithm is also an oblivious semi-randomized algorithm following the online primal dual framework.
When an impression $i$ arrives, each advertiser $a$ makes two offers $\Delta_a^R \beta_i$ and $\Delta_a^D \beta_i$.
They would be the increments of $\beta_i$ if $i$ is semi-assigned or assigned to advertiser $a$ respectively, determined by a specific dual update rule to be explained shortly.
Then, the impression $i$ either picks two offers of the first kind, or one offer of the second kind, whichever maximizes $\beta_i$.

Recall that we call the first case a randomized round, and say that the impression is semi-assigned to the two advertisers and the corresponding subsets given by the panoramic interval-level assignments.
In this case, we select one of the advertiser-subset combinations using the PanOCS for large bids in Algorithm~\ref{alg:large-bid-panocs-improved}, which has been presented in a generalized form that accepts both large and small bids but handle them differently.
In particular, it introduces negative correlation only among large bids.

Further recall that we call the second case a deterministic round, and say that the impression is assigned to the advertiser and the corresponding subset by the panoramic interval-level assignments.

In sum, the hybrid algorithm, defined in Algorithm~\ref{alg:hybrid-algorithm}, is almost identical to the basic algorithm at the high-level.
However, it uses a PanOCS that handles large and small bids differently, which in turn leads to a different dual update rule and different definitions of the offers $\Delta_a^R \beta_i$ and $\Delta_a^D \beta_i$.
The next subsections detail these differences.

\subsubsection{Primal Increments}

We first introduce a lower bound of the primal objective.
In the rest of the section, let $\gamma = 0.05144$ be the ratio given by Theorem~\ref{thm:panocs-large-bid} and its generalization stated as Lemma~\ref{lem:panocs-large-bid-general}.
Recall that $k_a^L(y)$ denote the number of times $y$ is semi-assigned by large bids before it is semi-assigned by a small bid.
For any advertiser $a \in A$ and any point $y \in [0, B_a)$, define:
\[
    \bar{x}_a(y) \defeq \begin{cases}
        1 - 2^{-k_a(y)} & y \in [0, \frac{B_a}{2}) ~; \\
        1 - 2^{-k_a(y)} (1 - \gamma)^{k_a^L(y)-1} & y \in [\frac{B_a}{2}, B_a), \textnormal{$k_a(y) \ne 1$ or $k_a^L(y) \ne 0$} ~; \\
        \frac{1}{2} - \frac{\gamma}{4} & y \in [\frac{B_a}{2}, B_a), 
        \textnormal{$k_a(y) = 1$ and $k_a^L(y) = 0$} ~.
    \end{cases}
\]

By Lemma~\ref{lem:panocs-large-bid-general}, for any advertiser $a \in A$ and any point $y \in [0, B_a)$:
\[
    x_a(y) \ge \bar{x}_a(y)
\]

The different definitions for the first and second halves of the interval $[0, B_a)$ differently is motivated by the online primal dual analysis of the small-bid algorithm of \citet{MehtaSVV/JACM/2007} (see Appendix~\ref{app:small-bid}), which also handles the two halves differently.
The lower bound of the second half is exactly the guarantee given by Lemma~\ref{lem:panocs-large-bid-general}, apart from a special case when $k_a(y) = 1$ and $k_a^L(y) = 0$.
The special case corresponds to when the first semi-assignment to $y$ is a small bid.
In this case, we penalize the small bid in order to reserve sufficient primal increments for future semi-assignments that are potentially large bids.
The analysis below will substantiate this intuition.
In the first half, however, we give up the correlation given by the PanOCS.
We remark that it is not an inferior choice and can be seen as banking up primal increments for the future.

Accordingly, define a surrogate primal objective which lower bounds the actual primal:
\[
    \bar{P} \defeq \sum_{a \in A} \int_0^\infty \bar{x}_a(y) dy
    ~.
\]

Hence, it would suffice to show competitiveness of the algorithm w.r.t.\ the surrogate objective.

\begin{algorithm}[t]
    \caption{Hybrid Algorithm \big(Parameterized by $\Delta\alpha_L^R(\cdot)$, $\Delta\alpha_R^R(\cdot)$, $\Delta\alpha_L^D(\cdot)$, and $\Delta\alpha_R^D(\cdot)$\big)}
    \label{alg:hybrid-algorithm}
    \begin{algorithmic}
        \smallskip
        \STATE \textbf{state variables:~}\\
            $k_a(y) \ge 0$, number of times $y$ is semi-assigned;
            $k_a(y) = \infty$ if $y$ is assigned in a deterministic round\\[1ex]
        %
        \FORALL{impression $i \in I$}
            \FORALL{advertiser $a \in A$}
                \STATE computer subset $Y_{ai} \subseteq [0,B_a)$ using panoramic interval-level assignment (Seciton~\ref{sec:panorama})
                \STATE compute $\Delta_a^R \beta_i$ and $\Delta_a^D \beta_i$ according to Equations~\eqref{eqn:hybrid-beta-r-increment} and \eqref{eqn:hybrid-beta-d-increment} 
            \ENDFOR
            \STATE find $a_1$, $a_2$ that maximize $\Delta_a^R \beta_i$, and $a^*$ that maximizes $\Delta_a^D \beta_i$
            \STATE \textbf{if} $\Delta^R_{a_1} \beta_i + \Delta^R_{a_2} \beta_i \ge \Delta^D_{a^*} \beta_i$ \hspace*{\fill} \textbf{\emph{\# randomized round}}
                    \INDSTATE assign $i$ to what PanOCS (large bids) selects between $a_1$ and $a_2$ and the corresponding subsets
                \STATE \textbf{else} (i.e., $\Delta^R_{a_1} \beta_i + \Delta^R_{a_2} \beta_i < \Delta^D_{a^*} \beta_i$) \hspace*{\fill} \textbf{\emph{\# deterministic round}}
                    \INDSTATE assign $i$ to $a^*$ and the corresponding subset
            \STATE \textbf{endif}
        \ENDFOR
    \end{algorithmic}
\end{algorithm}

\paragraph{Primal-increment Constants.}
%
We continue to introduce some constants for the increment in $\bar{x}_a(y)$ as a point $y$ gets further assignments and semi-assignments, depending on whether they are from large or small bids.
In the following discussions, subscripts $L$ and $R$ represent if $y$ is in the left half or the right half of the interval $[0, B_a)$.
Superscripts specify the nature of the assignments.
In particular, $D$ and $R$ stand for deterministic assignments and randomized semi-assignments respectively.
For points in the second half of the interval, a further superscript $L$ or $S$ indicates whether the assignments are from large or small bids.

We start with the simpler left half of the interval, which follows from the definition of $\bar{x}_a(y)$.
\begin{lemma}
    \label{lem:hybrid-primal-increment-left}
    For any advertiser $a \in A$ and any point $y \in [0, \frac{B_a}{2})$:
    \begin{itemize}
        \item The $k$-th \textbf{semi-assignment} to advertiser $a$ and point $y$ increases $\bar{x}_a(y)$ by:
            \[
                \Delta x^R_L(k) \defeq 2^{-k}
                ~.
            \]
        \item A \textbf{deterministic assignment} to advertiser $a$ and point $y$ after $k-1$ semi-assignments increases $\bar{x}_a(y)$ by:
            \[
                \Delta x^D_L(k) \defeq 2^{-k+1}
                ~.
            \]
            %
    \end{itemize}
\end{lemma}

As for the second half of the interval, the increments depend on the natural of the previous semi-assignments.
Nonetheless, we shall define history-free constants by taking the smallest increment over all possibilities.

\begin{lemma}
    \label{lem:hybrid-primal-increment-right}
    For any advertiser $a \in A$ and any $y \in [\frac{B_a}{2}, B_a)$:
    \begin{itemize}
        \item The $k$-th \textbf{semi-assignment} to advertiser $a$ and point $y$, if it is a \textbf{small bid}, increases $\bar{x}_a(y)$ by at least:
            \[
                \Delta x^{RS}_R(k) \defeq \begin{cases}
                    \frac{1}{2} - \frac{\gamma}{4} & k = 1 ~; \\[1ex]
                    2^{-k} (1-\gamma)^{k-2} & k \ge 2 ~.
                \end{cases}
            \]
        \item The $k$-th \textbf{semi-assignment} to advertiser $a$ and point $y$, if it is a \textbf{large bid}, increases $\bar{x}_a(y)$ by at least:
            \[
                \Delta x^{RL}_R(k) \defeq \begin{cases}
                    \frac{1}{2} & k = 1 ~; \\[1ex]
                    2^{-k} (1-\gamma)^{k-2} (1+\gamma) & k \ge 2 ~.
                \end{cases}
            \]
        \item A \textbf{deterministic assignment} to advertiser $a$ and point $y$ after $k-1$ semi-assignments increases $\bar{x}_a(y)$ by at least:
            \[
                \Delta x^D_L(k) \defeq 2^{-k+1} (1-\gamma)^{k-2}
                ~.
            \]
    \end{itemize}
\end{lemma}

\begin{proof}
    Recall the definition of $\bar{x}_a(y)$.
    When $y \in [ \frac{B_a}{2}, B_a)$, it falls into the second or the third case, which we restate below.
    If $k_a(y) \ne 1$ or $k_a^L(y) \ne 0$, we call it the regular case:
    \begin{equation}
        \label{eqn:hybrid-primal-increment-right-regular}
        \bar{x}_a(y) = 1 - 2^{-k_a(y)} (1 - \gamma)^{k_a^L(y)-1}.
    \end{equation}

    If $k_a(y) = 1$ and $k_a^L(y) = 0$, we call it special case:
    \begin{equation}
        \label{eqn:hybrid-primal-increment-right-special}
        \bar{x}_a(y) = \frac{1}{2} - \frac{\gamma}{4}.
    \end{equation}
    \textbf{(Semi-assignment, small bid)~}
    The case of $k = 1$ follows by the definition in the special case, i.e., Eqn.~\eqref{eqn:hybrid-primal-increment-right-special}.
    For $k \ge 2$, the expression of $\bar{x}_a(y)$ follows the regular case in Eqn.~\eqref{eqn:hybrid-primal-increment-right-regular}, both before and after this semi-assignment, and $k_a^L(y) \le k-1$ stays the same.
    Hence, the increment equals:
    \[
        2^{-k} (1-\gamma)^{k_a^L(y)-1}
        ~,
    \]

    The minimum is achieved when $k^L_a(y) = k-1$, and equals the $\Delta x^{RS}_R(k)$ defined in the lemma.\\[2ex]
    \textbf{(Semi-assignment, large bid)~}  
    First consider the binding case when all semi-assignments to $y$ are large bids.
    Then, $k_a(y) = k_a^L(y)$ holds throughout, in particular, before and after the semi-assignment at hand.
    Hence, if follows from the definition of $\bar{x}_a(y)$ that its increment equals the $\Delta x^{RL}_R(k)$ defined in the lemma.

    Next, suppose there was some previous semi-assignment to $y$ that is a small bid.
    If $k = 2$, it corresponds to the case the first semi-assignment is a small bid and the second one is a large bid.
    In other words, $\bar{x}_a(y)$ before the large bid at hand equals $\frac{1}{2} - \frac{\gamma}{4}$ due to the special case in its definition.
    After the semi-assignment of the large bid, $\bar{x}_a(y) = \frac{3}{4}$ by definition.
    The increment is therefore $\frac{1+\gamma}{4}$, which equals the $\Delta x^{RL}_R(k)$ for $k = 2$ defined in the lemma.
    This part crucially uses that $\bar{x}_a(y)$ is smaller than $\frac{1}{2}$ in the special case, reserving part of the gain for the second round.

    Finally, suppose $k \ge 3$.
    The expression of $\bar{x}_a(y)$ follows the regular case in Eqn.~\eqref{eqn:hybrid-primal-increment-right-regular}, both before and after this semi-assignment, and $k_a^L(y)$ stays the same.
    Hence, the increment equals:
    \[
        2^{-k} (1-\gamma)^{k_a^L(y)-1}
        ~.
    \]
 
    Further observe that we have $k_a^L(y) \le k - 2$ which achieves equality only when the first $k-2$ semi-assignments are all large bids and the $k-1$-th semi-assignment is a small one.
    Therefore, the increment is at least:
    \[
        2^{-k} (1-\gamma)^{k-3}
        ~.
    \]
    
    This is strictly greater than the $\Delta x^{RL}_R(k)$ defined in the lemma because $1 > (1-\gamma)(1+\gamma)$.\\[2ex]
    \textbf{(Deterministic assignment)~}
    Observe that $k_a^L(y) \le k_a(y) = k-1$ before the assignment.
    By the definition of $\bar{x}_a(y)$, the increment equals:
    \[
        2^{-k+1} (1-\gamma)^{k_a^L(y)-1}
        ~,
    \]

    It is minimized when $k_a^L(y) = k - 1$, and equals the $\Delta x^D_R(k)$ defined in the lemma.
\end{proof}

\subsubsection{Invariants}

The dual update rule is driven by two invariants.
The first one is easy to state:
\begin{invariant}
    Dual increment equals the primal increment defined in the previous subsection. 
\end{invariant}

%
Next we explain another invariant that determines the value of $\alpha_a(y)$ based on the status of advertiser $a$.
It is similar to Eqn.~\eqref{eqn:alpha-invariant} in the basic algorithm yet more complicated.
To begin with, its counterpart in the basic algorithm uses only one group of parameters, while here it needs four:
%
\begin{itemize}
    \item $\Delta\alpha_L^R(k)$:
        Increment in $\alpha_a(y)$ when $y \in [0, \frac{B_a}{2})$ is semi-assigned for the $k$-th time.
    \item $\Delta\alpha_R^R(k)$:
        Increment in $\alpha_a(y)$ when $y \in [\frac{B_a}{2}, B_a)$ is semi-assigned for the $k$-th time.
    \item $\Delta\alpha_L^D(k)$:
        Increment in $\alpha_a(y)$ when $y \in [0, \frac{B_a}{2})$ is deterministically assigned if there are $k-1$ semi-assignments before that.
    \item $\Delta\alpha_R^D(k)$:
        Increment in $\alpha_a(y)$ when $y \in [\frac{B_a}{2}, B_a)$ is deterministically assigned if there are $k-1$ semi-assignments before that.
\end{itemize}

These parameters will be selected in the online primal analysis by solving an LP to optimize the competitive ratio.

Let us make a few remarks regarding the design of these parameters.
First, observe that the increment in $\alpha_a(y)$ does not depend on whether the bid is small or large.
This is intentional so that the subsequent of arguments for approximate dual feasibility is history-free.
It also means that the costs of having small bids, which means smaller primal increments, are completely charged to the $\beta$ variables.
This intuitively indicates that the algorithm is less likely to choose small bids in randomized rounds compared to the large ones, all other things being equal.

Further, we handle the left and right halves of the interval $[0, B_a)$ differently.
Again, this is motivated by the analysis of the small-bid algorithm, and is aligned with our definition of the primal increments in the previous subsection.

Finally, we introduce two additional groups of parameters to model the increments in $\alpha_a(y)$ due to deterministic assignments.
By contrast, the basic algorithm uses an ad-hoc choice of $\Delta \alpha^D(k) = \sum_{\ell=k}^\infty \Delta \alpha^R(\ell)$.
We bring in the additional parameters because a constraint concerning the gain of $\beta$ variables in deterministic assignments, i.e., Eqn.~\eqref{eqn:beta-bound-half-to-a}, is nonbinding in the analysis of the basic algorithm.
In other words, we could have let the $\beta$ variables get less and let the $\alpha$ variables get more in deterministic assignments.
Although the slack is inconsequential in the basic algorithm, the extra gain is crucial for the hybrid algorithm and its analysis.

Let $k_a^R(y)$ be the number of semi-assignments to $y$ so that we retain this information even if $y$ has been eventually deterministically assigned.
Observe that $k_a^R(y) = k_a(y)$ if $k_a(y) < \infty$.

\begin{invariant}
    For any advertiser $a \in A$ and any point $y \in [0, B_a)$:
    %
    \begin{align*}
        \alpha_a(y) = \begin{cases}
            \displaystyle
            \sum_{\ell=1}^{k_a^R(y)} \Delta\alpha^R(y,\ell)
            & k_a(y) \neq \infty \\
            \displaystyle
            \sum_{\ell=1}^{k_a^R(y)} \Delta\alpha^R(y,\ell) + \Delta\alpha^D({k_a^R(y)+1})
            & k_a(y) = \infty.
        \end{cases}
    \end{align*}
\end{invariant}



\subsubsection{Dual Increments: \texorpdfstring{$\alpha$}{Alpha} Variables}

%

We further introduce notations $\Delta \alpha^R(y, k)$ and $\Delta \alpha^D(y, k)$ as follows so that subsequent integrals are more succinct.
\begin{align*}
    \Delta\alpha^R(y,k) &
    \defeq \begin{cases}
        \Delta\alpha^R_L(k) & y \in [0, \frac{B_a}{2}) ~; \\[1ex]
        \Delta\alpha^R_R(k) & y \in [\frac{B_a}{2}, B_a) ~.
    \end{cases} \\
    \Delta\alpha^D(y,k) &
    \defeq \begin{cases}
        \Delta\alpha^D_L(k) & y \in [0, \frac{B_a}{2}) ~; \\[1ex]
        \Delta\alpha^D_R(k) & y \in [\frac{B_a}{2}, B_a) ~.
    \end{cases}
\end{align*}

By the invariant regarding the $\alpha$ variables and the accounting by the point-level in Eqn.~\eqref{eqn:alpha-point-level}, the dual increment of $\alpha_a$ when an impression is semi-assigned or assigned to advertiser $a$ and the corresponding subset $Y_{ai}$ are ($k_a(y)$'s are the values before the assignment):
%
$$
    \Delta_i^R\alpha_a = \int_{Y_{ai}} \Delta\alpha^R(y,k_a(y) + 1) dy, \quad
    \Delta_i^D\alpha_a = \int_{Y_{ai}} \Delta\alpha^D(y,k_a(y) + 1) dy. 
$$

\subsubsection{Dual Increments: \texorpdfstring{$\beta$}{Beta} Variables}

By the first invariant, the definition of primal increments, and the definition of dual increments in terms of the $\alpha$ variables, the increments of $\beta$ variables by the point-level have been pinned down:
\begin{itemize}
    \item The $k$-th semi-assignment, $y \in [0, \frac{B_a}{2})$:
        \[
            \Delta\beta_L^R(k) \defeq \Delta x_L^{R}(k) - \Delta \alpha_L^R(k)
            ~.
        \]
    \item The $k$-th semi-assignment, small bid, $y \in [\frac{B_a}{2}, B_a)$:
        \[
            \Delta\beta_R^{RS}(k) \defeq \Delta x_R^{RS}(k) - \Delta \alpha_R^R(k)
            ~.
        \]
    \item The $k$-th semi-assignment, large bid, $y \in [\frac{B_a}{2}, B_a)$:
        \[
            \Delta\beta_R^{RL}(k) \defeq \Delta x_R^{RL}(k) - \Delta \alpha_R^R(k)
            ~.
        \]
    \item Deterministic assignment after $k-1$ semi-assignments, $y \in [0, \frac{B_a}{2})$:
        \[
            \Delta \beta_L^D(k) \defeq \Delta x_L^D(k) - \Delta \alpha_L^D(k)
            ~.
        \]
    \item Deterministic assignment after $k-1$ semi-assignments, $y \in [\frac{B_a}{2}, B_a)$:
        \[
            \Delta \beta_R^D(k) \defeq \Delta x_R^D(k) - \Delta \alpha_R^D(k)
            ~.
        \]
\end{itemize}

The next lemma shows that the $\beta$ increments for small bids are smaller than those for large bids.
This is because the primal increments for small bids are smaller, while the increments in $\alpha$ variables are the same for both small and large bids by definition.

\begin{lemma}
    \label{lem:hybrid-large-vs-small}
    For any $k \ge 1$, we have:
    \[
        \Delta \beta_L^{RL}(k) \ge \Delta \beta_L^{RS} (k)
        \quad,\quad
        \Delta \beta_R^{RL}(k) \ge \Delta \beta_R^{RS} (k)
        ~.
    \]
\end{lemma}

Again, we further introduce notations $\beta^{RS}(y,k)$, $\beta^{RL}(y,k)$, and $\beta^D(y,k)$ in order to make integrals succinct in subsequent arguments.
\begin{align*}
    \Delta\beta^{RL}(y,k)
    &
    \defeq \begin{cases}
        \Delta\beta_L^{R}(k) & y\in[0, \frac{B_a}{2}) ~; \\[1ex]
        \Delta\beta_R^{RL}(k) & y\in[\frac{B_a}{2}, B_a) ~.
    \end{cases}
    \\
    \Delta\beta^{RS}(y,k)
    &
    \defeq \begin{cases}
        \Delta\beta_L^{R}(k) & y\in[0, \frac{B_a}{2}) ~; \\[1ex]
        \Delta\beta_R^{RS}(j) & y\in[\frac{B_a}{2}, B_a) ~.
    \end{cases}
    \\
    \Delta\beta^{D}(y,k)
    &
    \defeq \begin{cases}
        \Delta\beta_L^D(k) & y \in [0, \frac{B_a}{2}) ~; \\[1ex]
        \Delta\beta_R^D(j) & y \in [\frac{B_a}{2},B_a) ~.
    \end{cases}
\end{align*}

%
Therefore, the increments of $\beta_i$ when $i$ is semi-assigned or assigned to $a$, i.e., the offers $\Delta^{R}_a\beta_i$ and $\Delta^D_a\beta_i$ in the hybrid algorithm, are defined as:
\begin{align}
    \label{eqn:hybrid-beta-r-increment}
    \Delta^{R}_a\beta_i
    &
    \defeq \begin{cases}
        \displaystyle
        \int_{Y_{ai}} \Delta\beta^{RL} (y,k_a(y) + 1) dy & \textnormal{large bid, i.e., } \frac{B_a}{2} \le b_{ai} \le B_a ~; \\[2ex]
        \displaystyle
        \int_{Y_{ai}} \Delta\beta^{RS} (y,k_a(y) + 1) dy & \textnormal{small bid, i.e., } 0 \le b_{ai} < \frac{B_a}{2} ~.
    \end{cases}
    \\
    \label{eqn:hybrid-beta-d-increment}
    \Delta^D_a\beta_i
    &
    \defeq ~~~ \int_{Y_{ai}} \Delta\beta^D(y,k_a(y) + 1) dy
    ~.
\end{align}

\subsubsection{Regularity Constraints}

Finally, we introduce several sets of regularity constraints on the parameters.
The first two correspond to the monotonicity assumption, i.e., Eqn.~\eqref{eqn:dual-monotonicity}, and the superiority of randomized round, i.e., Eqn.~\eqref{eqn:random-vs-deter}, in the basic algorithm.
The last one assumes nonnegativity of the dual variables.
From now on, we shall label the constraints that will appear in the LP for optimizing the parameters and the competitive ratio by (C1), (C2), (C3), and so on.

\paragraph{Monotonicity.}
The first set of constraints state that $\Delta\beta^{RL}(y,k)$, $\Delta\beta^{RS}(y,k)$ and $\Delta\beta^D(y,k)$ are nondecreasing w.r.t.\ $y+k\cdot B_a$.
That is, they are larger for a smaller $k$, and conditioned on the same $k$ they are larger for a smaller $y$'s (i.e., a point $y$ on the left half of the interval).
\begin{align}
    \label{con:hybrid-delta-beta-rl-decreasing}
    \tag{C1}
    \forall k \ge 1: &&
    \Delta\beta^{RL}_{L}(k) &\geq \Delta\beta^{RL}_{R}(k) ~,&\quad \Delta\beta^{RL}_{R}(k) &\geq \Delta\beta^{RL}_{L}(k+1) ~, \\
    \label{con:hybrid-delta-beta-rs-decreasing}
    \tag{C2}
    \forall k \ge 1: &&
    \Delta\beta^{RS}_{L}(k) &\geq \Delta\beta^{RS}_{R}(k) ~,&\quad \Delta\beta^{RS}_{R}(k) &\geq \Delta\beta^{RS}_{L}(k+1) ~, \\
    \label{con:hybrid-delta-beta-d-decreasing}
    \tag{C3}
    \forall k \ge 1: &&
    \Delta\beta^{D}_{L}(k)  &\geq \Delta\beta^{D}_{R}(k) ~,&\quad \Delta\beta^{D}_{R}(k) &\geq \Delta\beta^{D}_{L}(k+1) ~.
\end{align}

We state below two useful consequences of the monotonicity.

\begin{lemma}
    \label{lem:hybrid-panoramic-interval}
    The panoramic interval-level assignment chooses a subset $Y_{ai}$ that maximizes the offers $\Delta_a^R \beta_i$ in Eqn.~\eqref{eqn:hybrid-beta-r-increment} and $\Delta_a^D \beta_i$ in Eqn.~\eqref{eqn:hybrid-beta-d-increment}.
    %
\end{lemma}

\begin{lemma}
    \label{lem:hybrid-offer-concave}
    The offers $\Delta_a^R \beta_i$ in Eqn.~\eqref{eqn:hybrid-beta-r-increment} and $\Delta_a^D \beta_i$ in Eqn.~\eqref{eqn:hybrid-beta-d-increment} as functions of impression $i$'s bid $b_{ai}$ are concave.
\end{lemma}


\paragraph{Superiority of Randomized Rounds.}
The second set of constraints implies that the algorithm prefers randomizing over two equally good advertisers with semi-assignments, over a deterministic assignment to only one of them.
%
\begin{equation}
     \label{con:hybrid-random-vs-deterministic}
     \tag{C4}
     \begin{aligned}
         \forall k \ge 1 : \quad && 
         2 \cdot \Delta \beta_L^{RL}(k) \ge \Delta \beta_L^D(k) \quad,\quad 2 \cdot \Delta \beta_L^{RS}(k) \ge \Delta \beta_L^D(k) ~; \\
         \forall k \ge 1 : \quad && 
         2 \cdot \Delta \beta_R^{RL}(k) \ge  \Delta \beta_R^D(k) \quad,\quad 2 \cdot \Delta \beta_R^{RS}(k) \ge  \Delta \beta_R^D(k) ~.
     \end{aligned}
\end{equation}

Recall that the parameters satisfy that $\Delta \beta_L^{RL}(k) \ge \Delta \beta_L^{RS}(k)$ and $\Delta \beta_R^{RL}(k) \ge \Delta \beta_R^{RS}(k)$.
In this sense, listing only the inequalities for small bids would have sufficed.
Nevertheless, we opt to make it more explicit above.

As a direct corollary:
\begin{lemma}
    \label{lem:hybrid-random-superior}
    For any advertiser $a \in A$ and any impression $i \in I$:
    \[
        \Delta_a^D \beta_i \le 2 \cdot \Delta_a^R \beta_i
        ~.
    \]
\end{lemma}


\paragraph{Nonnegativity.}
We shall choose the parameters $\Delta \alpha_L^R(\cdot), \Delta \alpha_L^D(\cdot), \Delta \alpha_R^R(\cdot), \Delta \alpha_R^D(\cdot)$ to be nonnegative.
Further, we shall ensure that the corresponding $\beta$ parameters are nonnegative. 
%
%
\begin{align}
    \label{con:hybrid-alpha-non-negative}
    \tag{C5}
    \forall k \ge 1 : \qquad &
    \Delta \alpha_L^R(k) \geq 0,\ \Delta \alpha_R^R(k) \geq 0,\ \Delta \alpha_L^D(k) \geq 0,\ \Delta \alpha_R^D(k) \geq 0 ~; \\
    \label{con:hybrid-beta-non-negative}
    \tag{C6}
    \forall k \ge 1 : \qquad &
    \Delta \beta_L^R(k) \geq 0,\  \Delta \beta_R^{RS}(k) \geq 0,\ \Delta \beta_R^{RL}(k) \geq 0,\ \Delta \beta_L^D(k) \geq 0,\ \Delta \beta_R^D(k) \geq 0 ~.
\end{align}

\begin{lemma}
    \label{lem:hybrid-non-negative}
    For any advertiser $a \in A$ and any impression $i$, $\alpha_a \ge 0$ and $\beta_i \ge 0$.
\end{lemma}

\subsection{Online Primal Dual Analysis}


In this section, we let the competitive ratio $\Gamma$ also be a parameter to be optimized together with the other ones in the analysis.
Next, we derive a set of sufficient conditions for proving that the hybrid algorithm is $\Gamma$-competitive.

\paragraph{Reverse Weak Duality.}
This holds for any choice of the parameters by the design of the primal dual algorithm.
In particular, we ensure the invariant that dual increment equals the lower bound of the surrogate primal increment given in Lemma~\ref{lem:hybrid-primal-increment-left} and Lemma~\ref{lem:hybrid-primal-increment-right}.
Therefore, we have $\bar{P} \ge D$.
Recall that the surrogate primal objective lower bounds the actual one, i.e., $P \ge \bar{P}$, we get reverse weak duality.

\paragraph{Approximate Dual Feasibility.}
%
%
Fix any advertiser $a$, and an impression set $S$.
We restate approximate dual feasibility below, where the contribution of $\alpha_a$ is accounted by the point-level:
\begin{equation}
    \label{eqn:hybrid-basic-approximate-feasibility}
    \int_0^{B_a} \alpha_a(y) dy + \sum_{i \in S} \beta_i \ge \Gamma\cdot b_a(S).
    ~.
\end{equation}



Next, we will explain how to lower bound the $\alpha$ and $\beta$ variables respectively, and will show a charging scheme that distributes the lower bounds of the $\beta$ variables to the points $y \in [0, B_a)$.
These are similar to their counterparts in the basic algorithm.
Unlike the basic algorithm, however, here the lower bound distributed to each point may not be at least $\Gamma$ on its own.
We will demonstrate how to prove Eqn.~\eqref{eqn:hybrid-basic-approximate-feasibility} by constructing an appropriate measure preserving mapping between points in the left half and those in the right half of the interval $[0, B_a)$, such that the lower bound charged to each pair is at least $2 \Gamma$.

\subsubsection[Lower Bound of Alpha]{Lower Bound of $\alpha_a$}

By the definition of the $\alpha$-invariant, the gain from $\alpha_a(y)$ depends on $k_a(y)$ and $k^R_a(y)$.
To simplify the case, we impose further constraints below:
%
\begin{align}
    \label{con:hybrid-random-vs-determine-alpha-1}
    \tag{C7}
    \forall k \ge 1 : \qquad &
    \Delta \alpha_L^R(k) \leq \Delta \alpha_L^D(k) - \alpha_L^D(k+1) ~, \\
    \label{con:hybrid-random-vs-determine-alpha-2}
     \tag{C8}
    \forall k \ge 1 : \qquad &
    \Delta \alpha_R^R(k) \leq \Delta \alpha_R^D(k) - \alpha_R^D(k+1) ~.
\end{align}

They imply that for any point $y$ with $k_a(y) = \infty$, the larger $k^R_a(\cdot)$ is, the small $\alpha_a(y)$ is by the $\alpha$-invariant.
How large could $k^R_a(\cdot)$ be?
For this, recall the property of $k_a(\cdot)$ from Lemma~\ref{lem:K-property}.
It states that other than the subsets that have been deterministically assigned, the value of $k_a(y)$ equals $\tilde{k}_a(y)$ defined as:
\[
    \tilde{k}_a(y) \defeq \begin{cases}
        k_{\min} + 1 & y<y^* ~; \\
        k_{\min} & y \geq y^* ~.
    \end{cases}
\]

Here, recall that $k_{\min} = \min_{z \in [0, B_a)} k_a(z)$, and $y^*$ denote the start point of the next subset in the panoramic interval-level assignment.

Then, the largest possible value of $k^R_a(y)$ is $\tilde{k}_a(y) - 1$ because at least the last time must be reserved for the deterministic assignment.
Further recall that $Y_D$ denote the subset of of points in $[0, B_a)$ that are deterministically assigned:
\[
    Y_D \defeq \big\{y\in[0,B_a) : y\text{ is deterministically assigned by the end of the algorithm }\big\}
    ~.
\]

We have the following lower bound of the $\alpha$ variables.

\begin{lemma}
    \label{hybrid:gain-from-alpha}
    For any advertiser $a \in A$ and any point $y \in [0, B_a)$, we have:
    \[
        \alpha_a(y) \ge \begin{cases}
            \displaystyle
            \sum_{\ell = 1}^{\tilde{k}_a(y)} \Delta \alpha^R(y, \ell) & y \notin Y_D ~; \\[4ex]
            \displaystyle
            \sum_{\ell = 1}^{\tilde{k}_a(y) - 1} \Delta \alpha^R(y, \ell) + \Delta \alpha^D(y, \tilde{k}_a(y)) & y \in Y_D ~.
        \end{cases}
    \]
    %
\end{lemma}

In fact, the former case always holds with equality, although this observation is unimportant for our analysis.

\subsubsection[Lower Bound of Beta Variables and Charging to Points]{Lower Bounds of $\beta$ Variables and Charging to Points}

We now turn to the lower bound of $\beta_i$ for impressions $i \in S$.
Similar to the analysis of the basic algorithm in Section~\ref{sec:basic-algorithm}, it depends on the matching status of impression $i$, in particular, whether $i$ is semi-assigned or assigned to advertiser $a$ and whether it is a deterministic or randomized round.
Recall the definitions of sets $N$ and $R$ from the analysis of the basic algorithm:
%
\begin{align*}
    N & \defeq \big\{i \in S : i\text{ is neither assigned nor semi-assigned to }a\big\}  ~, \\
    R & \defeq \big\{i \in S : i\text{ is semi-assigned to }a\big\} ~.
\end{align*}

Further define the sum of the bids in these subsets as $b_N$ and $b_R$ respectively for future reference:
\[
    b_N \defeq \sum_{i\in N} b_{ai}
    \quad,\quad
    b_R \defeq \sum_{i \in R} b_{ai}
    ~.
\]

We will find two subset $Y_N, Y_R \subseteq [0, B_a)$, and will distribute the lower bounds of $\beta_i$'s, $i \in N$ and $i \in R$, to the points $y \in Y_N$ and $y \in Y_R$ respectively.
More precisely, we shall define $\beta(y)$'s such that:
\begin{itemize}
    \item $Y_N$, $Y_R$, and $Y_D$ are disjoint.
    \item $Y_N$, $Y_R$, and $Y_D$ have total measure at least $b_a(S)$, i.e.:
        \begin{equation}
            \label{eqn:hybrid-total-measure}
            \mu(Y_N) + \mu(Y_R) + \mu(Y_D) \geq b_a(S)
            ~.
        \end{equation}
    \item The $\beta(y)$'s lower bound the $\beta_i$'s, i.e.:
        \begin{align}
            \label{eqn:hybrid-distribution-N}
            \sum_{i \in N} \beta_i & 
            \ge \int_{Y_N} \beta(y) dy ~; \\
            \label{eqn:hybrid-distribution-R}
            \sum_{i \in R} \beta_i &
            \ge \int_{Y_R} \beta(y) dy ~.
        \end{align}
\end{itemize}


%
\paragraph{Construction of $Y_N$ and the Corresponding $\beta(y)$'s.}
Consider the impressions $i \in N$.
There are two different cases depending on whether $i$ is a deterministic or randomized round.
We claim that in both cases:
\[
    \beta_i \ge 2 \cdot \Delta_a^R \beta_i
    ~.
\]
%

Suppose it is a deterministic round.
Since $i$ chooses advertiser $a^*$ deterministically instead of randomizing between advertisers $a^*$ and $a$, $\beta_i = \Delta_{a^*}^D \beta_i \ge \Delta_{a^*}^R \beta_i + \Delta_a^R \beta_i$.
Further by Lemma~\ref{lem:hybrid-random-superior}, $\Delta_{a^*}^D \beta_i \le 2 \Delta_{a^*}^R \beta_i$.
Cancelling $\Delta_{a^*}^R \beta_i$ by combining the two inequalities leads to $\beta_i= \Delta_{a^*}^D \beta_i \ge 2 \Delta_a^R \beta_i$.

Suppose it is a randomized round.
By definition, both candidates in this round offer at least $\Delta_a^R \beta_i$, or else the algorithm would have chosen advertiser $a$ instead.
Hence, $\beta_i \ge 2 \Delta_a^R \beta_i$.

Next, express $\Delta_a^R \beta_i$ in terms of the state variables using the definition of $\Delta_a^R \beta_i$ in Eqn.~\eqref{eqn:hybrid-beta-r-increment}.
Recall that $k_a^i(y)$'s denote the values of the state variables when impression $i$ arrives, and $Y_{ai}$ denotes the subset by the panoramic interval-level assignment, should $i$ be semi-assigned or assigned to advertiser $a$ when it arrives.
If $b_{ai}$ is a small:
\begin{equation}
    \label{eqn:hybrid-basic-beta-bound-n-small}
    \beta_i \ge 2 \cdot \int_{Y_{ai}} \Delta \beta^{RS}(k_a^i(y)+1) dy
    ~.
\end{equation}

If $b_{ai}$ is large:
\begin{equation}
    \label{eqn:hybrid-basic-beta-bound-n-large}
    \beta_i \ge 2 \cdot \int_{Y_{ai}} \Delta \beta^{RL}(k_a^i(y)+1) dy
    ~.
\end{equation}

We need to further derive a lower bound w.r.t.\ the state variables $k_a(y)$'s at the end of the algorithm.
Its proof is almost verbatim to its counterpart in the basic algorithm, i.e., Lemma~\ref{lem:hybrid-basic-status-comparison-n}.
We include it for completeness.

\begin{lemma}
    \label{lem:hybrid-basic-status-comparison-n}
    For any subset $\tilde{Y}_{ai}$ with measure at most $b_{ai}$, we have:
    \begin{align*}
        \int_{Y_{ai}} \Delta \beta^{RS}(k_a^i(y)+1) dy & \ge \int_{\tilde{Y}_{ai}} \Delta \beta^{RS}(k_a(y)+1) dy
        ~, \\
        \int_{Y_{ai}} \Delta \beta^{RL}(k_a^i(y)+1) dy & \ge \int_{\tilde{Y}_{ai}} \Delta \beta^{RL}(k_a(y)+1) dy
        ~.
    \end{align*}
\end{lemma}

\begin{proof}
    Since the panoramic interval-level assignment chooses a subset $Y_{ai}$ of measure $b_{ai}$ with the minimum $k_a^i(y)$'s, i.e., Lemma~\ref{lem:hybrid-panoramic-interval}, by the monotonicity of $\Delta \beta^{RS} (\cdot)$ and $\Delta \beta^{RL} (\cdot)$ in Eqn.~\eqref{con:hybrid-delta-beta-rl-decreasing} and Eqn.~\eqref{con:hybrid-delta-beta-rs-decreasing}, we have:
    \begin{align*}
        \int_{Y_{ai}} \Delta \beta^{RS}(k_a^i(y)+1) dy & \ge \int_{\tilde{Y}_{ai}} \Delta \beta^{RS}(k_a^i(y)+1) dy
        ~, \\
        \int_{Y_{ai}} \Delta \beta^{RL}(k_a^i(y)+1) dy & \ge \int_{\tilde{Y}_{ai}} \Delta \beta^{RL}(k_a^i(y)+1) dy
        ~.
    \end{align*}

    Further observe that $k_a(y) \ge k_a^i(y)$ for any $y \in [0, B_a)$.
    Applying the monotonicity of $\Delta \beta^{RS} (\cdot)$ and $\Delta \beta^{RL} (\cdot)$ once again proves the lemma.
\end{proof}

Next, define $Y_N$ as:
\begin{equation}
    \label{eqn:hybrid-beta-charging-n}
    Y_N \defeq \big[y^*, y^* \oplus_{Y_D} b_N \big)
    ~.
\end{equation}

Here, recall that $[y, y \oplus_Y b)$ denotes a subset of $[0, B_a)$ with measure $b$, obtained by scanning through the interval starting from $y$ and exluding the points in $Y$.

The values of $\beta(y)$'s are determined by two possible lower bounds $\Psi^{NL}$ and $\Psi^{NS}$ for $\sum_{i \in N} \beta_i$, which we explain below.
The first one corresponds to when there is a single large impression $i \in N$ with bid $b_{ai} = b_N$.
The second one corresponds to when there are one or two small impressions $i \in N$ whose bids sum to $b_N$, and in case of two impressions, one of them has the largest possible small bid, i.e., $\frac{B_a}{2}$.

Formally, define the first lower bound as:
\[
    \Psi^{NL} \defeq 2\cdot\int_{Y_N} \Delta\beta^{RL}(y,k_a(y)+1)dy
    ~.
\]

To define the second lower bound, let $b_{N1} = \min \{B_a/2,~B_N\}$ and further define:
\begin{align*}
    Y_{N1} & = \big[y^*, y^* \oplus_{Y_D} b_{N1}\big) \setminus Y_D,\\[1ex]
    Y_{N2} & = \big[y^* \oplus_{Y_D} b_{N1}, y^* \oplus_{Y_D} b_{N}\big) \setminus Y_D.
\end{align*}

Observe that $Y_{N2} = Y_N \setminus Y_{N1}$ by definition.
The second lower bound is then defined as:
%
\[
    \Psi^{NS} \defeq 2 \int_{Y_{N1}} \Delta\beta^{RS} \big( y, k_a(y)+1 \big) dy + 2 \int_{Y_{N2}} \Delta\beta^{RS} \big( y\ominus_{Y_D} \tfrac{B_a}{2}, k_a(y \ominus_{Y_D} \tfrac{B_a}{2}) + 1 \big) dy
    ~.
\]

An explanation for the design of $y$ in the second integral is due.
This part is relevant only when $b_N > \frac{B_a}{2}$.
That is, we consider two small bids: the first is $\frac{B}{2}$ and the second is $b_N - \frac{B_a}{2}$.
Here, mapping $y$ to $y \ominus_{Y_D} \frac{B_a}{2}$ is a measure preserving map from $Y_{N2}$ to the first $b_N - \frac{B_a}{2}$ measure of $Y_{N1}$.
By doing so, we explicitly use the fact that the second small bid's $\beta_i$ can be lower bounded using a subset starting from $y^*$ in Lemma~\ref{lem:hybrid-basic-status-comparison-n}, even though the lower bound is eventually charged to $Y_{N2}$ which does not start from $y^*$.

By the monotonicity of parameters $\Delta \beta^{RS}_L(\cdot)$ and $\Delta \beta^{RS}_R(\cdot)$, i.e., Eqn.~\eqref{con:hybrid-delta-beta-rs-decreasing}, this is strictly better than the trivial bound without the mapping of $y$, which we used in the analysis of the basic algorithm.
\textbf{This is the main power of small bids, which pays for the penalties they had in the dual update rule.}

The values of $\beta(y)$'s are set according to the smaller of the two lower bounds:
\begin{equation}
    \label{eqn:hybrid-beta-definition}
    \beta(y) \defeq \begin{cases}
        2 \cdot \Delta \beta^{RL}(y,k_a(y)+1) & \Psi^{NL} < \Psi^{NS} ~; \\[2ex]
        2 \cdot \Delta \beta^{RS}(y,k_a(y)+1) & \Psi^{NL} \geq \Psi^{NS}, y \in Y_{N1} ~; \\[2ex]
        2 \cdot \Delta \beta^{RS} \big( y\ominus_{Y_D} \tfrac{B_a}{2}, k_a(y\ominus_{Y_D} \tfrac{B_a}{2})+1 \big) & \Psi^{NL} \geq \Psi^{NS}, y\in Y_{N2} ~.
    \end{cases}
\end{equation}


Finally, we prove that the smaller of the two indeed lower bounds $\sum_{i \in N} \beta_i$.

    %
\paragraph{Proof of Eqn.~\eqref{eqn:hybrid-distribution-N}.}
We restate the inequality below:
\[
    \sum_{i\in N} \beta_i \geq \int_{Y_N} \beta(y) = 2\cdot\min \left\{ \Psi^{NL},~\Psi^{NS}\right\}
    ~.
\]
%

Let $\Phi^{NS}(b)$, $0 \le b \le \frac{B_a}{2}$, denote the offer from advertiser $a$ for a small bid $b$ \emph{at the final state of the algorithm}.
Observe that it would be semi-assigned to $[y^*, y^* \oplus_{Y_D} b) \setminus Y_D$.
We have:
\[
    \Phi^{NS}(b) \defeq \int_{[y^*, y^* \oplus_{Y_D} b) \setminus Y_D} \Delta \beta^{RS} \big( y, k_a(y)+1 \big) dy
    ~.
\]

Similarly, $\Phi^{NL}(b)$, $\frac{B_a}{2} < b \le B_a$, denote the offer from advertiser $a$ for a large bid $b$ \emph{at the final state of the algorithm}:
\[
    \Phi^{NL}(b) \defeq \int_{[y^*, y^* \oplus_{Y_D} b) \setminus Y_D} \Delta \beta^{RL} \big( y, k_a(y)+1 \big) dy
    ~.
\]

For any impression $i \in N$, letting $\tilde{Y}_{ai} = [y^*, y^* \oplus_{Y_D} b_{ai})$ in Lemma~\ref{lem:hybrid-basic-status-comparison-n}, we have:
\begin{equation}
    \label{eqn:hybrid-beta-bound-charging-n}
    \begin{aligned}
        \int_{Y_{ai}} \Delta \beta^{RS}(y, k_a^i(y)+1) dy
        &
        \ge \Phi^{NS}(b_{ai})
        ~, \\
        \int_{Y_{ai}} \Delta \beta^{RL}(y, k_a^i(y)+1) dy
        &
        \ge \Phi^{NL}(b_{ai})
        ~.
    \end{aligned}
\end{equation}

Recall that the LHS of the above inequalities are half the lower bound of $\beta_i$ for small and large bids respectively, due to Eqn.~\eqref{eqn:hybrid-basic-beta-bound-n-small} and Eqn.~\eqref{eqn:hybrid-basic-beta-bound-n-large}.
In other words, we lower bound each $\beta_i$ by what advertiser $a$ would have offered if impression $i$ arrived at the end.

The rest of the proof further transforms the sum of these lower bounds for $\beta_i$ for $i \in N$ into the stated bound in the lemma.
First we may asssume wlog that $b_N \le B_a - \mu(Y_D)$.
Otherwise, we could have decrease the bid of some impressions in $N$ such that the LHS of Eqn.~\eqref{eqn:hybrid-distribution-N} decreases, while the RHS stays the same.

Further, it is wlog to merge small bids into at most two; in the case of two small bids, it is wlog that the larger one has size $\frac{B_a}{2}$.
This is because $\Phi^{NS}(\cdot)$ is concave by By Lemma~\ref{lem:hybrid-offer-concave}.
Formally, for any two small bids $b \ge b'$ and any $\delta$, concavity implies:
\[
    \Phi^{NS}(b) + \Phi^{NS}(b') \ge \Phi^{NS}(b+\delta) + \Phi^{NS}(b'-\delta)
    ~.
\]
Then, we may let $\delta = b'$ if $b + b' \le \frac{B_a}{2}$, and let $\delta = \frac{B_a}{2} - b$ otherwise.
Repeating this operation proves the claim.

Finally, we claim that it is wlog to assume having either only small bids, or a single large bid.
Observe that there can be at most one large bid by defintion.
In the presence of both large and small bids, and after the aforementioned merging of small bids, it must be the case that we have one large bid, say $b > \frac{B_a}{2}$, and one small bid $b' \le \frac{B_a}{2}$.

Since both $\Phi^{NS}(\cdot)$ and $\Phi^{NL}(\cdot)$ are concave by Lemma~\ref{lem:hybrid-offer-concave}, we either have:
\[
    \Phi^{NL}(b) + \Phi^{NS}(b') \ge \Phi^{NL}(b+\delta) + \Phi^{NS}(b'-\delta)
    ~,
\]
for any $0 \le \delta \le b'$, or:
\[
    \Phi^{NL}(b) + \Phi^{NS}(b') \ge \Phi^{NL}(b-\delta) + \Phi^{NS}(b'+\delta)
    ~,
\]
for any $0 \le \delta \le b - \frac{B_a}{2}$.
The range of $\delta$ in the second case is chosen such that the large bid does not become small. 
Observe that the small bid can not become large without letting the large bid become small since they sum to at most $B_a$.

In the former case, we let $\delta = b'$ to eliminate the small impression.

In the latter case, we let $\delta = b - \frac{B_a}{2}$.
Then, the claim follows by the observation that conditioned on having the same size $\frac{B_a}{2}$, downgrading the large bid into a small bid leads to a smaller offer and thus, a smaller lower bound for the corresponding $\beta_i$.
The observation follows by the definition of $\beta$ increment in Eqn.~\eqref{eqn:hybrid-beta-definition}, and the comparison of $\beta$ increments for large and small bids in Lemma~\ref{lem:hybrid-large-vs-small}.
%


\paragraph{Reading Guide.}
The remaining parts of the subsection, including the construction of $Y_R$ and the correponsding $\beta(y)$'s, the disjointness of $Y_N$, $Y_R$, and $Y_D$, and their measure bounds, are almost verbatim to the counterparts in the basic algorithm.
We include them below for completeness.
Nonetheless, readers may want to skip to the next subsection.

\paragraph{Construction of $Y_R$ and the Corresponding $\beta(y)$'s.}
Since the algorithm does not choose matching to advertiser $a$ deterministically, $\beta_i \ge \Delta_a^D \beta_i$.
By the definition of $\Delta_a^D \beta_i$ in Eqn.~\eqref{eqn:hybrid-beta-d-increment}:
\begin{equation}
    \label{eqn:hybrid-beta-bound-r}
    \beta_i \geq \int_{Y_R} \Delta\beta^D(y,k_a^i(y)).
\end{equation}

We need to further derive a lower bound w.r.t.\ the $k_a(y)$'s at the end of the algorithm.
    The next lemma is similar to Lemma~\ref{lem:hybrid-basic-status-comparison-n} in the previous case, but more generally considers arbitrary $\hat{k}_a(y) \ge k_a^i(y)$ instead of only the $k_a(y)$ at the end of the algorithm.

\begin{lemma}
    \label{lem:hybrid-status-comparison-r}
    Consider any $\hat{k}_a(y)$'s such that $\hat{k}_a(y) \ge k_a^i(y)$ for any $y \in [0, B_a)$.
    Then, for any subset $\tilde{Y}_{ai}$ with measure at most $b_{ai}$:
    \[
        \int_{Y_{ai}} \Delta\beta^D(y,k^i_a(y)+1) dy \ge \int_{\tilde{Y}_{ai}} \Delta\beta^D(y,\hat{k}_a(y)+1) dy
        ~.
    \]
\end{lemma}

\begin{proof}
    Since the panoramic interval-level assignment chooses a subset $Y_{ai}$ of measure $b_{ai}$ with the minimum and left most $k_a^i(y)$'s, combining with the monotonicity of $\Delta \beta^D$ in Eqn.~\eqref{con:hybrid-delta-beta-d-decreasing}:
    \[
        \int_{Y_{ai}} \Delta\beta^D(y,k^i_a(y)+1) dy \ge \int_{\tilde{Y}_{ai}} \Delta\beta^D(y,k^i_a(y)+1) dy
        ~.
    \]

    The lemma then follows by the assumption that $\hat{k}_a(y) \ge k_a^i(y)$ for any $y \in [0, B_a)$, and by the monotonicity of $\Delta \beta^D$ in Eqn.~\eqref{con:hybrid-delta-beta-d-decreasing}.
\end{proof}

    For any $i \in R$, define $\tilde{Y}_{ai}$ as:
    %
    \begin{equation}
        \label{eqn:hybrid-beta-charging-per-impression-r}
        \tilde{Y}_{ai} = \Big[ y^* \ominus_{Y_D} \sum_{i' \in R\,:\,i' \ge i} b_{ai'}, y^* \ominus_{Y_D} \sum_{i' \in R\,:\,i' > i} b_{ai'} \Big) \setminus Y_D
        ~.
    \end{equation}

    In other words, we scan \emph{backwards} through the interval $[0, B_a)$ starting from $y^*$, treating the interval as a circle by gluing its endpoints.
    Then, we construct $\tilde{Y}_{ai}$'s for $i \in R$ one at a time by their arrival order from latest to earliest, letting each be a subset excluding $Y_D$ with measure up to $b_{ai}$.
    If $\sum_{i \in N} b_{ai} \le B_a - \mu(Y_D)$, which we consider the canonical case, these would be the panoramic interval-level assignments if these $i \in R$ arrived at the end of the instance, assuming the same final state of the algorithm.
    By the definition of the boundary case, we stop scanning through $[0, B_a)$ after a full circle;
    therefore, the above $\tilde{Y}_{ai}$'s are disjoint.

    Define $Y_R$ and the corresponding $\beta(y)$ as:
    \begin{equation}
        \label{eqn:hybrid-beta-charging-r}
        \begin{aligned}
            Y_R & \defeq \bigcup_{i \in R} \tilde{Y}_{ai} \setminus Y_N ~, \\
            \forall y \in Y_R : \qquad \beta(y) & \defeq \Delta\beta^D(y,k_a(y)) ~.
        \end{aligned}
    \end{equation}

    We remark that $Y_R$ can be simplified as $Y_R = [y^* \ominus_{Y_D} b_R, y^*) \setminus Y_D$ if $b_R + b_N + \mu(Y_D) \le B_a$, which we consider the canonical case of the analysis.

\paragraph{Proof of Eqn.~\eqref{eqn:hybrid-distribution-R}.}
We restate the inequality below:
\[
    \sum_{i\in R} \beta_i \geq \int_{Y_R} \beta(y) dy = \int_{Y_R} \Delta\beta^D(y,k_a(y)) dy
    ~.
\]

For any $i \in R$, define $k^{-i}_a(y)$ by considering what the state variables of advertiser $a$ would have been before the arrival of $i$ \emph{if the impressions in $R$ were the latest ones in the instance}.
    More precisely, for any $i \in R$, let:
    \[
        k^{-i}_a(y) \defeq \begin{cases}
            k_a(y) - 1 & y \in \big[ y^* \ominus_{Y_D} \sum_{i' \in R\,:\,i' \ge i} b_{ai'}, y^* \big) \setminus Y_D ~; \\
            k_a(y) & \text{otherwise.}
        \end{cases}
    \]

    Intuitively, these are the largest possible values of $k_a^i(y)$'s.
    We restate Lemma~\ref{lem:basic-k-comparison-r} below, which still holds in the hybrid algorithm with a verbatim proof.

    \begin{lemma}
        \label{lem:hybrid-k-comparison-r}
        For any $i \in R$ and any $y \in [0, B_a)$:
        \[
            k^{-i}_a(y) \ge k^i_a(y)
            ~.
        \]
    \end{lemma}

    Consider any $i \in R$.
    By definition, $\tilde{Y}_{ai}$ is a subset with measure at most $b_{ai}$.
    Further, Lemma~\ref{lem:hybrid-k-comparison-r} above allows us to apply Lemma~\ref{lem:hybrid-status-comparison-r}, which by the monotonicity of $\Delta \beta^D(\cdot)$ in Eqn.~\eqref{con:hybrid-delta-beta-d-decreasing} gives:
    \[
        \int_{Y_{ai}} \Delta\beta^D(y,k_a^i(y) +1) \ge \int_{\tilde{Y}_{ai}} \Delta\beta^D(y,k_a^{-i}(y) +1) dy
        ~.
    \]

    Finally, by $k_a^{-i}(y) = k_a(y) - 1$ for any $y \in \tilde{Y}_{ai}$:
    \begin{equation}
        \label{eqn:hybrid-beta-bound-charging-r}
        \int_{Y_{ai}} \Delta\beta^D(y,k_a^i(y) +1) dy \ge \int_{\tilde{Y}_{ai}} \Delta\beta^D(y,k_a(y)) dy
        ~.
    \end{equation}

    Eqn.~\eqref{eqn:hybrid-distribution-R} then follows by Eqn.~\eqref{eqn:hybrid-beta-bound-r}, the above inequality in Eqn.~\eqref{eqn:hybrid-beta-bound-charging-r}, and the definition of $Y_R$ and the correposnding $\beta(y)$ for $y \in Y_R$ in Eqn.~\eqref{eqn:hybrid-beta-charging-r}, through a sequence of inequalities as follows:
    \begin{align*}
        \sum_{i \in R} \beta_i
        &
        \ge \sum_{i \in R} \int_{Y_{ai}} \Delta\beta^D(y,k_a^i(y) +1) dy
        &&
        \text{(Eqn.~\eqref{eqn:hybrid-beta-bound-r})} \\
        &
        \ge \sum_{i \in R} \int_{\tilde{Y}_{ai}} \Delta\beta^D(y,k_a(y)) dy
        &&
        \text{(Eqn.~\eqref{eqn:hybrid-beta-bound-charging-r})} \\
        &
        \ge \int_{Y_R} \beta(y) dy
        ~.
        &&
        \text{(Eqn.~\eqref{eqn:hybrid-beta-charging-r})}
    \end{align*}

\paragraph{Disjointness.}
\label{par:hybrid-disjoint-and-measure}
The sets $Y_N$ and $Y_R$ can be written as:
\[
    Y_N = \big[ y^* , ~y^* \oplus_{Y_D} b_N \big) \setminus Y_D
    \quad,\quad
    Y_R = \big [ y^* \ominus_{Y_D} b_R ,~ y^* \big) \setminus (Y_D \cup Y_N)
    ~.
\]
Hence, they are disjoint by definition.

\paragraph{Measure Bound.}
%
%
If $b_N + b_R + \mu(Y_D) > B_a$, the union of $Y_D$, $\big[ y^* , ~y^* \oplus_{Y_D} b_N \big)$, and $\big[ y^* \ominus_{Y_D} b_R ,~ y^* \big)$ covers $[0, B_a)$.
Further by the above equivalent forms of $Y_N$ and $Y_R$, the union of $Y_N$, $Y_R$, and $Y_D$ also covers $[0, B_a)$.
Then, the measure bound follows by:
\[
    \mu(Y_N) + \mu(Y_R) + \mu(Y_D) \geq B_a \ge b_a(S)
    ~.
\]

Otherwise, $Y_R$ simplifies as $\big [ y^* \ominus_{Y_D} b_R ,~ y^* \big) \setminus Y_D$.
We have:
\[
    \mu(Y_N) = b_N = \sum_{i \in R} b_{ai}
    \quad , \quad
    \mu(Y_R) = b_R = \sum_{i \in R} b_{ai}
    ~.
\]

Further, any impression $i \in S$ that is not in $N$ or $R$ must be deterministically assigned.
Hence:
\[
    \mu(Y_D) \geq \sum_{i\in S\setminus (N \cup R)} b_{ai}
    ~.
\]

Together we have:
\[
    \mu(Y_R) + \mu(Y_N) + \mu(Y_D) \geq \sum_{i\in S} b_{ai} \ge b_a(S)
    ~.
\]

\subsubsection{Amortization: Pairing Points Between Left and Right}
\label{subsec:simply-optimization}


Let us summarize the lower bound from $\alpha_a(y)$ and $\beta(y)$ by the point-level below.
%
\begin{itemize}
    \item If $y \in Y_N$ (large bid subcase), by Lemma~\ref{hybrid:gain-from-alpha} and Eqn.~\eqref{eqn:hybrid-beta-definition}, $\alpha_a(y) + \beta(y)$ is at least:
        \[
            \begin{cases}
                \psi^{NL}_L \big( \tilde{k}_a(y) \big) \defeq \sum_{\ell=1}^{\tilde{k}_a(y)} \Delta \alpha^R_L(\ell) + 2 \cdot \Delta \beta^{RL}_L \big(\tilde{k}_a(y) + 1\big)
                &
                0 \le y < \frac{B_a}{2}
                ~; \\[2ex]
                \psi^{NL}_R \big( \tilde{k}_a(y) \big) \defeq \sum_{\ell=1}^{\tilde{k}_a(y)} \Delta \alpha^R_R(\ell) + 2 \cdot \Delta \beta^{RL}_R \big(\tilde{k}_a(y) + 1\big)
                &
                \frac{B_a}{2} \le y < B_a
                ~.
            \end{cases}
        \]
    \item If $y \in Y_N$ and further $y \in Y_{N1}$ (small bid subcase, first small bid), by Lemma~\ref{hybrid:gain-from-alpha} and Eqn.~\eqref{eqn:hybrid-beta-definition}, $\alpha_a(y) + \beta(y)$ is at least:
        \[
            \begin{cases}
                \psi^{NS1}_L \big( \tilde{k}_a(y) \big) \defeq \sum_{\ell=1}^{\tilde{k}_a(y)} \Delta \alpha^R_L(\ell) + 2 \cdot \Delta \beta^{RS}_L \big(\tilde{k}_a(y) + 1\big)
                & 
                0 \le y < \frac{B_a}{2}
                ~; \\[2ex]
                \psi^{NS1}_R \big( \tilde{k}_a(y) \big) \defeq \sum_{\ell=1}^{\tilde{k}_a(y)} \Delta \alpha^R_R(\ell) + 2 \cdot \Delta \beta^{RS}_R \big(\tilde{k}_a(y) + 1\big)
                & 
                \frac{B_a}{2} \le y < B_a
                ~.
            \end{cases}
        \]
    \item If $y \in Y_N$ and further $y \in Y_{N2}$ (small bid subcase, second small bid), by Lemma~\ref{hybrid:gain-from-alpha} and Eqn.~\eqref{eqn:hybrid-beta-definition}, $\alpha_a(y) + \beta(y)$ is at least:
        \[
            \begin{cases}
                \sum_{\ell=1}^{\tilde{k}_a(y)} \Delta \alpha^R_L(\ell) + 2 \cdot \Delta \beta^{RS} \big(y \ominus_{Y_D} \tfrac{B_a}{2}, \tilde{k}_a(y \ominus_{Y_D} \tfrac{B_a}{2}) + 1\big)
                & 
                0 \le y < \frac{B_a}{2}
                ~; \\[2ex]
                \sum_{\ell=1}^{\tilde{k}_a(y)} \Delta \alpha^R_R(\ell) + 2 \cdot \Delta \beta^{RS} \big(y \ominus_{Y_D} \tfrac{B_a}{2}, \tilde{k}_a(y \ominus_{Y_D} \tfrac{B_a}{2}) + 1\big)
                & 
                \frac{B_a}{2} \le y < B_a
                ~. \\[2ex]
            \end{cases}
        \]
        This case is more involved since it is unclear whether point $y \ominus_{Y_D} \frac{B_a}{2}$ is on the left or the right.
        If point $y$ is on the left half, i.e., $y < \frac{B_a}{2}$, both $y^*$ and $y \ominus_{Y_D} \frac{B_a}{2}$ are on the right of $y$.
        Hence, $\tilde{k}_a(y \ominus_{Y_D} \tfrac{B_a}{2}) \le \tilde{k}_a(y) - 1$.
        By the monotonicity of $\Delta \beta^{RS}(y, k)$:
        \[
            \alpha_a(y) + \beta(y) \ge \psi^{NS2}_L \big(\tilde{k}_a(y) \big) \defeq \sum_{\ell=1}^{\tilde{k}_a(y)} \Delta \alpha^R_L(\ell) + 2 \cdot \Delta \beta^{RS}_R \big(\tilde{k}_a(y)\big)
            \qquad 0 \le y < \frac{B_a}{2}
            ~.
        \]

        Next suppose point $y$ is on the right half, i.e., $\frac{B_a}{2} \le y < B_a$.
        If point $y \ominus_{Y_D} \frac{B_a}{2}$ is on the left half, $\tilde{k}_a(y \ominus_{Y_D} \tfrac{B_a}{2}) \le \tilde{k}_a(y)$ and $\beta(y) \ge 2 \Delta \beta_L^{RS} ( \tilde{k}_a(y) + 1)$.
        Otherwise, i.e., point $y \ominus_{Y_D} \frac{B_a}{2}$ is on the right, $\tilde{k}_a(y \ominus_{Y_D} \tfrac{B_a}{2}) \le \tilde{k}_a(y) - 1$, and $\beta(y) \ge 2 \Delta \beta_R^{RS}(\tilde{k}_a(y))$, which is even larger than the previous case by the monotonicity of $\Delta \beta^{RS}(y, k)$.
        In sum:
        \[
            \alpha_a(y) + \beta(y) \ge \psi^{NS2}_R \big(\tilde{k}_a(y) \big) \defeq \sum_{\ell=1}^{\tilde{k}_a(y)} \Delta \alpha^R_R(\ell) + 2 \cdot \Delta \beta^{RS}_L \big(\tilde{k}_a(y) + 1\big)
            \qquad \frac{B_a}{2} \le y < B_a
            ~.
        \]
    \item If $y \in Y_R$, by Lemma~\ref{hybrid:gain-from-alpha} and Eqn.~\eqref{eqn:hybrid-beta-charging-r}, $\alpha_a(y) + \beta(y)$ is at least:
        \[
            \begin{cases}
                \psi^R_L\big( \tilde{k}_a(y) \big) \defeq \sum_{\ell=1}^{\tilde{k}_a(y)} \Delta \alpha^R_L(\ell) + \Delta \beta^D_L(\tilde{k}_a(y))
                &
                0 \le y < \frac{B_a}{2}
                ~; \\[2ex]
                \psi^R_R\big( \tilde{k}_a(y) \big) \defeq \sum_{\ell=1}^{\tilde{k}_a(y)} \Delta \alpha^R_R(\ell) + \Delta \beta^D_R(\tilde{k}_a(y))
                & 
                \frac{B_a}{2} \le y < B_a
                ~.
            \end{cases}
        \]
    \item If $y \in Y_D$, by Lemma~\ref{hybrid:gain-from-alpha} and define $\beta(y) = 0$, $\alpha_a(y) + \beta(y)$ is at least:
        %
        \[
            \begin{cases}
                \psi^D_L\big(\tilde{k}_a(y)\big) \defeq \sum_{\ell=1}^{\tilde{k}_a(y)-1} \Delta \alpha^R_L(\ell) + \Delta \alpha^D_L \big( \tilde{k}_a(y) \big)
                &
                0 \le y < \frac{B_a}{2}
                ~; \\[2ex]
                \psi^D_R\big(\tilde{k}_a(y)\big) \defeq \sum_{\ell=1}^{\tilde{k}_a(y)-1} \Delta \alpha^R_R(\ell) + \Delta \alpha^D_R \big( \tilde{k}_a(y) \big)
                & 
                \frac{B_a}{2} \le y < B_a
                ~.
            \end{cases}
        \]
\end{itemize}

\paragraph{Approximate Dual Feasibility Fails Locally.}
In the previous analysis of the basic algorithm, approximate dual feasibility holds locally in that $\alpha_a(y) + \beta(y) \ge \Gamma$ for any $y \in Y_N \cup Y_R \cup Y_D$.
If we follow the same strategy, the above functions from $\psi^{NL}_L(k)$ to $\psi^D_R(k)$ need to be at least $\Gamma$ for all possible values of $k = \tilde{k}_a(y)$.
This is impossible, however, for any nontrivial competitive ratio $\Gamma > 0.5$.
This shall not be surprising since it goes against the strategy of handling the left and right halves of the interval $[0, B_a)$ differently in the hybrid algorithm.

\paragraph{New Plan: Pairing Points Between Left and Right.}
Based on the above discussion, an amortization between left and right is needed.
Concretely, we will design a measure preserving map $h: [0, \frac{B_a}{2}) \mapsto [\frac{B_a}{2}, B_a)$ such that:
%
%
\begin{equation}
    \label{eqn:hybrid-approximate-feasibility-group-2}
    \forall y \in \Big[0, 
    \frac{B_a}{2} \Big) \quad, \qquad  \alpha_a(y) + \beta(y) + \alpha_a(h(y)) + \beta(h(y)) \geq 2\Gamma
    ~.
\end{equation}

This would be sufficient for approximate dual feasibility if the union of $Y_N$, $Y_R$, and $Y_D$ covers $[0, B_a)$.
To make the above pairing idea works in the general case, we further impose the following constraint on the increments in $\beta$ variables:
\begin{equation}
    \label{con:hybrid-simlfy-to-full}
    \tag{C9}
    \Delta\beta^{RS}_L(1) \le \frac{\Gamma}{2}
    ~.
\end{equation}


\begin{lemma}
    Eqn.~\eqref{eqn:hybrid-approximate-feasibility-group-2} implies approximate dual feasibility.
\end{lemma}

\begin{proof}
    First consider the case when the union of $Y_N$, $Y_R$, and $Y_D$ covers $[0, B_a)$.
    \begin{align*}
        \int_0^{B_a} \alpha_a(y) ~dy + \sum_{i \in S} \beta_i
        &
        \ge \int_0^{B_a} \alpha_a(y) ~dy + \int_{Y_N \cup Y_N \cup Y_D} \beta(y) ~dy 
        \tag{Eqn.~\eqref{eqn:hybrid-distribution-N} and \eqref{eqn:hybrid-distribution-R}} \\
        &
        = \int_0^{B_a} \big( \alpha_a(y) + \beta(y) \big) ~dy \\
        &
        = \int_0^{\frac{B_a}{2}} \big( \alpha_a(y) + \beta(y) + \alpha_a(h(y)) + \beta(h(y)) \big) ~dy \tag{$h$ is measure preserving} \\
        &
        \ge \int_0^{\frac{B_a}{2}} 2\Gamma ~dy = \Gamma \cdot B_a
        ~.
        \tag{Eqn.~\eqref{eqn:hybrid-approximate-feasibility-group-2}}
    \end{align*}

    Next, consider an instance in which this does not hold, i.e.: 
    \[
        \sum_{i \in S} b_{ai} \le \mu(Y_N \cup Y_R \cup Y_D) < B_a
        ~.
    \]

    Suppose we add to $S$ a set of impressions that are small bids w.r.t.\ advertiser $a$ to $S$ summing to $B_a - \sum_{i \in S} b_{ai}$.
    Then, the RHS of approximate dual feasibility increases by $\Gamma \cdot (B_a - \sum_{i \in S} b_{ai})$, while the LHS increases by at most this amount due to Eqn.~\eqref{con:hybrid-simlfy-to-full} and the monotonicity of $\Delta \beta^{RS}_L(\cdot)$ and $\Delta \beta^{RS}_R(\cdot)$ in Eqn.~\eqref{con:hybrid-delta-beta-rs-decreasing}.
    In other words, the new impressions make approximate dual feasibility harder to satisfy, and prove the lemma by reducing it to the case when the union of $Y_N$, $Y_R$, and $Y_D$ covers $[0, B_a)$.
\end{proof}

\paragraph{Naïve Measure Preserving Map.}
Consider a trivial measure preserving map $h(y) = y + \frac{B_a}{2}$.
By the property of $k_a(\cdot)$ in Lemma~\ref{lem:K-property}, there are two possible combinations of $\Tilde{k}_a(y)$ and $\Tilde{k}_a(h(y))$:
\begin{itemize}
    \item $\Tilde{k_a}(y) = k_{\min} + 1$ and $\Tilde{k_a}(h(y)) = k_{\min}$; or
    \item $\Tilde{k_a}(y) = k_{\min}$ and $\Tilde{k_a}(h(y)) = k_{\min}$.
\end{itemize}

Moreover, we can rule out an impossible case by the policy of interval-level assignment.
\begin{lemma}[Impossible Case]
    For any point $y$ on the left half, and any point $h(y)$ on the right half, it is impossible that $y \in N$ and $\Tilde{k}_a(y) = k_{\min} + 1$, while $h(y) \in R$ and $\Tilde{k}_a(h(y)) = k_{\min}$.
\end{lemma}

\begin{proof}
    Suppose for contrary that there are such points $y$ and $h(y)$.
    By $y \in N$, we have $[y^*, y) \subseteq Y_N$.
    Further by $\tilde{k}_a(y) = \kmin+1$, it follows from Lemma~\ref{lem:K-property} that $y < y^*$.
    Hence, by the panoramic treatment of the notation for intervals:
    \[
        [y^*, y) = [0, y) \cup [y^*, B_a) \subseteq Y_N
        ~.
    \]

    Similarly, by $h(y) \in Y_R$, we have $[h(y), y^*) \subseteq Y_R$.
    Further by $\tilde{k}_a(h(y)) = \kmin$, it follows from Lemma~\ref{lem:K-property} that $h(y) \ge y^*$.
    Hence, by the panoramic treatment of notation:
    \[
        [h(y), y^*) = [0, y^*) \cup [h^(y), B_a) \subseteq Y_R
        ~.
    \]

    Together we conclude that $Y_N \cap Y_R$ is not empty, contradicting their construction.
\end{proof}

Table~\ref{table:first-approximate-feasibility-constraint} summarizes the worst-case bound for the LHS of Eqn.~\eqref{eqn:hybrid-approximate-feasibility-group-2}.
Here, whenever $y$, or $h(y)$, or both belong to $Y_N$, we use the worst-case bound when they are in the small-bid subcase, and when $y, h(y) \in Y_{N1}$.
This is indeed the main drawback of the trivial mapping since we cannot exploit the extra gain from the $Y_{N2}$ case of small bids.
As a result, these constraints are still too restricted to get any competitive better than $0.5$.

\begin{table}[t]
    \renewcommand{\arraystretch}{1.4}
    \centering
    \begin{tabular}{|c|c|c|c|}
        \hline
        & $Y_N,k$ & $Y_R,k$ & $Y_D,k$ \\ \hline
        $Y_N,k$             & $\psi_L^{NS1}(k) + \psi_R^{NS1}(k)$       & $\psi_L^{NS1}(k) + \psi_R^R(k)$           & $\psi_L^{NS1}(k) + \psi_R^D(k)$       \\ \hline
        $Y_N,k+1$           & $\psi_L^{NS1}(k+1) + \psi_R^{NS1}(k)$     & (impossible) 
        & $\psi_L^{NS1}(k+1) + \psi_R^D(k)$     \\ \hline
        $Y_R,k$             & $\psi_L^R(k) + \psi_R^{NS1}(k)$           & $\psi_L^R(k) + \psi_R^R(k)$               & $\psi_L^R(k) + \psi_R^D(k)$          \\ \hline
        $Y_R,k+1$           & $\psi_L^R(k+1) + \psi_R^{NS1}(k)$         & $\psi_L^R(k+1) + \psi_R^R(k)$             & $\psi_L^R(k+1) + \psi_R^D(k)$        \\ \hline
        $Y_D,k$             & $\psi_L^D(k) + \psi_R^{NS1}(k)$           & $\psi_L^D(k) + \psi_R^R(k)$               & $\psi_L^D(k) + \psi_R^D(k)$          \\ \hline
        $Y_D,k+1$           & $\psi_L^D(k+1) + \psi_R^{NS1}(k)$         & $\psi_L^D(k+1) + \psi_R^R(k)$             & $\psi_L^D(k+1) + \psi_R^D(k)$        \\ \hline
    \end{tabular}
    \caption{
        Approximate dual feasibility constraints with the naïve map $h(y) = y + \frac{B_a}{2}$.
        Rows are combinations of $y$'s type and $\tilde{k}_a(y)$.
        Columns are combinations of $h(y)$'s type and $\tilde{k}_a(h(y))$.
        Write $\tilde{k}_a(h(y))$ as $k$ for brevity.
        The range is $k \ge 1$;
        the first, second, fourth, and sixth cells in the first column further include $k = 0$.
        This is because $y \not \in Y_R \cup Y_D$ when $\tilde{k}_a(y) = 0$ and $h(y)\not\in Y_R \cup Y_D$ when  $\tilde{k}_a(h(y)) = 0$.
        Each formula in the table shall be at least $2 \Gamma$.
    }
    \label{table:first-approximate-feasibility-constraint}
\end{table}

\paragraph{Our Measure Preserving Map.}
Next, we design a better mapping to improve the cases when $y, h(y) \in Y_N$.
In particular, we shall rule out the case when both of them are in $Y_{N1}$, i.e., the top-left cells in Table~\ref{table:final-approximate-feasibility-constraint}.
Intuitively, this is possible because the total measure of $Y_{N1}$ is at most $\frac{B_a}{2}$.
Concretely, we construct the measure-preserving map $h$ in four steps as follows:
%
%
\begin{enumerate}
    \item For any $y \in [0, \frac{B_a}{2})$ such that $y \in Y_{N2}$ (which means $y\ominus_{Y_D} \frac{B_a}{2} \in Y_{N1}$), and $y\ominus_{Y_D} \frac{B_a}{2} \in [\frac{B_a}{2},B_a)$, map $y$ to $h(y) \defeq y \ominus_{Y_D} \frac{B_a}{2}$.
    \item For any $y \in [0, \frac{B_a}{2})$ such that $y \oplus_{Y_D} \frac{B_a}{2} \in Y_{N2}$ (which means $y \in Y_{N1}$), and $y \oplus_{Y_D} \frac{B_a}{2} \in [\frac{B_a}{2}, B_a)$, map $y$ to $h(y) \defeq y\oplus_{Y_D} B_a/2$.
\end{enumerate}

These two steps consider pairs of points $y_1 \in Y_{N1}$ and $y_2 = y_1 \oplus_{Y_D} \frac{B_a}{2} \in Y_{N2}$ such that there is one of the left half and one on the right half.
The measure-preserving map $h$ then maps the former to the latter. We show an example in Figure~\ref{fig:measure-preserving-map}.

\begin{enumerate}
    \setcounter{enumi}{2}
    \item For the points $y \in [0, \frac{B_a}{2})$ and $y \in Y_{N1}$ whose map $h(y)$ remains undefined after the first two steps, map to the unmapped points in $[\frac{B_a}{2}, B_a) \setminus Y_{N1}$ in an arbitrary measure-preserving way.
\end{enumerate}

The third step ensures that the points in $Y_{N1}$ are not mapped with each other.
Why is this possible?
Suppose a measure of $\mu_0$ has been mapped from each half of the interval $[0, B_a)$ in the first two steps.
Recall that the first two steps only define a mapping between $Y_{N1}$-$Y_{N2}$ pairs.
Further by the observation that $Y_{N1}$ has measure at most $\frac{B_a}{2}$, the total measure of the unmapped points in $Y_{N1}$ is at most $\frac{B_a}{2} - \mu_0$.%
\footnote{It holds with equality in the case of two small bids, and with strict inequality in the case of only one small bid strictly smaller than $\frac{B_a}{2}$.}
Suppose $\mu_L$ of these measure are on the left half.
Then, $Y_{N1}$ has a total measure of at most $\frac{B_a}{2} - \mu_0 - \mu_L$ on the right.
Finally, since the unmapped measure on the right half is precisely $\frac{B_a}{2} - \mu_0$, we conclude that the unmapped measure on the right half excluding $Y_{N1}$ is at least $(\frac{B_a}{2} - \mu_0) - (\frac{B_a}{2} - \mu_0 - \mu_L) = \mu_L$.


\begin{enumerate}
    \setcounter{enumi}{3}
    %
    \item For the points $y \in [0,\frac{B_a}{2})$ whose map $h(y)$ is still undefined, map to the unmapped points in $[\frac{B_a}{2}, B_a)$ in an arbitrary measure-preserving way.
\end{enumerate}

\begin{figure}
    \centering
    \includegraphics[width=.5\textwidth]{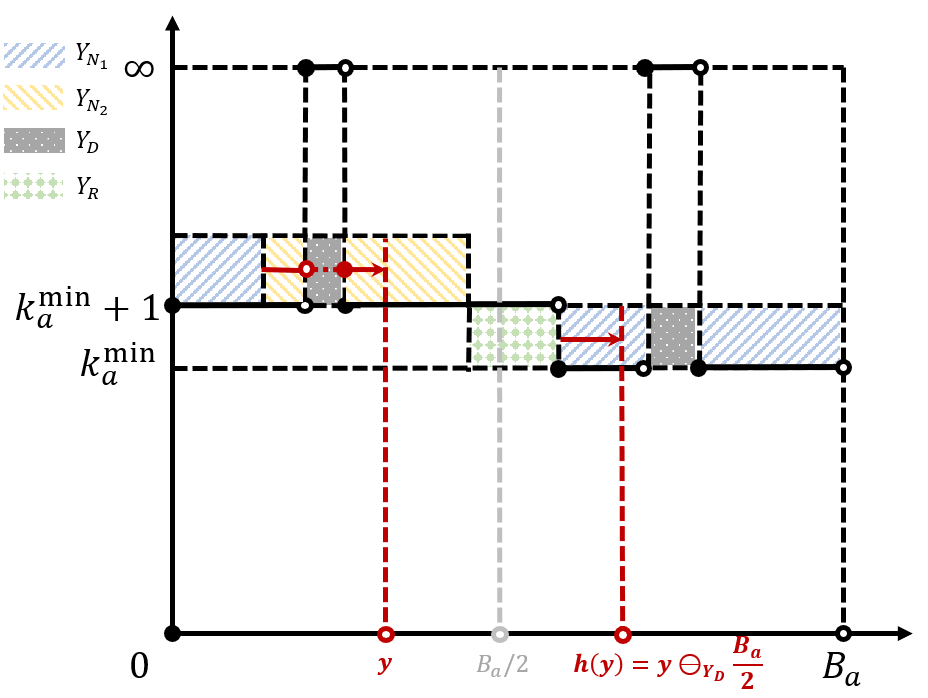}
    \caption{Example: map $y$ to $h(y)$. ($\mu(Y_{N_1})=B_a/2$)}
    \label{fig:measure-preserving-map}

\end{figure}

Next, we present the main property of the above measure-preserving map $h$.

\begin{lemma}[Refined $Y_N$-$Y_N$ Small-bid Cases]
    \label{lem:hybrid-h-function-property}
    Suppose $y$ and $h(y)$ are both in $Y_N$ in the small-bid subcase.
    Then, it must be one of the following cases:
    \begin{itemize}
        \item $y \in Y_{N1}$ with $\tilde{k}_a(y) = \kmin$ and $h(y) \in Y_{N2}$ with $\tilde{k}_a(h(y)) = \kmin$.
        \item $y \in Y_{N1}$ with $\tilde{k}_a(y) = \kmin +   1$ and $h(y) \in Y_{N2}$ with $\tilde{k}_a(h(y)) = \kmin + 1$.
        \item $y \in Y_{N2}$ with $\tilde{k}_a(y) = \kmin + 1$ and $h(y) \in Y_{N1}$ with $\tilde{k}_a(h(y)) = \kmin$.
    \end{itemize}
\end{lemma}



\begin{proof}
    By definition, two points in $Y_{N1}$ can not be matched.
    Further, if $Y_{N2}$ is nonempty then $Y_{N1}$ must have measure $\frac{B_a}{2}$.
    Hence, any measure-preserving map that does not map points in $Y_{N1}$ to each other must not map points outside $Y_{N1}$ to each other.
    In other words, two points in $Y_{N2}$ cannot be matched.
    Hence, we have one of $y$ and $h(y)$ in $Y_{N1}$ and the other in $Y_{N2}$.

    Suppose $y \in Y_{N1}$ and $h(y) \in Y_{N2}$.
    By definition, $[y, h(y)) \setminus Y_D \subseteq Y_N$, which does not contain $y^*$ in it.
    In other words, $y < \frac{B_a}{2} \le h(y)$ must be on the same side of $y^*$.
    Then, by Lemma~\ref{lem:K-property} either $\tilde{k}_a(y) = \tilde{k}_a(h(y)) = \kmin$, which happens if they are both on the right of $y^*$, or $\tilde{k}_a(y) = \tilde{k}_a(h(y)) = \kmin+1$, which happens if they are both on the left of $y^*$.

    Next suppose $y \in Y_{N2}$ and $h(y) \in Y_{N1}$.
    Similarly, by definition $[h(y), y) \subseteq Y_N$, which does not contain $y^*$ in it.
    Here recall the panoramic treatment of notation.
    By $y < \frac{B_a}{2} \le h(y)$ we further conclude $[h(y), y) = [0, y) \cup [h(y), B_a) \subseteq Y_N$.
    In other words, $y$ is on the left of $y^*$ and $h(y)$ is on the right of $y^*$.
    Then, by Lemma~\ref{lem:K-property}, $\tilde{k}_a(y) = \kmin+1$ and $\tilde{k}_a(h(y)) = \kmin$.
\end{proof}


Hence, for the top-left cells in the table which concern the cases when both $y$ and $h(y)$ are in $Y_N$, it suffices to consider the large-bid subcase and the small-bid subcase stated in the above lemma.
We formulate the refined constraints in Table~\ref{table:final-approximate-feasibility-constraint}.

\begin{table}[t]
    \renewcommand{\arraystretch}{2.0}
    \centering
    \begin{tabular}{|c|c|c|c|}
        \hline
        & $Y_N,k$ & $Y_R,k$ & $Y_D,k$ \\ \hline
        $Y_N,k$             & 
        $\begin{aligned}
            & \psi_L^{NS1}(k) + \psi_R^{NS2}(k) ~, \\   
            & \psi_L^{NL}(k) + \psi_R^{NL}(k)
        \end{aligned}$
        & $\psi_L^{NS1}(k) + \psi_R^R(k)$           & $\psi_L^{NS1}(k) + \psi_R^D(k)$       \\
        \hline
        $Y_N,k+1$           &
        $\begin{aligned}
            & \psi_L^{NS2}(k+1) + \psi_R^{NS1}(k) ~,\\
            & \psi_L^{NL}(k+1) + \psi_R^{NL}(k)
        \end{aligned}$
        & (impossible) 
        & $\psi_L^{NS1}(k+1) + \psi_R^D(k)$     \\ \hline
        $Y_R,k$             & $\psi_L^R(k) + \psi_R^{NS1}(k)$           & $\psi_L^R(k) + \psi_R^R(k)$               & $\psi_L^R(k) + \psi_R^D(k)$
        \\ \hline
        $Y_R,k+1$           & $\psi_L^R(k+1) + \psi_R^{NS1}(k)$         & $\psi_L^R(k+1) + \psi_R^R(k)$             & $\psi_L^R(k+1) + \psi_R^D(k)$        \\ \hline
        $Y_D,k$             & $\psi_L^D(k) + \psi_R^{NS1}(k)$           & $\psi_L^D(k) + \psi_R^R(k)$               & $\psi_L^D(k) + \psi_R^D(k)$          \\ \hline
        $Y_D,k+1$           & $\psi_L^D(k+1) + \psi_R^{NS1}(k)$         & $\psi_L^D(k+1) + \psi_R^R(k)$             & $\psi_L^D(k+1) + \psi_R^D(k)$        \\ \hline
    \end{tabular}
    \caption{
        Approximate dual feasibility constraints with our measure preserving map.
        Rows are combinations of $y$'s type and $\tilde{k}_a(y)$.
        Columns are combinations of $h(y)$'s type and $\tilde{k}_a(h(y))$.
        Write $\tilde{k}_a(h(y))$ as $k$ for brevity.
        The range is $k \ge 1$;
        the first, second, fourth, and sixth cells in the first column further include $k = 0$.
        This is because $y \not \in Y_R \cup Y_D$ when $\tilde{k}_a(y) = 0$ and $h(y)\not\in Y_R \cup Y_D$ when  $\tilde{k}_a(h(y)) = 0$.
        Each formula in the table shall be at least $2 \Gamma$.
    }
    \label{table:final-approximate-feasibility-constraint}
\end{table}

\subsection{Optimizing the Parameters: Proof of Theorem~\ref{thm:main}}

Finally, to avoid having an infinite number of parameters and constraints, we set all parameters, including  $\Delta\alpha_L^R(k)$, $\Delta\alpha_R^R(k)$, $\Delta\alpha_L^D(k)$, $\Delta\alpha_L^D(k)$, $\Delta\beta_L^{RL}(k)$, $\Delta\beta_R^{RL}(k)$, $\Delta\beta_L^{RS}(k)$, $\Delta\beta_R^{RS}(k)$, $\Delta\beta_L^D(k)$, and $\Delta\beta_R^D(k)$, to be $0$ when $k > \kmax$ for some sufficiently large integer $\kmax$.
Doing so does not violate the regularity constraints, including the monotonicity constraints in Equations~\eqref{con:hybrid-delta-beta-rl-decreasing}, \eqref{con:hybrid-delta-beta-rs-decreasing}, and \eqref{con:hybrid-delta-beta-d-decreasing}, the superiority of randomized rounds in Eqn.~\eqref{con:hybrid-random-vs-deterministic}, and the nonnegativity in Equations~\eqref{con:hybrid-alpha-non-negative} and \eqref{con:hybrid-beta-non-negative}, so long as nonzero parameters for $1 \le k \le \kmax$ satisfy them.
Further, the same applies to the constraints for simplifying the lower bounds of $\alpha$ variables in Equations~\eqref{con:hybrid-random-vs-determine-alpha-1} and \eqref{con:hybrid-random-vs-determine-alpha-2}.
Further, the simplifying constraint in Eqn.~\eqref{con:hybrid-simlfy-to-full}, which allows us to prove  approximate dual feasibility by pairing points between left and right, does not consider large $k$ and thus, is unaffected.
Finally, the approximate feasibility constraints in Table~\ref{table:final-approximate-feasibility-constraint} for $k > \kmax$ simplifies to a single boundary constraint as follows:
\begin{equation}
    \label{con:alpha-boundary}
    \tag{C10}
    \sum_{k=1}^{k_{\max}} \big( \Delta\alpha_L^R(k) + \Delta\alpha_R^R(k) \big) \geq 2\Gamma
    ~.
\end{equation}

The finite LP has $O(k_{\max})$ constraints and is formulated below:
\begin{align*}
    \label{eqn:matching-lp-primal}
    \text{maximize} \quad & \Gamma \\
    \text{subject to} \quad 
    &\text{Regularity Constraints~\eqref{con:hybrid-delta-beta-rl-decreasing}~\eqref{con:hybrid-delta-beta-rs-decreasing}~\eqref{con:hybrid-delta-beta-d-decreasing}~\eqref{con:hybrid-random-vs-deterministic}\eqref{con:hybrid-alpha-non-negative}~\eqref{con:hybrid-beta-non-negative}}&\quad & 1 \le k \le \kmax \\
    &\text{Simplifying Constraints~\eqref{con:hybrid-random-vs-determine-alpha-1}~\eqref{con:hybrid-random-vs-determine-alpha-2}}~ &\quad & 1 \le k \le \kmax \\ 
    & \text{Simplifying Constraint \eqref{con:hybrid-simlfy-to-full}} \\
    &\text{Approximate Dual Feasibility Constraints in Table~\ref{table:final-approximate-feasibility-constraint}}  &\quad & 0 \text{ or } 1 \le k \le \kmax \\
    &\text{Boundary Constraint~\eqref{con:alpha-boundary}}\\
\end{align*}

Finally, we set $k_{\max}=20$ and solve the LP using the PuLP package in Python.%
\footnote{Our code is at available at \url{http://www.zyhwtc.com:8080/file/hybrid.py}}
This gives a set of parameters with $\Gamma > 0.5016$.
In other words, we have found a set of parameters which ensure approximate dual feasibility w.r.t.\ the desired competitive ratio.
Recall that we always have $\alpha_a \geq 0$ and $\beta_i \geq 0$ and reverse weak duality.
This completes the proof of Theorem~\ref{thm:main}.
 


\appendix

\section{Small Bids}
\label{app:small-bid}

This section presents the analysis of the algorithm by \citet{MehtaSVV/JACM/2007} for small bids, i.e., $b_{ai} \le \frac{1}{2} B_a$ for any advertiser $a \in A$ and any impression $i \in I$.
In this case, we consider a deterministic algorithm that can achieve a $\frac{5}{9}$ competitive ratio, which is part of results by \citet{MehtaSVV/JACM/2007}.
We restate and analyze it using online primal dual framework and the configuration LP as demonstrated in Section~\ref{sec:prelim}.
This may serve as a warmup for readers who are not familiar with the framework.
Further, this is combined with the approach dealing with the large bids case in Section~\ref{sec:pd-algorithm} to obtain a hybrid algorithm in Section~\ref{sec:hybrid}, and to get the competitive ratio as stated in Theorem~\ref{thm:main}.

\subsection{Online Primal Dual Algorithm}
\label{sec:pd-algorithm-small-bids}

Algorithm~\ref{alg:pd-algorithm-small-bids} is driven by maximizing the dual variable $\beta_i$ for each impression $i$.
For each advertiser $a$, maintain the following invariant based on the subset of impressions that are already assigned to $a$, denoted as $S_a$:
\begin{equation}
    \label{eqn:small-bid-alpha-invariant}
    \alpha_a \defeq B_a \cdot \alpha \bigg( \frac{b_a(S_a)}{B_a} \bigg)
    ~,
\end{equation}
where $\alpha : [0, 1] \mapsto [0, 1]$ is a function to be optimized in the analysis.
The online primal dual algorithm and analysis shall impose several conditions on $\alpha$, which we shall explain shortly in the next subsection.

The computation of $\beta_i$ is based on the online primal dual framework in Lemma~\ref{lem:online-primal-dual}.
First, let the primal objective equal the dual objective, i.e., $P=D$.
In fact, we make the increments of primal and dual objectives equal in each round of assignments.
That is, if an impression $i$ is assigned to an advertiser $a$, the assigned subset of impressions to advertiser $a$ changes from $S_a$ to $S_a \cup \{i\}$:
\[
\Delta D = \Delta P = b_a(S_a \cup \{i\}) - b_a(S_a)
~. 
\]
Second, we divide the increment of the dual objective into two parts, the increment of $\alpha_a$ and the value of $\beta_i$. By Eqn.~\eqref{eqn:small-bid-alpha-invariant}, the former equals:
\begin{equation*}
    \label{eqn:small-bid-alpha-update}
    \Delta \alpha_a = B_a \bigg( \alpha \bigg( \frac{b_a(S_a \cup \{i\})}{B_a} \bigg) - \alpha \bigg( \frac{b_a(S_a)}{B_a} \bigg) \bigg)
    ~.
\end{equation*}
For convenience of notations, for any $y \in [0,1]$, define:
\begin{equation}
    \label{eqn:small-bids-beta-definition}
    \beta(y) \defeq y - \alpha(y)
    ~.
\end{equation}

Thus, define:
\begin{equation}
\label{eqn:small-bid-beta-update}
        \beta_i  \defeq \Delta D - \Delta \alpha_a 
         = B_a \bigg( \beta \bigg( \frac{b_a(S_a \cup \{i\})}{B_a} \bigg) - \beta \bigg( \frac{b_a(S_a)}{B_a} \bigg) \bigg)
        ~.
\end{equation}

Algorithm~\ref{alg:pd-algorithm-small-bids} assigns each impression $i$ a neighboring advertiser to maximize the value of $\beta_i$.

\begin{algorithm}[t]
    \caption{A Deterministic Online Primal Dual Algorithm by \citet{MehtaSVV/JACM/2007} (Parameterized by a function $\alpha : [0, 1] \mapsto [0, 1]$)}
    \label{alg:pd-algorithm-small-bids}
    \begin{algorithmic}
        \smallskip
        \STATE \textbf{state variables:}
            $S_a$, the subset of impressions that are assigned to $a$ already\\[1ex]
        \FORALL{impression $i$}
            \FORALL{advertiser $a \in A$}
                \STATE compute $\beta_i$ according to Eqn.~\eqref{eqn:small-bid-beta-update}
            \ENDFOR
            \STATE find $a^*$ that maximize $\beta_i$, and assign $i$ to $a^*$
        \ENDFOR
    \end{algorithmic}
\end{algorithm}

\subsection{Online Primal Dual Analysis}
\label{sec:pd-analysis-small-bids}

This subsection presents the online primal dual analysis using the framework in Lemma~\ref{lem:online-primal-dual}, and proves the following theorem.
%
\begin{theorem}
    \label{thm:small-bids-main}
    Algorithm~\ref{alg:pd-algorithm-small-bids} is $\frac{5}{9}$-competitive for AdWords under small bids assumption.
\end{theorem}

Recall that the primal and dual assignments ensure that the increments of primal and dual objectives are equal in every step by definition. 
Next, we derive a set of conditions on the function $\alpha$ which imply the approximate dual feasibility.
Finally, we optimize $\alpha$ by solving a set of inequalities derived by the feasibility analysis.
Recall the approximate dual feasibility in Lemma~\ref{lem:online-primal-dual}.
For any advertiser $a$ and any subset of impressions $S \subseteq I$, we need:
\begin{equation}
    \label{eqn:small-bid-approx-dual-feasible}
    \alpha_a + \sum_{i \in S} \beta_i \ge \Gamma \cdot b_a(S)
    ~.
\end{equation}

\paragraph{Conditions on $\alpha$.}
We first give some conditions for $\alpha$, which simplify the subsequent analysis.
%
\begin{enumerate}
    \item \textbf{Initial values.}
        The above primal and dual assignments ensure equal increments in the primal and dual objectives in every step.
        In order to get equal primal and dual objectives, we further need them to have value $0$ initially.
        It follows from its definition that the primal objective equals $0$ at the beginning. Thus we need:
        \begin{equation}
        \label{eqn:small-bids-initial-value}
            \alpha(0) = 0, \quad \beta(0) = 0
            ~.
        \end{equation}
    \item \textbf{Convexity of $\alpha$ and Concavity of $\beta$.}
        We shall choose $\alpha$ such that the portion of the primal increment that is assigned to $\alpha_a$, i.e., the ratio of the increment of $\alpha_a$ to the primal increment, is nondecreasing in the value of $\alpha_a$.
        This further implies the concavity of $\beta$ by definition.
        We need:
        \begin{equation}
            \label{eqn:small-bids-convexity}
            \forall y \in [0,1] \qquad \alpha''(y) \ge 0, \quad \beta''(y) \le 0
            ~.
        \end{equation}
        This condition is driven by the online primal dual analysis of approximate dual feasibility, i.e., Eqn.~\eqref{eqn:small-bid-approx-dual-feasible}.
    In particular, the crux case is when none of $i \in S$ is assigned to $a$ and thus, the value of $\beta_i$ is lower bounded by what advertiser $a$ offers in Eqn.~\eqref{eqn:small-bid-beta-update}.
    When $\alpha$ is smaller, we need to offer a larger portion of the gain to $\beta$ in order to guarantee approximate dual feasibility;
    and vice versa.
    \item \textbf{Curvature of $\alpha$.} We restrict the curvature of $\alpha$ with upper and lower bounds on its derivative, which implies bounds on the derivative of $\beta$ by definition:
    \begin{equation}
       \forall y \in [0,1] \qquad 1 - \Gamma \le \alpha'(y) \le 1, \quad 0 \le \beta'(y) \le \Gamma
       ~.
       \label{eqn:small-bids-curvature}
    \end{equation}
    The upper bound on $\alpha'$ and the lower bound on $\beta'$ ensures that the assignment of $\beta_i$'s in Eqn.~\eqref{eqn:small-bid-beta-update} satisfies nonnegativity.
    The lower bound on $\alpha'$ and the upper bound on $\beta'$ is driven by the observation that offering a $\Gamma$ portion of the gain of an edge $(a, i)$ to $\beta_i$ is sufficient for covering the contribution of the edge to the RHS of Eqn.~\eqref{eqn:small-bid-approx-dual-feasible}.
    
\end{enumerate}

\paragraph{Contribution from $\beta_i$.}
We next show the approximate dual feasibility, i.e., Eqn.~\eqref{eqn:small-bid-approx-dual-feasible}, by characterizing the contribution from $i \in S$ for any $S \subseteq I$.
Let $S_a$ be the set of impressions assigned to advertiser $a$. 
\begin{lemma}
\label{lem:small-bids-beta-bound}
For any impression $i \in S$,
\[
    \beta_i \ge B_a \bigg( \beta \bigg( \frac{b_a(S_a \cup \{i\})}{B_a} \bigg) - \beta \bigg( \frac{b_a(S_a)}{B_a} \bigg) \bigg)
    ~.
\]
\end{lemma}
\begin{proof}
Any impression $i \in S$ could have gotten a share equal to the above, except that $S_a$ might be a smaller subset at the time when $i$ arrives;
the RHS above is therefore a valid lower bound by the concavity of $\beta$ and the definition of the algorithm.
\end{proof}

\begin{proof}[Proof of Theorem~\ref{thm:small-bids-main}.]
    Combining Lemma~\ref{lem:small-bids-beta-bound} with the definition of $\alpha_a$, it remains to prove that:
    \[
        B_a \alpha \bigg( \frac{b_a(S_a)}{B_a} \bigg) + \sum_{i \in S} B_a \bigg( \beta \bigg( \frac{b_a(S_a \cup \{i\})}{B_a} \bigg) - \beta \bigg( \frac{b_a(S_a)}{B_a} \bigg) \bigg) \ge \Gamma \cdot b_a(S)
        ~.
    \]

    Next, we simplify the above inequality using the sufficient conditions in Equations~\eqref{eqn:small-bids-convexity} and \eqref{eqn:small-bids-curvature}.
    First, dividing both sides by $B_a$, it becomes clear that only the ratios of the bids $b_{ai}$'s to the budget $B_a$ matter.
    Hence, we may wlog normalize $B_a = 1$ to simplify notations.
    The above inequality turns to:
    \[
        \alpha \big( b_a(S_a) \big) + \sum_{i \in S} \big( \beta \big( b_a(S_a \cup \{i\}) \big) - \beta \big( b_a(S_a) \big) \big) \ge \Gamma \cdot b_a(S)
        ~,
    \]
    where the definition of $b_a(\cdot)$ becomes:
    \[
        \forall S \subseteq I: \quad b_a(S) \defeq \min \bigg\{ 1, \sum_{i \in S} b_{ai} \bigg\}
        ~.
    \]

    Second, we claim that it suffices to consider the case when $\sum_{i \in S} b_{ai} = B_a = 1$.
    If $\sum_{i \in S} b_{ai}$ is strictly larger than $B_a$, we may decrease some $b_{ai}$:
    the LHS weakly decreases while the RHS remains the same.
    If $\sum_{i \in S} b_{ai}$ is strictly smaller than $B_a$, on the other hand, we may increase some $b_{ai}$ so that the LHS increases at rate at most $\Gamma$ (by the definition of $b_a$ and Eqn.~\eqref{eqn:small-bids-convexity}), the the RHS increases at rate exactly $\Gamma$.

    Finally, recall that the small bids assumption ensures any $b_{ai} \le \frac{1}{2} B_a$. 
    Thus by the concavity of $\beta$ (Eqn.~\eqref{eqn:small-bids-curvature}) and $b_a(\cdot)$, it suffices to consider $|S| = 2$ with $b_{ai} = \frac{1}{2}$ for both $i \in S$.
    
    Therefore, for any $b = b_a(S_a)$, we need:
    \[
        \alpha ( b ) + 2 \cdot \big( \beta ( \min \{ b + \tfrac{1}{2}, 1 \} ) - \beta (b) \big) \ge \Gamma
        ~.
    \]

    Combining with the definition of $\beta$ in Eqn.~\eqref{eqn:small-bids-beta-definition}, and writing the cases of $0 \le \beta \le \frac{1}{2}$ and $\frac{1}{2} < \beta \le 1$ separately, we get that:
    \begin{equation*}
        \label{eqn:small-bid-approx-dual-feasible-simplified}
        \begin{aligned}
            0 \le b \le \tfrac{1}{2}: \quad & 
            3 \alpha ( b ) - 2 \alpha ( b + \tfrac{1}{2} ) + 1 \ge \Gamma ~; \\
            \tfrac{1}{2} < b \le 1: \quad & 
            3 \alpha ( b ) - 2 \alpha ( 1) + 2 (1 - b) \ge \Gamma ~.
        \end{aligned}
    \end{equation*}
    
    Solving it with the boundary condition $\alpha(0) = 0$ (Eqn.~\eqref{eqn:small-bids-initial-value}) gives $\Gamma = \frac{5}{9}$ with:
    \begin{equation*}
        \alpha(y) = \begin{cases}
            \frac{4}{9} y & 0 \le y \le \frac{1}{2} ~; \\
            \frac{2}{3} y - \frac{1}{9} & \frac{1}{2} < y \le 1 ~.
        \end{cases}
    \end{equation*}
\end{proof}

\bibliographystyle{plainnat}
\bibliography{matching}

\end{document}